\documentclass[lettersize,journal]{IEEEtran}

\usepackage{amsmath,amsfonts,amsthm,amssymb,bm}
\usepackage{bbm}
\usepackage{algorithmic}
\usepackage{algorithm}
\usepackage{array}
\usepackage[caption=false,font=normalsize,labelfont=sf,textfont=sf]{subfig}
\usepackage{textcomp}
\usepackage{stfloats}
\usepackage{url}
\usepackage{verbatim}
\usepackage{graphicx}
\usepackage{cite}
\usepackage[mathcal]{euscript}
\usepackage{placeins}
\usepackage[hidelinks]{hyperref}
\def\f{\bm{f}}
\def\h{\bm{h}}
\def\c{\bm{c}}
\def\K{\mathcal{K}}
\def\D{\mathcal{D}}
\def\H{\mathcal{H}}
\def\F{\mathcal{F}}
\def\E{\mathbb{E}}
\def\P{\mathbb{P}}
\def\C{\mathcal{C}}
\def\M{\mathcal{M}}
\def\R{\bm{R}}

\def\bgamma{\bm{\gamma}}
\def\blambda{\bm{\lambda}}
\def\emp{{\sf emp}}
\def\opt{{\sf opt}}
\def\ave{{\sf ave}}
\def\bmu{\bm{\mu}}

\DeclareMathOperator*{\sign}{sgn}
\DeclareMathAlphabet{\mathdutchcal}{U}{dutchcal}{m}{n}

\newcommand{\Ent}{\operatorname{Ent}}

\newcommand{\Ttrain}{\mathcal{T}}

\newtheorem{theorem}{Theorem}

\newtheorem{proposition}{Proposition}

\newtheorem{assump}{Assumption}

\allowdisplaybreaks

\newcommand{\givensmall}{\mkern1mu|\mkern1mu}
\newcommand{\given}{\mkern2mu|\mkern2mu}

\usepackage{setspace}


\begin{document}

\title{Test-Time Collaborative Classification over Multi-Agent Networks}

\author{Ping~Hu,~\IEEEmembership{Member,~IEEE,}
        Mert~Kayaalp,~\IEEEmembership{Member,~IEEE} and~Ali~H.~Sayed,~\IEEEmembership{Fellow,~IEEE}
\thanks{Ping Hu and Ali H. Sayed are with the School of Engineering, \'Ecole Polytechnique F\'ed\'erale de Lausanne (EPFL), 1015 Lausanne, Switzerland (email: ping.hu@epfl.ch; ali.sayed@epfl.ch). Mert Kayaalp is now with the DTI, SUPSI, Dalle Molle Institute for Artificial Intelligence (IDSIA USI-SUPSI), 6962 Lugano, Switzerland, where he is affiliated with the UBS-IDSIA AI Lab (email: mert.kayaalp@idsia.ch). This work was done while he was a Ph.D. student and Post-doc at the Adaptive Systems Laboratory at EPFL. Mert Kayaalp's work was supported by UBS Switzerland AG and its affiliates through the UBS-IDSIA AI Lab. A preliminary version of this work was presented at the Algorithmic Collective Action Workshop, NeurIPS 2025.}
}

\maketitle

\begin{abstract}
The increasing heterogeneity of multi-agent systems poses significant challenges for jointly training a global model across agents. At the same time, cooperative inference between agents has long been recognized as a powerful mechanism for distributed decision making over networks. Motivated by these observations, we propose a collaboration framework for distributed binary classification over multi-agent networks, where a set of independently trained agents, potentially differing in architecture, feature space, or modality, coordinate their actions during test time to form collective predictions. This coordination is achieved by exchanging local decision statistics through a distributed learning protocol. We develop a theoretical and experimental study of this independent training and cooperative inference paradigm, and examine its performance under different communication budgets and distributed learning rules. We establish classification error guarantees under sufficient, finite-round, and finite-precision communication, together with PAC-style generalization bounds. These results capture the influence of model heterogeneity, network topology, combination policy, and communication constraints on prediction accuracy. Taken together with the experimental results, they reveal both the price of independent training and the benefit of collective prediction for the proposed distributed decision making framework with models learned from data.
\end{abstract}


\section{Introduction}
\label{sec: intro}

\IEEEPARstart{D}{istributed} learning and inference is a key paradigm for large-scale and distributed systems, where data is generated and processed across multiple agents.\footnote{Throughout this paper, we use the term ``distributed'' to describe systems without a central coordinator. Closely related architectures are also termed ``decentralized'' in the machine learning literature.} Within this context, two major lines of research have emerged: \textit{distributed inference}, extensively studied in the statistics, signal processing, and control communities~\cite{sayed2013diffusion,kar2013consensus+,jordan2019communication}, and \textit{distributed machine learning}, primarily advanced within the machine learning (ML) and data science communities~\cite{verbraeken2020survey,jin2024collaborative}.

In distributed inference, a collection of agents is connected through a communication network. Each agent owns a family of statistical models that describes the distribution of its observations under a common, shared latent variable (e.g., a hidden state or hypothesis). The collective objective is to infer this latent variable through local computations and communication. To this end, agents iteratively update their beliefs or estimates by combining private observations with information received from their neighbors. Classical instances of this paradigm include distributed estimation~\cite{kar2013consensus+,sayed2013diffusion,jordan2019communication}, distributed detection~\cite{matta2016diffusion,bajovic2016distributed,shahrampour2016distributed}, and social learning~\cite{matta2025social}. A key feature of these problems is that collaboration takes place during the \textit{inference} phase. In contrast, distributed ML has mainly focused on collaboration during the \textit{training} phase, as exemplified by federated learning and decentralized learning. In this setting, agents jointly train a global model on distributed data, which is later deployed for inference either centrally or across devices. The inference stage is thus governed by a single unified model, and the system operates as a \textit{single} decision maker. This is fundamentally different from distributed inference, where each agent acts as an autonomous decision maker, and collaboration serves to improve the quality of individual or consensus predictions.

While these two research lines have evolved largely independently, they share some core motivation (e.g., scalability, decentralization, and privacy), and often employ similar tools, particularly those from distributed optimization~\cite{sayed2014adaptation,nedic2020distributed}. Recently, efforts have been made to bridge the two fields in order to leverage their complementary benefits. A promising line of work is inference with machine-learned models, as studied in~\cite{braca2022statistical,bordignon2023learning,hu2025non-asymptotic}. In~\cite{braca2022statistical}, this idea is examined in the context of hypothesis testing within a centralized setting involving a single agent. More relevantly,~\cite{bordignon2023learning} and~\cite{hu2025non-asymptotic} study the social learning task where multiple agents train local classifiers on their own data to extract discriminative information for inference in the absence of explicit statistical models. In both cases, inference is performed using a stream of unlabeled observations. In this work, we consider instead a distributed binary classification setting in which each agent has access to only a \emph{single} local testing sample associated with a common, unknown class label. The agents are collectively tasked with identifying this label from their own observations. This setting naturally arises in applications such as sensor fusion (e.g., vehicles capturing different views of the same object) and multi-view classification (e.g., image and text modalities of the same entity). The goal of this work is to develop a distributed classification framework for this setting.

\subsection{Our work}

To this end, we propose a collective prediction framework for distributed classification over multi-agent networks, which combines independent local training with collaborative inference. Specifically, during the \emph{training} phase, each agent trains a local classifier using its own dataset. This independent training process is motivated by several trends and challenges in modern ML systems. First, it has become increasingly common in large-scale and privacy-sensitive applications, such as edge computing~\cite{yilmaz2025private} and foundation models~\cite{yadav2024survey}. For example, large language models or domain-specific experts may be trained independently across different entities (e.g., companies, institutions, clients, or devices) and later coordinated at inference time~\cite{mavromatis2024pack}. Second, feature and model heterogeneity often make joint training impractical. For example, agents may process different input modalities or rely on distinct model architectures tailored to their local data. During the \emph{testing} phase, each agent receives a private testing sample and participates in a distributed inference protocol to infer the shared class label by exchanging predictive information with its neighbors. This collaborative framework introduces several challenges that are not present in the standard supervised classification setting:
\begin{itemize}
    \item[(i)] \textit{How should agents aggregate predictions at test time?} Unlike standard ML pipelines that rely on a single trained model, each agent in our setting is equipped with an independently trained local model. This raises the question of how agents can effectively aggregate their individual predictions during inference.
    \item[(ii)] \textit{How does agent heterogeneity govern the generalization performance?} While classical supervised generalization theory typically analyzes a single learned predictor, our setting involves collective predictions based on multiple independently trained local classifiers. This raises the question of how agents' inherent heterogeneity and inter-agent interactions shape the generalization behavior of the entire distributed system.
    \item[(iii)] \textit{How does limited communication affect prediction?} Collective prediction typically relies on sufficient information exchange. When communication is limited in either the number of rounds or message precision, the resulting predictions may deviate from the sufficient-communication limit. This raises the question of how communication constraints affect the individual prediction accuracy.
\end{itemize}

To introduce our framework, we assume that agents produce soft predictions using their local classifiers and adopt the classical DeGroot model~\cite{degroot1974reaching} as an interpretable and analytically tractable mechanism for information aggregation during inference. In this model, agents iteratively update their predictions by taking a weighted average of their neighbors’ predictions according to a \emph{combination policy} defined over a fixed communication topology~\cite{degroot1974reaching}. The DeGroot model has been extensively studied in the distributed inference literature and is well known for achieving consensus in multi-agent systems~\cite{olfati2007consensus,acemoglu2011opinion}. DeGroot updating has also been examined experimentally in human social networks~\cite{chandrasekhar2020testing}. Moreover, its linear update structure makes it particularly well-suited for distributed implementation with finite communication rounds.

The novelty of the present work is not in proposing a new information aggregation rule, but in formulating and analyzing this inference-time collaboration framework. Accordingly, this work develops a theoretical and experimental study with the following contributions:

(i) We establish theoretical guarantees for distributed classification over multi-agent networks under the independent training and collaborative inference paradigm. The derived bounds characterize how network structure, local model quality, and data heterogeneity influence predictive accuracy.

(ii) We analyze how communication constraints affect classification, deriving agent-dependent bounds for both finite-round and finite-precision communication. These results characterize how incomplete consensus, network topology, combination policy, and communication precision influence prediction accuracy.

(iii) We evaluate the proposed framework on two complementary benchmarks and empirically demonstrate the benefits of cooperative inference under different observation settings.

These contributions provide theoretical foundations for test-time collaboration in distributed classification and highlight connections between distributed inference and distributed ML.

\subsection{Related Literature}

\subsubsection{Distributed machine learning}
Most existing work in this domain has focused on collaboratively training a shared global model using data distributed across multiple devices or computational nodes. Two dominant approaches are federated learning, which relies on a central server to aggregate locally trained updates~\cite{mcmahan2017communication,konevcny2016federated}, and decentralized learning, which dispenses with the server and synchronizes models through peer-to-peer communication over a fixed network graph~\cite{lian2017can,jiang2017collaborative}. Variants of these methods address challenges such as non-i.i.d.\ data~\cite{zhao2018federated,lin2020ensemble}, agent sampling~\cite{kairouz2021advances}, personalization~\cite{fallah2020personalized}, and communication efficiency during training~\cite{wen2022federated,xie2024convergence}. Despite their topological differences, in their canonical forms, both paradigms aim to produce a shared global model during training. At inference time, this model is either executed on a central server or deployed across devices, so the system effectively operates as a \emph{centralized} decision maker, even when implementation is distributed. These joint-training frameworks are therefore complementary to the setting studied here: they aim to produce a shared global model during training, whereas our work focuses on inference-time collaboration among independently trained local models. A small body of work has explored collaboration during inference, where a set of independently trained models exchange information with \emph{each other} to improve accuracy via trust-score design~\cite{mendler2021test,alzubi2018consensus}. Yet they all assume that all agents observe the \emph{same} test input, which restricts agents' heterogeneity to model parameters or inductive biases. This setup closely parallels ensemble learning and multiple classifier combination discussed later.

In contrast, we study a setting in which each agent receives its own private testing sample, and these samples may be statistically dependent across agents. This formulation naturally accommodates heterogeneous feature spaces (e.g., multimodal sensors or domain-specific models) and enables distributed decision-making without feature alignment. Collaborative inference in this setting has received little attention, and our work provides a formal theoretical and experimental treatment.

\subsubsection{Distributed inference}

Distributed inference studies how a network of agents collaborates to estimate a latent variable based on private observations and local communication~\cite{kar2013consensus+,sayed2013diffusion,jordan2019communication}. Applications include state estimation in smart power grids, cooperative perception in multi-robot systems, and opinion dynamics in social networks. These applications motivate different modeling frameworks depending on the nature of the latent variable. For discrete hypothesis spaces, a widely studied paradigm is social learning (SL), in which agents iteratively exchange local beliefs to identify the true hypothesis. SL naturally accommodates heterogeneous agents, since each agent may observe different types of signals modeled by distinct statistical distributions~\cite{matta2025social}. Classical SL methods assume that agents possess a family of likelihood models, which are used in updating beliefs via Bayes’ rule (or suitable variants)~\cite{jadbabaie2012nonbayesian,lalitha2018social,nedic2017fast,bordignon2021adaptive,kayaalp2022arithmetic}. While analytically elegant, this assumption is restrictive in real-world problems, where data distributions are complex and rarely admit closed-form models. To address this limitation,~\cite{bordignon2023learning} proposed a social machine learning strategy, where each agent first trains a local classifier using its own labeled examples, and then performs the SL rule based on the trained classifier. The asymptotic and non-asymptotic performance of this strategy has been examined in~\cite{bordignon2023learning,hu2025non-asymptotic}.

Our work follows this data-driven perspective but differs in the task formulation. In SL, agents typically process streams of unlabeled i.i.d.\ data to infer a shared latent hypothesis. In the supervised classification setting studied here, each agent receives only a \emph{single} testing sample for collective prediction, and collaboration takes place through a finite-round exchange of learned local decision statistics at test time. This difference changes the operational role of communication and leads to analytical questions that are not directly addressed in existing SL studies or in classical single-agent supervised classification.

\subsubsection{Multiple classifier combination}

This is a classical line of research in ML that studies how to aggregate the predictions of multiple models to improve performance. While our framework is distributed and collaborative in nature, it bears a strong conceptual connection to this literature. Specifically, when all agents follow the DeGroot model for information aggregation at test time, their predictions converge (in the limit of infinite communication) to a consensus that is a \emph{weighted average} of the local classifier outputs. This resembles the effect of a centralized fusion rule, a central topic in the literature on multiple classifier combination. Classical techniques in this area include voting schemes, bagging, boosting, and stacking~\cite{polikar2006ensemble, kuncheva2014combining}. These approaches form the foundation of what is often referred to as \emph{ensemble learning}, where multiple base classifiers are combined to improve generalization and robustness~\cite{bartlett1998boosting,zhouensemble2012}. More advanced strategies, such as mixture-of-experts models, combine classifiers through input-dependent gating~\cite{yadav2024survey}. These approaches, however, typically operate in a \emph{centralized} setting during test time and assume a common feature representation across experts.

In contrast, our framework operates in a fully-distributed environment during both training and testing phases. When communication is limited to a finite number of rounds, consensus may not be reached, and predictions remain agent-dependent. This dynamic, iterative process differs fundamentally from the static, one-shot aggregation used in most ensemble methods.

\textbf{Notation:} We use boldface font to denote random variables and normal font for their realizations, e.g., $\h$ and $h$. $\E$ and $\P$ denote the expectation and probability operators, respectively. $\mathsf{1}[\cdot]$ denotes the indicator function, and $\mathbbm{1}$ denotes the all-ones vector.


\section{Problem Formulation}
\label{sec: SML strategy}

We are given a network of $K$ agents indexed by $k$ and a binary classification task with class label $\gamma$. The sets of agents and labels are denoted by $\K\triangleq\{1,2,\dots,K\}$ and $\Gamma\triangleq\{+1,-1\}$, respectively. We assume that the semantic meaning of the labels is fixed and shared across the network during both training and inference. Each agent $k$ holds a local training set $\D_k$ of $N_k$ labeled examples $(\h_{k,n},\bgamma_{k,n})$, where $\h_{k,n}$ is the $n$-th feature vector and $\bgamma_{k,n}\in\Gamma$ is its label. Let $\H_k$ denote the feature space of agent $k$. A network observation consists of one local view from each agent and therefore takes values in the product space
\begin{equation}\label{eq: network feature space}
    \H\triangleq\H_1\times\cdots\times\H_K.
\end{equation}
For each $\gamma\in\Gamma$, let $P_\gamma$ denote the class-conditional distribution of a network observation $\h=(\h_1,\ldots,\h_K)\in\H$, and let $P_{\gamma,k}$ denote its $k$-th marginal. When this marginal admits a density or probability mass function, we denote it by $p_k(h|\gamma)$.

For the binary classification task, we assume a uniform class prior. Then, $P_\gamma$ induces the joint distribution $Q$ of $(\h,\bgamma)$:
\begin{equation}\label{eq: network population law}
    \P_Q(\bgamma=\gamma)=\frac12,  \quad \h\mid\{\bgamma=\gamma\}\sim P_\gamma,
    \quad \gamma\in\Gamma.
\end{equation}
We denote by $Q_k$ the corresponding marginal distribution of $(\h_k,\bgamma)$ at agent $k$. Conditional on $\bgamma$, the local observations in $\h$ may be statistically dependent across agents. The agents are \emph{heterogeneous} as they may have different feature spaces $\H_k$ (e.g., different views in multi-view learning) or distinct class-conditional models $p_k(h|\gamma)$ (e.g., different local distributions in personalized federated learning). The posterior log-ratio is sufficient for binary decision making. Under the uniform prior in~\eqref{eq: network population law}, Bayes' rule gives
\begin{equation}\label{eq: logit function}
    \Lambda_k(h)\triangleq \log\frac{p_k(+1|h)}{p_k(-1|h)} =
    \log\frac{p_k(h|\!+1)}{p_k(h|\!-1)}, \quad \forall h\in\H_k
\end{equation}
where $p_k(\gamma|h)$ denotes the posterior probability. The function $\Lambda_k$ is also known as the logit function in the literature.

\subsection{The Training Phase}
\label{sec: the training phase}

In the training phase, each agent $k$ uses its own dataset $\D_k$ to train a local classifier. For each agent $k$, the samples in $\D_k$ are independently drawn according to $Q_k$. Moreover, samples with the same index $n$ may be statistically dependent across agents, while the corresponding cross-agent sample collections are independent across indices.  We denote by $f_k:\H_k\to\mathbb{R}$ a generic real-valued score function from an admissible function class $\F_k$. The function $f_k$ serves as an approximation to the logit function $\Lambda_k$ in~\eqref{eq: logit function}. Formally,
\begin{equation}
    f_k(h) = \log\frac{\widehat{p}_k(+1|h)}{\widehat{p}_k(-1|h)}
\end{equation}
where $\widehat{p}_k(\gamma|h)$ denotes the estimated posterior probability. Given $\D_k$, the empirical risk associated with a candidate score function $f_k\in\F_k$ is
\begin{equation}\label{eq: empirical risk}
    \R_{k,\emp}(f_k)\triangleq \frac{1}{N_k}\sum_{n=1}^{N_k}\Phi(\bgamma_{k,n}f_k(\h_{k,n}))
\end{equation}
where $\Phi$ denotes the loss function used for training. Let $\mathcal A_k$ denote the local training algorithm at agent $k$ and let $\mathcal{S}_k$ collect its internal randomness. The resulting trained score function is denoted by
\begin{equation}\label{eq: actual trained f_k}
    \widehat{\f}_k \triangleq \mathcal A_k(\D_k,\mathcal{S}_k) \in\F_k.
\end{equation}
We measure the optimization accuracy of $\widehat{\f}_k$ by its empirical optimality gap, namely, the difference between its empirical risk and the infimum empirical risk over $\F_k$:
\begin{equation}\label{eq: local optimization error}
    \mathrm{Gap}_{k,\mathrm{opt}}(\widehat{\f}_k)\triangleq\R_{k,\emp}(\widehat{\f}_k)-\inf_{f_k\in\F_k}\R_{k,\emp}(f_k)\geq 0.
\end{equation}
A zero gap corresponds to exact empirical risk minimization (ERM), while a positive gap quantifies the residual optimization error of the trained score function. 

For the theoretical analysis, we use the following assumptions on the loss function $\Phi$ and the function class $\F_k$, which are commonly adopted in statistical learning theory~\cite{bartlett2006convexity,mohri2018foundations}.
\begin{assump}[\textbf{Conditions on the loss function}]\label{assump: risk function}
The loss function $\Phi:\mathbb{R}\to\mathbb{R}_{+}$ is convex, non-increasing and differentiable at $0$ with $\Phi^\prime(0)<0$. Also, it is $L_{\Phi}$-Lipschitz.  \qed
\end{assump}
\begin{assump}[\textbf{Boundedness of functions}]\label{assump: bound}
        There exists a constant $\beta>0$ such that $|f_k(h)|\leq\beta$ for every $k\in\K$, $f_k\in\F_k$, and $h\in\H_k$.  \qed
\end{assump}

\noindent
The training phase is completed by empirically centering the learned score $\widehat{\f}_k$, following~\cite{bordignon2023learning,hu2025non-asymptotic}. Since independently trained local classifiers may exhibit different empirical offsets, we subtract each score's empirical training mean to better align the local outputs before collaboration. Specifically, we define the \emph{empirical training mean as}
\begin{equation}\label{eq: empirical training mean}
    \bmu_{k,\emp}(f_k)\triangleq\frac{1}{N_k}\sum_{n=1}^{N_k}f_k(\h_{k,n}), \quad \forall f_k\in\F_k.
\end{equation}
The resulting classifier generates a centered score:
\begin{equation}\label{eq: trained classifier}
    \c_k(h)\triangleq\widehat{\f}_k(h) - \bmu_{k,\emp}(\widehat{\f}_k), \quad \forall h\in\H_k.
\end{equation}
Therefore, after training, the sign of $\c_k(h)$ determines the local prediction at agent $k$ for any unseen feature vector $h\in\H_k$. The role of empirical centering is examined in Section~\ref{sec: consistent training}.

\subsection{The Prediction Phase}
\label{sec: the prediction phase}
In the prediction phase, the network receives a fresh testing example $(\h^*,\bgamma^*)$ drawn from $Q$ independently of the training data and local training-algorithm randomness, with each agent $k$ observing only its local component $\h_k^*\in\H_k$. These local views may remain statistically dependent conditioned on $\bgamma^*$. The goal of the agents is to make a \emph{collective} prediction about $\bgamma^*$ by collaborating with their neighbors. We use $t$ to index the communication rounds and denote by $\blambda_{k,t}$ the \emph{decision statistic} of agent $k$ at round $t$, whose sign determines its prediction. Following the DeGroot model, the agents iteratively combine their decision statistics over the network. Upon observing $\h_k^*$, agent $k$ initializes its decision statistic using the local classifier $\c_k$:
\begin{equation}\label{eq: local decision statistic}
    \blambda_{k,0}\triangleq\c_k(\h_k^*).
\end{equation}
Then, each agent communicates with its neighbors and updates its decision statistic by aggregating the local decision statistics in the neighborhood~\cite{degroot1974reaching}:
\begin{equation}\label{eq: distributed learning rule}
    \blambda_{k,t}=\sum_{\ell=1}^{K}a_{\ell k}\blambda_{\ell,t-1}
\end{equation}
where $a_{\ell k}$ denotes the combination weight agent $k$ assigns to its neighbor $\ell$, which satisfies
\begin{equation}\label{eq: combination weight}
    \sum_{\ell=1}^K a_{\ell k}=1, \;\; a_{\ell k}> 0\;\; \forall \ell\in \mathcal{N}_k,\text{ and }a_{\ell k}=0 \;\;\forall \ell\notin\mathcal{N}_k
\end{equation}
with $\mathcal{N}_k$ denoting the neighboring set of agent $k$. Communication proceeds in synchronous rounds over this fixed graph, and the exchanged messages are assumed to be delivered reliably. Accordingly, the communication model considered here does not include asynchronous updates, time-varying connectivity, or packet losses. Moreover, we impose the following assumption on the topology of the communication network to ensure information diffusion throughout the network.
\begin{assump}[\textbf{Strongly-connected graph}]\label{assump: network}
    The graph of the communication network is strongly connected. That is, there exist paths with positive combination weights between any two distinct agents in both directions (the two paths need not be the same), and at least one agent has a self-loop, i.e., $a_{kk}>0$ for some agent $k$~\cite{sayed2014adaptation}.  \qed
\end{assump}
\noindent Under this assumption and from the Perron-Frobenius theorem~\cite{sayed2022inference},  the combination matrix $A=[a_{\ell k}]$ is primitive and admits a Perron vector $\pi$ satisfying
\begin{equation}\label{eq: Perron eigenvector}
    A\pi=\pi,\;\; \sum_{k=1}^{K}\pi_k=1,\text{ and }\pi_k>0, \;\; \forall k\in\K.
\end{equation}
The entries of $\pi$ quantify the relative influence of the agents following the adopted combination policy. In network science terminology~\cite{newman2018networks,lewis2011network}, this Perron vector corresponds to the notion of eigenvector centrality associated with the matrix $A$. Moreover, the second-largest magnitude among all eigenvalues of $A$, denoted by $\sigma_A$, is strictly smaller than 1.

To evaluate the performance of this framework, we first introduce a feasibility condition associated with the classification task, defined in terms of the target risk for training. For each agent $k$, the target risk ${\sf R}_k^o$ is the infimum of the expected risk over the local function class $\F_k$:
\begin{equation}\label{eq: expected target risk for the agent}
    {\sf R}_k^o\triangleq \inf_{f_k\in\F_k}\E_{(\h_k,\bgamma)\sim Q_k}\Phi(\bgamma f_k(\h_k)).
\end{equation}
The target risk for the \emph{network} is defined as a weighted average of the individual target risks involving the Perron vector $\pi$:
\begin{equation}\label{eq: expected target risk}
    {\sf R}^o\triangleq \sum_{k=1}^K\pi_k{\sf R}_k^o.
\end{equation}
The reason for introducing $\pi$ in defining ${\sf R}^o$ will be clear in our subsequent analysis for the classification error. The feasibility of the classification task requires the following assumption.
\begin{assump}[\bf{Feasibility}]\label{assump: feasibility}
    The network target risk satisfies ${\sf R}^o<\Phi(0)$.\qed
\end{assump}
\noindent We note that $\Phi(0)$ represents the expected risk of the model $f_k=0$, which corresponds to the case where agent $k$ assigns labels $+1$ and $-1$ with equal probability to all feature vectors. Hence, ${\sf R}_k^o=\Phi(0)$ implies that the feature vectors of agent $k$ provide no useful information for classification within the prescribed function class $\F_k$. Assumption~\ref{assump: feasibility} is a network-level feasibility condition: it does not require every agent to be individually informative, but only that the network target risk satisfies ${\sf R}^o<\Phi(0)$. Thus, the collective observations available in the network must contain enough class-discriminative information to outperform random guessing. If this condition fails, then the classification problem is effectively uninformative at the network level.

\section{Performance Guarantees}
\label{sec: main results}

This section develops four principal performance guarantees for the proposed framework: sufficient-communication classification, finite-round communication, finite-precision communication, and probably approximately correct (PAC)-style generalization. The analysis of the first three guarantees proceeds in two stages. We first characterize the training performance of the classifier network through an \emph{expected network margin}. Two supporting propositions are derived for clarifying the role of empirical centering in~\eqref{eq: trained classifier} and establishing when approximate local training attains such a margin with high probability. We then use this common margin event to characterize the prediction error under the three communication settings. The final PAC-style theorem provides a complementary guarantee that relates population error directly to the performance observed on the training data. We conclude by clarifying the relation to conventional social learning recursions and summarizing the corresponding results without empirical centering.

\subsection{Network Margin and Training Guarantee}
\label{sec: consistent training}

To define the expected network margin used throughout the analysis, we first identify the aggregate classifier obtained after sufficient communication. Following the DeGroot averaging rule~\eqref{eq: distributed learning rule}, agents agree on a \emph{common} decision statistic denoted by $\blambda_\ave$ after sufficient communication~\cite{degroot1974reaching}. That is, it holds almost surely (a.s.) that
\begin{equation}\label{eq: common decision statistic}
    \blambda_\ave \triangleq \lim\limits_{t\to\infty}\blambda_{k,t}\overset{\text{a.s.}}{=}\sum_{k=1}^K\pi_k\c_k(\h_k^*).
\end{equation}
Accordingly, the \emph{collective} prediction under sufficient communication is
\begin{equation}\label{eq: label of single-sample classification}
    \widehat{\bgamma}\triangleq\sign(\blambda_\ave)
\end{equation}
where any fixed tie rule may be used at zero. Equivalently, the distributed rule~\eqref{eq: distributed learning rule} functions as an aggregate classifier:
\begin{equation}\label{eq: single-sample ensemble classifier}
    \c_\ave(\h^*)\triangleq\sum_{k=1}^K\pi_k\c_k(\h_k^*).
\end{equation}
This aggregate classifier resembles a $\pi$-weighted soft decision combination of $K$ independently trained local classifiers. The analysis of $\c_\ave$ is nontrivial because the informativeness of the local scores may vary across agents and the local views may be statistically dependent. Existing theoretical results on soft decision combination often rely on specific distributional assumptions on the local scores~\cite{kuncheva2002theoretical}. We therefore develop an analysis that accommodates these heterogeneous and statistically dependent local observations.

Since the deployed local scores are empirically centered, we first examine how the empirical centering operation in~\eqref{eq: trained classifier} affects the aggregate score relative to the common zero decision threshold. For a generic collection of local score functions
\begin{equation}\label{eq: f_ave}
    f=(f_1,\ldots,f_K)\in\F\triangleq\F_1\times\cdots\times\F_K,
\end{equation}
we define the uncentered aggregate score by
\begin{equation}\label{eq: raw aggregate classifier}
    f_{\ave}(\h)\triangleq\sum_{k=1}^K\pi_k f_k(\h_k).
\end{equation}
Let $\bmu_{\rm raw}^{\gamma}(f)\triangleq \E^{(\gamma)}f_{\ave}(\h)$ denote its class-conditional mean, where $\E^{(\gamma)}$ denotes expectation under $P_\gamma$. We introduce the following network-level quantities associated with $\bmu_{\rm raw}^{\gamma}(f)$:
\begin{align}
\label{eq: mu_mid}
    \mu_{\rm mid}(f) &\triangleq \frac{\bmu_{\rm raw}^{+}(f)+\bmu_{\rm raw}^{-}(f)}{2},
    \\
    \label{eq: raw midpoint and signed half difference}
    D(f) &\triangleq \frac{\bmu_{\rm raw}^{+}(f)-\bmu_{\rm raw}^{-}(f)}{2}.
\end{align}
Moreover, we define the network empirical training mean by
\begin{equation}\label{eq: network empirical center main}
    \bmu_{\emp}(f)\triangleq\sum_{k=1}^K\pi_k\bmu_{k,\emp}(f_k).
\end{equation}
Under the uniform class prior, for any fixed $f$ independent of the training samples, the training setup in Section~\ref{sec: the training phase} implies $\E[\bmu_{\emp}(f)]=\mu_{\rm mid}(f)$. Hence, $\mu_{\rm mid}(f)$ is the population counterpart of the network empirical training mean. With the above definitions, we establish the following proposition.

\begin{proposition}[\textbf{Effect of additive centering}]
\label{prop: effect of additive centering}
For arbitrary additive centers $\chi_k$, let $\chi_\pi\triangleq\sum_k\pi_k\chi_k$ and define
\begin{equation}
    s_{\chi,f}(\h) \triangleq \sum_{k=1}^K\pi_k\bigl[f_k(\h_k)-\chi_k\bigr].
\end{equation}
Its expected margin over both classes at threshold zero is
\begin{align}
\nonumber
    \Delta_\chi(f) &\triangleq \min\Bigl\{\E^{(+1)}s_{\chi,f}(\h), -\E^{(-1)}s_{\chi,f}(\h)\Bigr\}\\
    \label{eq: margin under additive shift}
    &=D(f)-|\chi_\pi-\mu_{\rm mid}(f)|.
\end{align}
Hence the population midpoint $\chi_\pi=\mu_{\rm mid}(f)$ maximizes this margin. In particular, the margins under raw and empirically centered scores satisfy
\begin{align}
\label{eq: raw margin}
    \Delta_{\rm raw}(f)&=D(f)-|\mu_{\rm mid}(f)|, \\
    \label{eq: raw centered margin comparison}
    \Delta_{\rm cen}(f)&=D(f)-|\bmu_{\emp}(f)-\mu_{\rm mid}(f)|.
\end{align}
Therefore, empirical centering~\eqref{eq: trained classifier} improves the expected zero-threshold margin if and only if
\begin{equation}\label{eq: centering improvement condition}
    |\bmu_{\emp}(f)-\mu_{\rm mid}(f)| < |\mu_{\rm mid}(f)|.
\end{equation}
Moreover, empirical centering is invariant to the agent-specific additive offsets: if $f'_k(h)=f_k(h)+o_k$, then
\begin{equation}\label{eq: empirical centering invariance}
    f'_k(h)-\bmu_{k,\emp}(f'_k) = f_k(h)-\bmu_{k,\emp}(f_k).
\end{equation}
Consequently, the centered DeGroot recursion~\eqref{eq: distributed learning rule} stays unchanged by these offsets at every communication round.
\end{proposition}
\begin{proof}
See Appendix~\ref{appendix: proof additive centering}.
\end{proof}
\noindent For the learned classifiers $\c_k$ from Section~\ref{sec: the training phase}, Proposition~\ref{prop: effect of additive centering} applies with $f_k=\widehat{\f}_k$ and $\chi_k=\bmu_{k,\emp}(\widehat{\f}_k)$. It shows that empirical centering acts as a threshold alignment mechanism that changes the location of the aggregate score relative to the zero decision threshold while leaving the class-conditional mean separation $2D(\widehat{\f})$ unchanged. Its effect on the expected zero-threshold margin depends on how closely the realized empirical training mean $\bmu_{\emp}(\widehat{\f})$ aligns with the corresponding population midpoint $\mu_{\rm mid}(\widehat{\f})$, as characterized by \eqref{eq: centering improvement condition}. Moreover, empirical centering removes agent-specific additive offsets from the local scores and hence from the subsequent collaborative recursion.

We next characterize the expected network margin attained by the centered local classifiers. For any $f_k\in\F_k$, define its empirically centered score
\begin{equation}\label{eq: generic centered classifier}
    c_{k,f_k}(h)\triangleq f_k(h)-\bmu_{k,\emp}(f_k),
\end{equation}
and the corresponding class-conditional means
\begin{equation}\label{eq: expected decision statistic under class +1}
    \bmu_k^+(f_k)\triangleq\E_{\h_k}^{(+1)}c_{k,f_k}(\h_k),
\end{equation}
\begin{equation}\label{eq: expected decision statistic under class -1}
    \bmu_k^-(f_k)\triangleq\E_{\h_k}^{(-1)}c_{k,f_k}(\h_k),
\end{equation}
where $\E_{\h_k}^{(\gamma)}$ denotes expectation under $P_{\gamma,k}$. The corresponding network means are
\begin{equation}\label{eq: expected decision statistic for the network}
    \bmu^+(f)=\sum_{k=1}^K\pi_k\bmu_k^+(f_k),\quad
    \bmu^-(f)=\sum_{k=1}^K\pi_k\bmu_k^-(f_k).
\end{equation}
Before conditioning on the training phase, these quantities are random because the deployed local functions and empirical training means $\bmu_{k,\emp}(f_k)$ depend on the realized training data $\D_k$ and, when applicable, local optimization randomness $\mathcal{S}_k$. The inclusion of the Perron vector $\pi$ in~\eqref{eq: expected decision statistic for the network} is consistent with the collective nature of the prediction phase, where the same weights determine the common decision statistic in~\eqref{eq: common decision statistic}.

For a collection of learned functions $\widehat{\f}=(\widehat{\f}_1,\ldots,\widehat{\f}_K)$, we say that the classifier network satisfies the $\delta$-margin consistent training condition if
\begin{equation}\label{eq: delta-margin consistent training condition}
    \bmu^+(\widehat{\f})>\delta  \quad \text{and} \quad \bmu^-(\widehat{\f})<-\delta
\end{equation}
where $\delta\geq0$ is the prescribed \emph{decision margin}. In view of~\eqref{eq: single-sample ensemble classifier} and \eqref{eq: generic centered classifier}--\eqref{eq: expected decision statistic for the network}, this condition requires that the class-conditional mean of the aggregate score $\c_\ave(\h^*)$ lies on the correct side of the zero decision threshold for both classes, with a margin exceeding $\delta$. The performance of the training phase is then characterized by the probability that condition~\eqref{eq: delta-margin consistent training condition} holds for the classifier network. Define the corresponding event
\begin{equation}\label{eq: event of delta-margin consistent learning}
    \C_\delta \triangleq \Bigl\{ \bmu^+(\widehat{\f})>\delta,\, \bmu^-(\widehat{\f})<-\delta\Bigr\},
\end{equation}
and let
\begin{equation}
    P_{c,\delta}\triangleq\P(\C_\delta),
\end{equation}
where the probability is taken over the randomness from training. Analyzing $P_{c,\delta}$ is non-trivial, as it depends on various factors, including the size of the training sets, the complexity of the function class, the loss function, and the combination policy. A lower bound on $P_{c,\delta}$ under exact ERM, i.e., zero optimization gap in~\eqref{eq: local optimization error}, was established in~\cite{hu2025non-asymptotic} using tools from~\cite{bordignon2023learning}. We next extend it to approximate local optimization. The auxiliary quantities $N_{\max}$, $\alpha$, $\rho$, $\mathdutchcal E_\Phi(r,\delta)$, and $\delta_{\max}(r)$ appearing below are defined in Appendix~\ref{appendix: training auxiliary quantities}.

\begin{proposition}[\textbf{$\delta$-margin consistent training under approximate optimization}]
\label{prop: P_c_delta}
Suppose that, for every $k$, there exists a deterministic tolerance $\varepsilon_{k,\opt}\geq0$ such that
\begin{equation}\label{eq: approximate ERM condition}
    \mathrm{Gap}_{k,\opt}(\widehat{\f}_k) \leq \varepsilon_{k,\opt}
\end{equation}
almost surely, and define
\begin{equation}\label{eq: network optimization error}
    \varepsilon_{\opt} \triangleq \sum_{k=1}^K\pi_k\varepsilon_{k,\opt}.
\end{equation}
If $\mathsf R^o+\varepsilon_{\opt}<\Phi(0)$, $0\leq\delta<\delta_{\max}(\mathsf R^o+\varepsilon_{\opt})$, and $\rho<\mathdutchcal E_\Phi(\mathsf R^o+\varepsilon_{\opt},\delta)$, then under Assumptions~\ref{assump: risk function}--\ref{assump: feasibility},
\begin{equation}\label{eq: P_c_delta}
    P_{c,\delta} \geq 1 - 2\exp\Biggl\{-\frac{8N_{\max}}{\alpha^2\beta^2}\Bigl(\mathdutchcal E_\Phi(\mathsf R^o+\varepsilon_{\opt},\delta)-\rho \Bigr)^2 \Biggr\}.
\end{equation}
\end{proposition}
\begin{proof}
See Appendix~\ref{appendix: proof approximate ERM lemma}.
\end{proof}
\noindent
The key implication of Proposition~\ref{prop: P_c_delta} is that $P_{c,\delta}$ admits an exponential guarantee when the network Rademacher complexity $\rho$ is smaller than $\mathdutchcal E_\Phi(\mathsf R^o+\varepsilon_{\opt},\delta)$. The relevant properties of these auxiliary quantities are collected in Appendix~\ref{appendix: training auxiliary quantities}. In particular, for a fixed feasible decision margin $\delta$, $\mathdutchcal E_\Phi(r,\delta)$ decreases as the risk argument $r$ increases, while the feasible margin range characterized by $\delta_{\max}(r)$ also becomes smaller. Consequently, increasing $\varepsilon_{\opt}$ reduces $\mathdutchcal E_\Phi(\mathsf R^o+\varepsilon_{\opt},\delta)$, making the condition $\rho<\mathdutchcal E_\Phi(\mathsf R^o+\varepsilon_{\opt},\delta)$ more restrictive. Moreover, whenever this condition remains satisfied, the quantity $\mathdutchcal E_\Phi(\mathsf R^o+\varepsilon_{\opt},\delta)-\rho$ decreases, thereby reducing the exponential rate in~\eqref{eq: P_c_delta}. Thus, a larger optimization residual makes it more difficult to certify a prescribed decision margin and weakens the resulting training guarantee. Setting $\varepsilon_{k,\opt}=0$ for all agents recovers the exact-ERM specialization in~\cite{hu2025non-asymptotic}.

The bound also helps characterize the statistical cost of independent local training as reflected in the training guarantee, a defining feature of our framework when data pooling or joint training is unavailable. To quantify this cost, we compare independent training with a centralized benchmark under a fixed budget of available training samples. In the homogeneous setting considered in Appendix~\ref{appendix: centralized training comparison}, under matched optimization accuracy and when the relevant Rademacher complexity decreases with the number of training samples as described there, specializing Proposition~\ref{prop: P_c_delta} shows that the centralized benchmark yields a tighter lower bound on $P_{c,\delta}$ than an even split of the same samples among independently trained agents. The detailed comparison is provided in Appendix~\ref{appendix: centralized training comparison}.

\subsection{Sufficient-Communication Classification Guarantee}
\label{sec: convergence analysis}

We next turn to the classification performance of the trained network under sufficient communication. Section~\ref{sec: consistent training} characterizes the training performance through the probability that the learned network attains a prescribed decision margin $\delta$. We now examine how this margin condition translates into the probability of classification error for a fresh testing example. To distinguish the randomness arising from training from that of prediction, let
\begin{equation}\label{eq: training sigma field}
    \Ttrain \triangleq \sigma(\D_1,\ldots,\D_K, \mathcal{S}_1,\ldots,\mathcal{S}_K),
\end{equation}
which collects the training data and the randomness of the local training algorithms. Conditional on $\Ttrain$, the learned classifiers $\c_k$ and their empirical training means are therefore fixed. Moreover, the testing sample $(\h^*,\bgamma^*)$ is independent of $\Ttrain$. Since the training event $\C_\delta$ associated with~\eqref{eq: delta-margin consistent training condition} is determined by the training phase, it is $\Ttrain$-measurable.

From~\eqref{eq: label of single-sample classification}, under sufficient communication, the collective prediction $\widehat{\bgamma}$ is determined by the sign of the aggregate score $\c_\ave(\h^*)$. Let
\begin{equation}\label{eq: event of wrong classification-single sample case}
    \M\triangleq\{\bgamma^*\c_\ave(\h^*)\leq0\}
\end{equation}
denote the event that the signed aggregate score is nonpositive. Under any fixed tie-breaking rule at zero, the actual misclassification event is contained in $\M$. With $P_e\triangleq\P(\widehat{\bgamma}\neq\bgamma^*)$, the tower property and the $\Ttrain$-measurability of $\C_\delta$ give
\begin{align}
    \nonumber
    P_e &\leq\P(\M) =\E\!\left[\P(\M\givensmall\Ttrain)\right]\\
    \nonumber
    &= \E\!\left[\mathsf{1}[\C_\delta]\P(\M\givensmall\Ttrain)\right] + \E\!\left[\mathsf{1}[\overline{\C_\delta}]\P(\M\givensmall\Ttrain)\right]\\
    \label{eq: P_e decomposition}
    &\leq\E\!\left[\mathsf{1}[\C_\delta]\P(\M\givensmall\Ttrain)\right] + \P(\overline{\C_\delta}).
\end{align}
This decomposition separates the conditional prediction error for training realizations in $\C_\delta$ from the probability that training fails to attain the prescribed $\delta$-margin condition in~\eqref{eq: delta-margin consistent training condition}. The latter is characterized by Proposition~\ref{prop: P_c_delta}. Based on the event $\C_\delta$, the class-conditional mean of the aggregate score $\c_\ave(\h^*)$ lies on the correct side of the zero decision threshold with a margin greater than $\delta$. The remaining prediction error depends on the variability of the realized aggregate score $\c_\ave(\h^*)$ around its corresponding class-conditional mean.

The local observations $\h_1^*,\ldots,\h_K^*$ are different views of the same testing example and may therefore remain statistically dependent even after conditioning on the class label. Rather than requiring conditional independence, we allow their joint class-conditional distribution $P_\gamma$ to satisfy the following approximate tensorization property~\cite{caputo2015tensorization}. For a probability law $P$ and a bounded nonnegative function $G$, define
\begin{equation}
    \Ent_P(G) \triangleq \E_P[G\log G] - \E_P[G]\log\E_P[G].
\end{equation}
For each $k$, let $\h_{-k}$ collect all components of $\h$ except $\h_k$, and let $P_\gamma^{k|-k}(\cdot|\h_{-k})$ denote the conditional law of $\h_k$ given the remaining views under $P_\gamma$.
\begin{assump}[\textbf{Class-conditional approximate tensorization of entropy (ATE)}]
\label{assump: class conditional ATE}
For every $\gamma\in\Gamma$, there exists $\tau_\gamma\geq1$ such that
\begin{equation}\label{eq: class conditional ATE}
\Ent_{P_\gamma}(G)\leq\tau_\gamma\sum_{k=1}^K\E_{P_\gamma}\!\left[\Ent_{P_\gamma^{k|-k}(\cdot|\h_{-k})}\bigl(G(\cdot,\h_{-k})\bigr)\right]
\end{equation}
for every bounded nonnegative measurable $G$.\qed
\end{assump}
\noindent Assumption~\ref{assump: class conditional ATE} allows statistically dependent local views across agents beyond the product case. The coefficient $\tau_\gamma$ serves as a dependence parameter for the ATE condition: product class-conditional distributions satisfy the condition with $\tau_\gamma=1$, while larger values accommodate greater departures from this product structure. For convenience, define
\begin{equation}\label{eq: tau_max}
    \tau_{\max}\triangleq\max_{\gamma\in\Gamma}\tau_\gamma.
\end{equation}
Together with the training guarantee in Proposition~\ref{prop: P_c_delta}, the ATE condition~\eqref{eq: class conditional ATE} leads to the following sufficient-communication classification guarantee.

\begin{theorem}[\textbf{Sufficient-communication classification guarantee}]
\label{theorem: classification error}
Suppose Assumptions~\ref{assump: bound},~\ref{assump: network}, and~\ref{assump: class conditional ATE} hold, and let $\delta\geq0$. If the agents label their testing samples $\h_k^*$ with $\widehat{\bgamma}$ according to~\eqref{eq: label of single-sample classification}, then, almost surely on $\C_\delta$,
\begin{equation}
    \P(\M\givensmall\Ttrain) \leq \exp\Biggl\{ -\frac{\delta^2}{2\tau_{\max}\beta^2\sum_{k=1}^K\pi_k^2} \Biggr\}.
    \label{eq: P_e for single-sample case under delta-margin consistent training}
\end{equation}
If, in addition, the conditions of Proposition~\ref{prop: P_c_delta} hold, combining the two guarantees through \eqref{eq: P_e decomposition} yields an explicit upper bound on $P_e$.
\end{theorem}
\begin{proof}
See Appendix~\ref{appendix: classification error}.
\end{proof}
\noindent Theorem~\ref{theorem: classification error} establishes the connection between the training performance and the reliability of subsequent collaborative prediction. On $\C_\delta$, the conditional prediction error is controlled by the decision margin $\delta$, the dependence factor $\tau_{\max}$, and the Perron weights, while Proposition~\ref{prop: P_c_delta} characterizes the probability of attaining this favorable training event. Through~\eqref{eq: P_e decomposition}, the two guarantees jointly characterize the overall classification performance. We discuss several implications of this result.

\subsubsection{Relation to margin theory}

The decision margin $\delta$ plays complementary roles in training and prediction. Proposition~\ref{prop: P_c_delta} shows that requiring a larger $\delta$ makes the training condition more difficult to attain, whereas Theorem~\ref{theorem: classification error} shows that, once it is attained, a larger $\delta$ yields a stronger conditional prediction guarantee. This tradeoff is reminiscent of the role of margins in classical margin-based generalization theory~\cite{mohri2018foundations}. Given the classifier $\c_\ave$ from~\eqref{eq: single-sample ensemble classifier}, the quantity ${\bgamma}^*\c_\ave(\h^*)$ is known as the ``margin'' associated with the testing sample $\h^*$. The class-conditional expectations of the margin distribution associated with $\c_\ave$ are $\bmu^+(\widehat{\f})$ under class $+1$ and $-\bmu^-(\widehat{\f})$ under class $-1$. Therefore, the $\delta$-margin condition in~\eqref{eq: delta-margin consistent training condition} provides a first-order statistical characterization of the margin distribution for classifier $\c_\ave$. According to Theorem~\ref{theorem: classification error}, this characterization leads to a classification guarantee without specifying the exact distributions of the local decision statistics, in contrast to some distribution-specific soft combination analyses~\cite{kuncheva2002theoretical}.

\subsubsection{Role of the Perron weights}

From~\eqref{eq: common decision statistic}, the Perron entries $\pi_k$ determine the relative influence of the local classifiers in the sufficient-communication aggregate. Fusion weight design has been widely studied in ensemble learning, typically for classifiers sharing the same feature representation~\cite{fumera2005theoretical,heskes1997selecting, kuncheva2014combining}. The extensions to heterogeneous representations are more problem-dependent~\cite{polikar2006ensemble,kuncheva2014combining}. Here, rather than optimizing the fusion weights under arbitrary heterogeneity, we characterize the effect of the Perron weights induced by the communication policy. Their role extends beyond the factor $\sum_k\pi_k^2$ in the conditional prediction bound~\eqref{eq: P_e for single-sample case under delta-margin consistent training}: the training quantities $\rho$, $\mathsf R^o$, $\alpha$, and $\varepsilon_{\opt}$, as well as the achievable decision margin, also depend on $\pi$. Consequently, although the uniform Perron vector minimizes $\sum_k\pi_k^2$ when $\delta$ and $\tau_{\max}$ are fixed, it need not optimize the overall classification performance. Importantly, the simple averaging rule for constructing $A$, where each agent assigns equal combination weights to all neighbors, does not guarantee uniform $\pi$ on arbitrary topologies. Several useful rules for constructing doubly-stochastic $A$ are available in~\cite{sayed2014adaptation}.

\subsubsection{Benefits of cooperative prediction}

To isolate the benefit of cooperation during prediction, suppose that the same target decision margin $\delta$ is attained collaboratively and by a non-cooperative agent. The conditional error bound in Theorem~\ref{theorem: classification error} is tighter under collaboration whenever
\begin{equation}\label{eq: prediction benefit}
    \tau_{\max}\sum_{k=1}^K\pi_k^2<1.
\end{equation}
This condition reflects the balance between the gain from averaging the local scores and the penalty induced by their statistical dependence. In particular, for product class-conditional distributions, we may take $\tau_{\max}=1$, and \eqref{eq: prediction benefit} reduces to $\sum_k\pi_k^2<1$, which holds for any nontrivial network satisfying Assumption~\ref{assump: network}. Cooperation also provides a distinct benefit through the network-level margin condition. The event $\C_\delta$ requires only the $\pi$-weighted class-conditional means to attain the prescribed two-sided margin; individual classifiers need not satisfy the same condition. Consequently, agents with weak local evidence can be compensated by sufficiently informative or complementary observations from other agents. This is consistent with Assumption~\ref{assump: feasibility}, which likewise requires informativeness only at the network level.

\subsection{Finite-Round Communication}

Theorem~\ref{theorem: classification error} characterizes the sufficient-communication limit. When communication is restricted to a finite number of rounds, however, consensus may not yet be reached and the prediction remains agent-dependent. We next characterize how limited communication affects the classification performance at each agent.

According to~\eqref{eq: distributed learning rule}, after $t$ communication rounds, agent $k$ uses the following decision statistic
\begin{equation}\label{eq: decision statistic at time t}
    \blambda_{k,t} = \sum_{\ell=1}^K[A^t]_{\ell k}\c_\ell(\h_\ell^*),
\end{equation}
and predicts the label
\begin{equation}
    \widehat{\bgamma}_{k,t} \triangleq \sign(\blambda_{k,t}).
\end{equation}
Define the corresponding event
\begin{equation}\label{eq: finite round nonpositive margin event}
    \M_{k,t} \triangleq \{\bgamma^*\blambda_{k,t}\leq0\}.
\end{equation}
Under any fixed tie-breaking rule at zero, the misclassification event is contained in $\M_{k,t}$. Thus, with $P_{k,t}\triangleq\P(\widehat{\bgamma}_{k,t}\neq\bgamma^*)$, applying the same decomposition as in~\eqref{eq: P_e decomposition} gives
\begin{align}
    P_{k,t} \leq \P(\M_{k,t})\leq \E\!\left[ \mathsf{1}[\C_\delta]\P(\M_{k,t}\givensmall\Ttrain) \right] +\P(\overline{\C_\delta}).
    \label{eq: P_e equation at t-th communication round}
\end{align}
On the event $\C_\delta$, a positive decision margin is guaranteed for the $\pi$-weighted aggregate under sufficient communication~\eqref{eq: single-sample ensemble classifier}. At a finite communication round, from~\eqref{eq: decision statistic at time t}, agent $k$ combines the local scores $\c_\ell(\h_\ell^*)$ using the weights $[A^t]_{\ell k}$.  To characterize how much of this guaranteed margin is retained after $t$ rounds, fix any $\sigma\in(\sigma_A,1)$. Matrix theory ensures that there exists a constant $C(A,\sigma)>0$ such that~\cite{horn2012matrix}
\begin{equation}\label{eq: convergence constant main}
    |[A^t]_{\ell k}-\pi_\ell|\leq C(A,\sigma)\sigma^t
\end{equation}
for all agents $k,\ell$ and communication rounds $t$. Accordingly, define the residual finite-round margin
\begin{equation}\label{eq: finite round residual margin}
    r_t(\delta) \triangleq \delta-2\beta K C(A,\sigma)\sigma^t.
\end{equation}
The second term bounds the possible reduction of the guaranteed margin after only $t$ communication rounds.

\begin{theorem}[\textbf{Finite-round classification guarantee}]
\label{theorem: finite-time classification error}
Suppose Assumptions~\ref{assump: bound},~\ref{assump: network} and~\ref{assump: class conditional ATE} hold, and let $\delta>0$. Fix $\sigma\in(\sigma_A,1)$ and a corresponding $C(A,\sigma)$. Then, for every agent $k$ and every integer $t\geq0$ such that $r_t(\delta)>0$, almost surely on $\C_\delta$,
\begin{equation}\label{eq: P_e for t-th round single-sample case under delta-margin consistent training}
    \P(\M_{k,t}\givensmall\Ttrain) \leq \exp\Biggl\{ -\frac{r_t^2(\delta)}{2\tau_{\max}\beta^2\sum_{\ell=1}^K([A^t]_{\ell k})^2} \Biggr\}.
\end{equation}
If, in addition, the conditions of Proposition~\ref{prop: P_c_delta} hold, combining the two guarantees through \eqref{eq: P_e equation at t-th communication round} yields an explicit upper bound on $P_{k,t}$.
\end{theorem}
\begin{proof}
See Appendix~\ref{appendix: finite round communication}.
\end{proof}
\noindent The condition $r_t(\delta)>0$ provides a sufficient communication horizon for retaining a positive guaranteed margin. In particular, when $2\beta K C(A,\sigma)>\delta$, it is equivalent to
\begin{equation}\label{eq: t condition for t-th round}
    t>\frac{\log\frac{\delta}{2\beta K C(A,\sigma)}}{\log\sigma}.
\end{equation}
Thus, faster decay of the transient term $C(A,\sigma)\sigma^t$ allows a positive margin to be guaranteed with fewer communication rounds. The bound in \eqref{eq: P_e for t-th round single-sample case under delta-margin consistent training} is both round- and agent-dependent. Finite communication affects the prediction guarantee in two ways. First, the residual margin $r_t(\delta)$ reflects the reduction in the guaranteed decision margin caused by incomplete mixing. Second, the column $\{[A^t]_{\ell k}\}_{\ell=1}^K$ specifies the  fusion profile at agent $k$ and determines the corresponding concentration term. Therefore, finite-round performance depends on the full combination matrix $A$, through both the transient fusion weights and their convergence toward the Perron vector. The second-largest eigenvalue magnitude $\sigma_A$ determines the admissible geometric rates in \eqref{eq: convergence constant main}, while $C(A,\sigma)$ captures the associated matrix-power constant. This is consistent with the role of spectral properties in the convergence behavior of distributed averaging~\cite{xiao2004fast}. As $t\to\infty$, $A^t\to\pi\mathbbm{1}^\top$ and $r_t(\delta)\to\delta$. Accordingly, the finite-round conditional bound converges to the sufficient-communication bound in Theorem~\ref{theorem: classification error}.

\subsection{Finite-Precision Communication}
\label{sec: finite bit theory}

Theorem~\ref{theorem: finite-time classification error} considers finite-round communication with real-valued messages. In practical implementations, however, transmitted scalars are typically represented with finite precision. We therefore introduce a finite-bit version of the prediction-stage communication rule and characterize its effect on classification performance. The local training procedures, empirical centering, and the training event $\C_\delta$ remain unchanged.

Since both a bounded score and its empirical training mean lie in $[-\beta,\beta]$, every centered score satisfies
\begin{equation}\label{eq: centered score quantization range}
    |c_{k,f_k}(h)|\leq2\beta.
\end{equation}
Let $b\geq1$ denote the available bits per transmitted scalar, and let $2\leq L_b\leq2^b$ reconstruction levels be used. Over the interval $[-2\beta,2\beta]$, consider the uniform alphabet
\begin{align}
\nonumber
    \mathcal V_b &=\{v_j=-2\beta+j\Delta_b:\ j=0,\ldots,L_b-1\},\\
    \label{eq: finite bit alphabet theory}
    \Delta_b &= \frac{4\beta}{L_b-1}.
\end{align}
We consider here an adjacent-level stochastic rounding protocol~\cite{aysal2007probabilistic}. For $u\in[v_j,v_{j+1}]$, $j=0,\ldots,L_b-2$, the quantizer $\mathsf Q_b$ returns one of the two adjacent reconstruction levels $v_j$ and $v_{j+1}$ according to
\begin{equation}\label{eq: stochastic rounding protocol}
    \mathsf Q_b(u)=
    \begin{cases}
        v_j, & \text{with probability\;} \dfrac{v_{j+1}-u}{\Delta_b},\\
        v_{j+1}, & \text{with probability\;} \dfrac{u-v_j}{\Delta_b}.
    \end{cases}
\end{equation}
An important feature of this protocol is that the quantizer $\mathsf Q_b$ is unbiased, that is,
\begin{equation}
    \E[\mathsf Q_b(u)\givensmall u]=u.
\end{equation}
Let $x_{k,t}^{(b)}$ denote the finite-precision decision statistic at agent $k$ after $t$ communication rounds. Starting from $x_{k,0}^{(b)}=\c_k(\h_k^*)$, each agent quantizes its current decision statistic before transmission, and the receiving agents combine the quantized values according to
\begin{equation}\label{eq: finite bit recursion theory}
    q_{\ell,t}^{(b)} = \mathsf Q_b(x_{\ell,t}^{(b)}), \quad x_{k,t+1}^{(b)} = \sum_{\ell=1}^K a_{\ell k}q_{\ell,t}^{(b)}.
\end{equation}
The random draws used for stochastic rounding are independent across agents and communication rounds, with each agent sending the same quantized value to all of its neighbors. After $t$ rounds, agent $k$ predicts according to
\begin{equation}
    \widehat{\bgamma}_{k,t}^{(b)} \triangleq \sign(x_{k,t}^{(b)}).
\end{equation}
Define the event for the finite-precision decision statistic
\begin{equation}\label{eq: M_k_t for finite precision}
    \M_{k,t}^{(b)} \triangleq \{\bgamma^*x_{k,t}^{(b)}\leq0\},
\end{equation}
and let
\begin{equation}
    P_{k,t}^{(b)}\triangleq\P(\widehat{\bgamma}_{k,t}^{(b)}\neq\bgamma^*)
\end{equation}
denote the classification error probability. As in the real-valued finite-round case, the decomposition in~\eqref{eq: P_e decomposition}, with $\M$ replaced by $\M_{k,t}^{(b)}$, separates the conditional prediction error $\P(\M_{k,t}^{(b)}\givensmall\Ttrain)$ on $\C_\delta$ from the probability of the unfavorable training event $\overline{\C_\delta}$.

Since quantization is performed at every communication round, the perturbations introduced at earlier rounds are propagated through subsequent neighbor-aggregation steps. To state the resulting finite-precision guarantee compactly, define the accumulated quantization term
\begin{align}
\nonumber
    V_{k,t}^{\rm q} &\triangleq \frac{\Delta_b^2}{4} \sum_{r=1}^{t}\sum_{j=1}^K([A^r]_{jk})^2  \\
    \label{eq: quantization proxy theory}
    &=\frac{4\beta^2}{(L_b-1)^2}\sum_{r=1}^{t}\sum_{j=1}^K([A^r]_{jk})^2.
\end{align}
The quantity $V_{k,t}^{\rm q}$ depends jointly on the communication precision, the number of communication rounds, and the network combination policy.

\begin{theorem}[\textbf{Finite-precision classification guarantee}]
\label{theorem: finite bit communication}
Suppose Assumptions~\ref{assump: bound},~\ref{assump: network}, and~\ref{assump: class conditional ATE} hold, and let $\delta>0$. Suppose the agents follow the stochastic finite-precision recursion~\eqref{eq: finite bit recursion theory}. Fix $\sigma\in(\sigma_A,1)$ and a corresponding $C(A,\sigma)$. Then, for every agent $k$ and every integer $t\geq0$ such that $r_t(\delta)>0$, almost surely on $\C_\delta$,
\begin{equation}\label{eq: finite bit classification theory}
    \P(\M_{k,t}^{(b)}\givensmall\Ttrain) \leq \exp\Biggl\{ -\frac{r_t^2(\delta)}{2\left[ \tau_{\max}\beta^2\sum_{\ell=1}^K([A^t]_{\ell k})^2 + V_{k,t}^{\rm q} \right]} \Biggr\}.
\end{equation}
If, in addition, the conditions of Proposition~\ref{prop: P_c_delta} hold, combining the two bounds in~\eqref{eq: finite bit classification theory} and~\eqref{eq: P_c_delta} yields an explicit upper bound on $P_{k,t}^{(b)}$.
\end{theorem}
\begin{proof}
See Appendix~\ref{appendix: finite bit communication}.
\end{proof}
\noindent Theorem~\ref{theorem: finite bit communication} extends the real-valued finite-round guarantee~\eqref{eq: P_e for t-th round single-sample case under delta-margin consistent training} established in Theorem~\ref{theorem: finite-time classification error} by quantifying the additional effect of finite communication precision. The same residual finite-round margin term $r_t(\delta)$ from~\eqref{eq: finite round residual margin} appears in both guarantees, while finite precision introduces the additional term $V_{k,t}^{\rm q}$ in the denominator of the prediction exponent~\eqref{eq: finite bit classification theory}. This result also reveals a tradeoff pertaining to the number of communication rounds. Increasing $t$ increases the guaranteed residual margin $r_t(\delta)$ by reducing the effect of incomplete mixing, but also allows additional quantization effects to accumulate through $V_{k,t}^{\rm q}$. Moreover, the combination policy $A$ influences both effects through the finite-round fusion weights and the propagation of quantization noise across rounds. Consequently, at fixed precision, the bound need not improve monotonically with the number of communication rounds.

The communication precision also determines the communication cost. Let
\begin{equation}\label{eq: off diagonal communication links theory}
    E_{\rm off}(A) \triangleq \sum_{k=1}^K \sum_{\substack{\ell=1\\\ell\neq k}}^K \mathsf{1}[a_{\ell k}>0]
\end{equation}
denote the number of active directed off-diagonal links. If each quantizer index is represented by a fixed-length $b$-bit word, the per-round and total communication costs per testing example are
\begin{equation}\label{eq: finite bit communication cost theory}
    B_{\rm round} = bE_{\rm off}(A),  \quad  B_{\rm total}(t) = tbE_{\rm off}(A),
\end{equation}
respectively, excluding packet headers and coding overhead. Self-loops correspond to the local operations and are therefore not counted.  Moreover, for any fixed $t$, we have from~\eqref{eq: quantization proxy theory} that
\begin{equation}
    V_{k,t}^{\rm q} = O\bigl((L_b-1)^{-2}\bigr).
\end{equation}
Here, the $O(\cdot)$ notation describes the decay with respect to $L_b$ for fixed $t$. Under the full-code convention $L_b=2^b$, this becomes $V_{k,t}^{\rm q}=O(2^{-2b})$. Therefore, increasing the bit width increases the communication cost linearly, while the finite-precision term in the bound~\eqref{eq: finite bit classification theory} decays as $O(2^{-2b})$. For fixed $t$, as $b\to\infty$, the finite-precision guarantee converges to the corresponding real-valued finite-round guarantee.

\subsection{PAC-Style Generalization Guarantee}
\label{sec: PAC-style generalization bounds}

The preceding results characterize collaborative prediction through the lens of expected network margin. We now consider a complementary PAC-style perspective that relates the population classification error of the realized classifiers to their performance observed on the training data. This is the core idea of classical margin-based generalization bounds~\cite{mohri2018foundations}. Specifically, to quantify this performance, for $\eta>0$, we define the $\eta$-margin loss~\cite{mohri2018foundations}
\begin{equation}\label{eq: definition of eta-margin loss function}
    \Phi_\eta(x) = \min\biggl(1,\max\Bigl(0,1-\frac{x}{\eta}\Bigr)\biggr).
\end{equation}
It satisfies
\begin{equation}\label{eq: property of eta-margin loss}
    \mathsf{1}[\sign(f(\h))\neq\bgamma] \leq \Phi_\eta(\bgamma f(\h)) \leq \mathsf{1}[\bgamma f(\h)\leq\eta].
\end{equation}
The margin parameter $\eta$ used here is distinct from the decision-margin level $\delta$ in \eqref{eq: delta-margin consistent training condition}. For a training set $\{(h_n,\gamma_n)\}_{n=1}^{m}$, the empirical margin loss for an arbitrary function $f$ is given by
\begin{equation}\label{eq: empirical eta-margin loss}
    P_{e,\emp}^\eta(f)\triangleq \frac{1}{m}\sum_{n=1}^m\Phi_\eta(\gamma_nf(h_n)).
\end{equation}
By~\eqref{eq: property of eta-margin loss}, $P_{e,\emp}^\eta(f)$ upper bounds the empirical classification error on the same sample. When the classifiers depend on the training data, their classification errors on a fresh example are understood conditionally on the realized training information. Specifically, for a local classifier $\c_k$ and the aggregate classifier $\c_\ave$, respectively, we define
\begin{align}
\label{eq: PAC conditional local population error}
    P_e(\c_k) &\triangleq \P\bigl(\sign(\c_k(\h_k^*))\neq\bgamma^* \given \Ttrain\bigr), \\
    P_e(\c_\ave) &\triangleq \P\bigl(\sign(\c_\ave(\h^*))\neq\bgamma^* \given \Ttrain\bigr).
    \label{eq: PAC conditional network population error}
\end{align}
To evaluate the empirical margin loss of the aggregate classifier $\c_\ave$ with~\eqref{eq: empirical eta-margin loss}, the training samples must provide the aligned collection of local views on which $\c_\ave$ operates. Therefore, for the PAC analysis, we specialize the training setup to equal local sample sizes $N_k=N$ and aligned network examples. Specifically, let
\begin{equation}\label{eq: training set D for the network}
    \D \triangleq \bigl\{(\h_n,\bgamma_n)\bigr\}_{n=1}^{N}
\end{equation}
consist of $N$ i.i.d.\ draws from the joint distribution $Q$, where $\h_n=(\h_{1,n},\ldots,\h_{K,n})$. The local training set at agent $k$ is the corresponding projection
\begin{equation}\label{eq: local training set D_k}
    \D_k \triangleq \bigl\{(\h_{k,n},\bgamma_n)\bigr\}_{n=1}^{N}.
\end{equation}
This specialization leaves the local training procedures in Section~\ref{sec: the training phase} unchanged. The network examples are independent across $n$, whereas the $K$ views within each example may remain statistically dependent.

Based on the network training set $\D$ in~\eqref{eq: training set D for the network}, we establish the following PAC-style generalization bounds for the classifiers $\c_k$ and $\c_\ave$ using their empirical margin losses. Here, $\rho_k$ and $\rho$ denote the individual and network Rademacher complexities defined in Appendix~\ref{appendix: training auxiliary quantities}, evaluated at the common sample size $N_k=N$.

\begin{theorem}[\textbf{PAC-style margin generalization guarantee}]
\label{theorem: generalization bound for c_k and c_ave}
Under Assumption~\ref{assump: bound} and the sampling setup above, fix any $\eta>0$. Then, for any $\epsilon\in(0,1)$ and every fixed agent $k$,
\begin{equation}\label{eq: margin-based generalization bound for c_k}
    P_e(\c_k)\leq P_{e,\emp}^\eta(\c_k) + \frac{4}{\eta}\rho_k + \mathcal{R}(\epsilon,N)
\end{equation}
with probability at least $1-\epsilon$. Moreover,
\begin{equation}\label{eq: margin-based generalization bound for c_ave}
    P_e(\c_\ave)\leq P_{e,\emp}^\eta(\c_\ave) + \frac{4}{\eta}\rho + \mathcal{R}(\epsilon,N)
\end{equation}
with probability at least $1-\epsilon$, where
\begin{equation}\label{eq: definition of R(e,N)}
    \mathcal{R}(\epsilon,N)\triangleq \frac{4\beta}{\eta\sqrt{N}} \Biggl(1+\sqrt{2\log\frac{1}{\epsilon}}\Biggr).
\end{equation}
\end{theorem}
\begin{proof}
See Appendix~\ref{appendix: PAC generalization bound}.
\end{proof}
\noindent Theorem~\ref{theorem: generalization bound for c_k and c_ave} provides a complementary guarantee for the trained classifiers without investigating the prescribed network-margin event $\C_\delta$. It also does not require Assumption~\ref{assump: class conditional ATE}: dependence among the local views within each network example is incorporated through their joint distribution $Q$. Moreover, the guarantee does not require exact empirical risk minimization or a prescribed optimization gap. Therefore, it holds \emph{uniformly} across classifiers returned by the local training procedures in Section~\ref{sec: the training phase}.

From~\eqref{eq: margin-based generalization bound for c_k} and~\eqref{eq: margin-based generalization bound for c_ave}, the local and collaborative bounds have the same form, but differ in their empirical margin losses and complexity terms. Structurally, $\c_k$ operates on a single local feature space $\H_k$, whereas $\c_\ave$ operates on the joint feature space $\H_1\times\cdots\times\H_K$. The empirical margins of $\c_\ave$ can therefore reflect how information from different local views is combined. The theorem does not, however, impose a universal ordering between the local and collaborative bounds. Any collaborative advantage must be reflected in the aggregate margin loss $P_{e,\emp}^\eta(\c_\ave)$ relative to the corresponding complexity terms $\rho$. This differs from the conventional ensemble setting, where multiple models typically operate on the same feature representation and the benefit of combining their outputs is often discussed in terms of variance reduction~\cite{mohri2018foundations,fumera2008theoretical}.

\subsection{Relation to Social Learning}
\label{sec: Relation to SL}

The DeGroot rule~\eqref{eq: distributed learning rule} and conventional social learning (SL) strategies both aggregate information across a network, but differ in how evidence enters the inference process. In the present framework, agents repeatedly communicate decision statistics initialized from a fixed testing example, whereas conventional SL recursions continually incorporate newly arriving observations. To clarify the relation between these two inference mechanisms, Appendix~\ref{appendix: classification error using SL rule} considers a standard geometric-pooling SL recursion under a \emph{static-observation} specialization, in which each agent reuses the same local testing observation at every communication round. This specialization places the SL rule in the same single-testing-example setting as the DeGroot rule and enables a direct comparison, while remaining distinct from the conventional streaming operation of SL.

Under this specialization, repeated SL updates accumulate the same local evidence rather than new temporal evidence. After normalization by the number of communication rounds, the SL decision statistic converges to the same $\pi$-weighted aggregate $\c_\ave$ as the DeGroot rule~\eqref{eq: distributed learning rule}. Consequently, whenever $\c_\ave(\h^*)\neq0$, the two strategies yield the same prediction under sufficient communication. Their finite-round decision statistics are generally different, however, and therefore lead to different finite-round guarantees. The complete comparison and the corresponding finite-round SL guarantee are provided in Appendix~\ref{appendix: classification error using SL rule}.

\medskip
\noindent\textbf{Raw-score counterparts.}
The preceding analysis considers the empirically centered local scores $\c_k(\h_k)$ defined in~\eqref{eq: trained classifier}. As shown by Proposition~\ref{prop: effect of additive centering}, centering affects threshold alignment rather than being a structural requirement for collaboration. We refer to the learned scores $\widehat{\f}_k(\h_k)$ used without empirical centering as \emph{raw scores}. In Appendix~\ref{appendix: raw score counterparts}, we provide an analogous analysis for raw scores. This analysis shows that the main collaborative guarantees persist without empirical centering, although the conditions for attaining a prescribed zero-threshold margin and the constants appearing in the finite-round bounds change. Neither representation universally dominates the other: centering can improve threshold alignment, whereas raw scores avoid estimating the empirical center and have a smaller worst-case score magnitude than the centered scores.

\section{Numerical Simulations}
\label{sec:simulations}

We evaluate the proposed test-time collaborative classification framework under two complementary observation models. The first is a controlled patch-partition benchmark based on CIFAR-10, where each agent observes a spatial patch extracted from a common image. The second is a multi-view benchmark constructed from rendered objects in ModelNet40, where each agent observes the same 3D object from a different viewpoint. The two benchmarks introduce different forms of local heterogeneity, namely spatial partitioning and viewpoint variation. We further examine dependence across agents, heterogeneous local models, communication topology, quantization, adaptive stopping, and corrupted reports. Unless otherwise stated, all curves are averaged over multiple Monte Carlo repetitions and the error bars indicate $95\%$ confidence intervals. Additional implementation details and extended numerical results are provided in Appendix~\ref{app:exp_details}.

\begin{table*}[!t]
\centering
\caption{Training and inference coordination of the methods considered in the main baseline comparisons.}
\label{tab:baseline_regimes}
\scriptsize
\setlength{\tabcolsep}{4pt}
\renewcommand{\arraystretch}{1}
\begin{tabular}{
    p{0.19\textwidth}
    p{0.24\textwidth}
    p{0.27\textwidth}
    p{0.22\textwidth}}
\hline
\textbf{Method} & \textbf{Training coordination} & \textbf{Additional fitting} & \textbf{Inference coordination}\\
\hline
Non-cooperative & None & None & None\\
Avg-stat / Vote & None & None & One-shot centralized averaging / voting\\
Learned (simplex) fusion & None & Fusion rule fitted on aligned validation outputs & One-shot centralized fusion\\
AdaBoost & Centralized & Boosting-based fitting of local predictors & One-shot weighted fusion\\
Ours / Ours (no centering) & None & None & Iterative peer-to-peer scalar exchange\\
VFL-JT & Joint training on aligned labeled samples & Server-side fusion head learned jointly & One-shot server fusion\\
Central oracle & Centralized full-information training & None & Centralized full-information prediction\\
\hline
\end{tabular}
\end{table*}

\subsection{CIFAR-10 patch-partition benchmark}
\label{subsec:cifar_patch}

We begin with a controlled patch-partition benchmark based on CIFAR-10~\cite{krizhevsky2009learning}. We consider the binary task of distinguishing \emph{cats} from \emph{dogs}. Each $32\times 32$ RGB image is divided into a $3\times 3$ grid of non-overlapping spatial patches, yielding $K=9$ agents, and agent $k$ receives only the patch associated with its grid location. In this way, the agents observe different local regions of the same image and collaborate only during inference. The communication network is taken to be a directed Erd\"{o}s--R\'enyi graph~\cite{newman2018networks}, generated randomly at the beginning of the experiment. The observation map of the nine agents and the communication topology are illustrated in Fig.~\ref{fig: network and feature map}.

\begin{figure}[!t]
    \centering
    \subfloat[]{\includegraphics[width=0.4\linewidth]{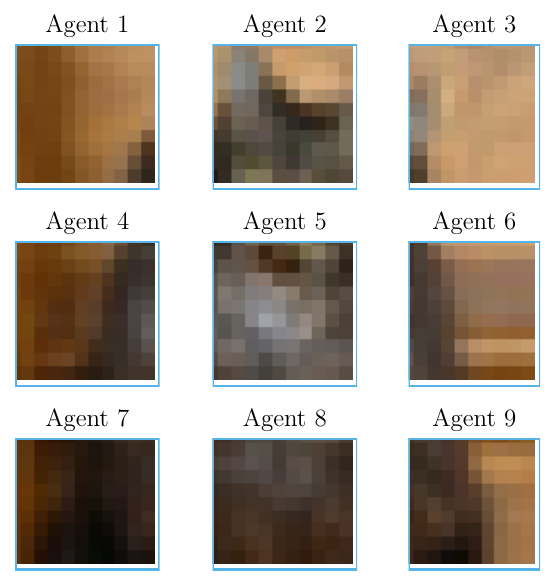}\label{fig: feature map}}
    \hfil
    \subfloat[]{\includegraphics[width = 0.45\linewidth]{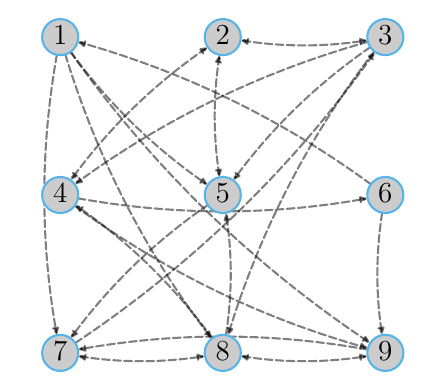}\label{fig: network}}
    \caption{CIFAR-10 patch-partition benchmark. Panel~(a) shows the local patch observed by each agent for a representative cat image, and panel~(b) shows the communication topology used for collaboration.}
    \label{fig: network and feature map}
\end{figure}

At each agent, a convolutional neural network (CNN) is trained independently on the corresponding patch dataset. For simplicity, we use the same local labeled dataset size $N_0$ for all agents. For each value of $N_0$, we construct balanced local datasets, train the classifiers independently, perform collaborative classification at test time, and average the results over multiple repetitions. The detailed training protocol is provided in Appendix~\ref{app:cifar_details}.

\begin{figure}[!t]
    \centering
    \subfloat[Decision margin]{\includegraphics[width=0.8\linewidth]{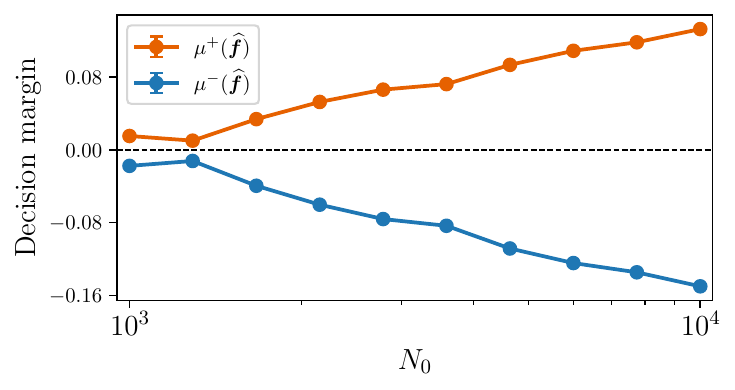}\label{fig:cifar_patch_margin}}
    \hfill
    \subfloat[Probability of error\label{fig:cifar_patch_pe}]{%
        \begin{minipage}{\linewidth}
            \centering
            \includegraphics[width=0.8\linewidth]{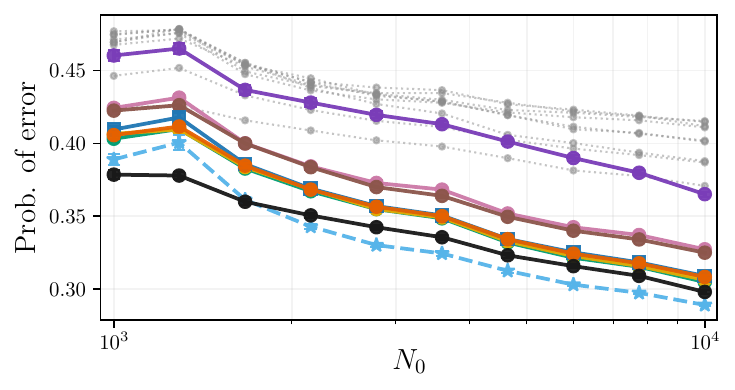}
            \vspace{-0.4em}
            \includegraphics[width=0.95\linewidth]{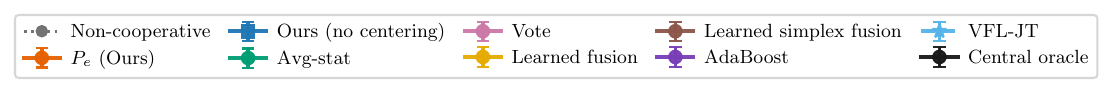}
        \end{minipage}%
    }
    \hfill
    \subfloat[Instantaneous probability of error]{\includegraphics[width=0.8\linewidth]{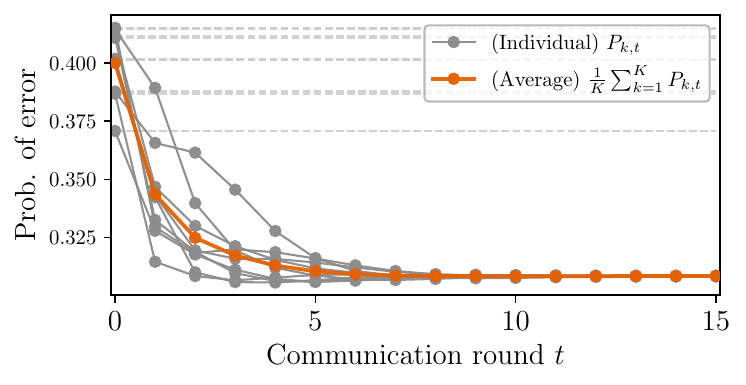}\label{fig:cifar_patch_process}}
    \caption{Main CIFAR-10 patch-partition results. Panel~(a) reports the empirical class-conditional decision statistics, panel~(b) compares the probability of error under different values of $N_0$, and panel~(c) shows the error evolution over the communication rounds for $N_0=10000$.}
    \label{fig:cifar_patch_main}
\end{figure}

Figure~\ref{fig:cifar_patch_main} summarizes the main performance results for this CIFAR-10 benchmark. Panel~(a) reports the empirical class-conditional decision statistics under different values of $N_0$. In the notation of Section~\ref{sec: main results}, these curves correspond to the empirical behavior of $\bmu^+(\widehat{\f})$ and $\bmu^-(\widehat{\f})$ in~\eqref{eq: expected decision statistic for the network}. As $N_0$ increases, $\bmu^+(\widehat{\f})$ moves farther above zero and $\bmu^-(\widehat{\f})$ moves farther below zero, so that the separation between the two curves grows overall, consistent with a larger achievable decision margin in the sense of condition~\eqref{eq: delta-margin consistent training condition}. This suggests that the network of trained classifiers becomes more informative as more training samples are available, which is beneficial for classification at test time.

Panel~(b) shows the corresponding classification performance. We compare the proposed method with several rules from ensemble learning and decision fusion, together with the non-cooperative baseline, an explicit centering ablation \emph{Ours (no centering)}, the VFL-JT baseline, and a central oracle trained on the full image as a full-information benchmark. The VFL-JT baseline uses a score-level joint-training scheme inspired by~\cite{hu2019fdml}, where the local models and a server-side fusion head are optimized jointly using aligned labeled samples. Table~\ref{tab:baseline_regimes} summarizes the training and inference coordination of the methods considered in the main baseline comparison. Since several of these reference rules depend on centralized aggregation of local outputs, the comparison is carried out in the regime of sufficient communication, where the distributed recursion~\eqref{eq: distributed learning rule} approaches its $\pi$-weighted consensus limit given in \eqref{eq: common decision statistic}. From Fig.~\ref{fig:cifar_patch_pe}, the probability of error decreases as $N_0$ increases, consistent with the increased decision margins observed in Fig.~\ref{fig:cifar_patch_margin}. At the same time, the proposed method consistently improves upon non-cooperative prediction and remains competitive with the stronger fusion baselines. With joint training, VFL-JT achieves lower error than the proposed method over the tested range. We also note that under sufficient communication, \emph{Avg-stat} coincides with the proposed method when $A$ is doubly-stochastic, while \emph{Ours} and \emph{Ours (no centering)} exhibit similar performance over the tested range.

Panel~(c) illustrates the evolution of $P_{k,t}$ during the collaboration process. As the communication round $t$ grows, the decision statistic $\blambda_{k,t}$ converges to $\blambda_\ave$, and $P_{k,t}$ approaches the error of the limiting aggregate classifier. It is worth noting that $P_{k,0}$ at $t=0$ corresponds to the error incurred by agent $k$ in the non-cooperative scenario. The light grey dashed horizontal lines in Fig.~\ref{fig:cifar_patch_process} indicate these baseline error levels for reference. The average of $P_{k,t}$ across all $K$ agents is also shown in Fig.~\ref{fig:cifar_patch_process}, which decreases over the communication rounds before stabilizing. This highlights the improvement in the network-wide performance under collective prediction.

We next present two additional studies for the CIFAR-10 patch-partition benchmark.

\subsubsection{Conditional mutual information between patches}
\label{sec:cifar_patch_cmi}

\begin{figure}[!ht]
    \centering
    \subfloat[Conditional mutual information matrix]{\includegraphics[width=0.75\linewidth]{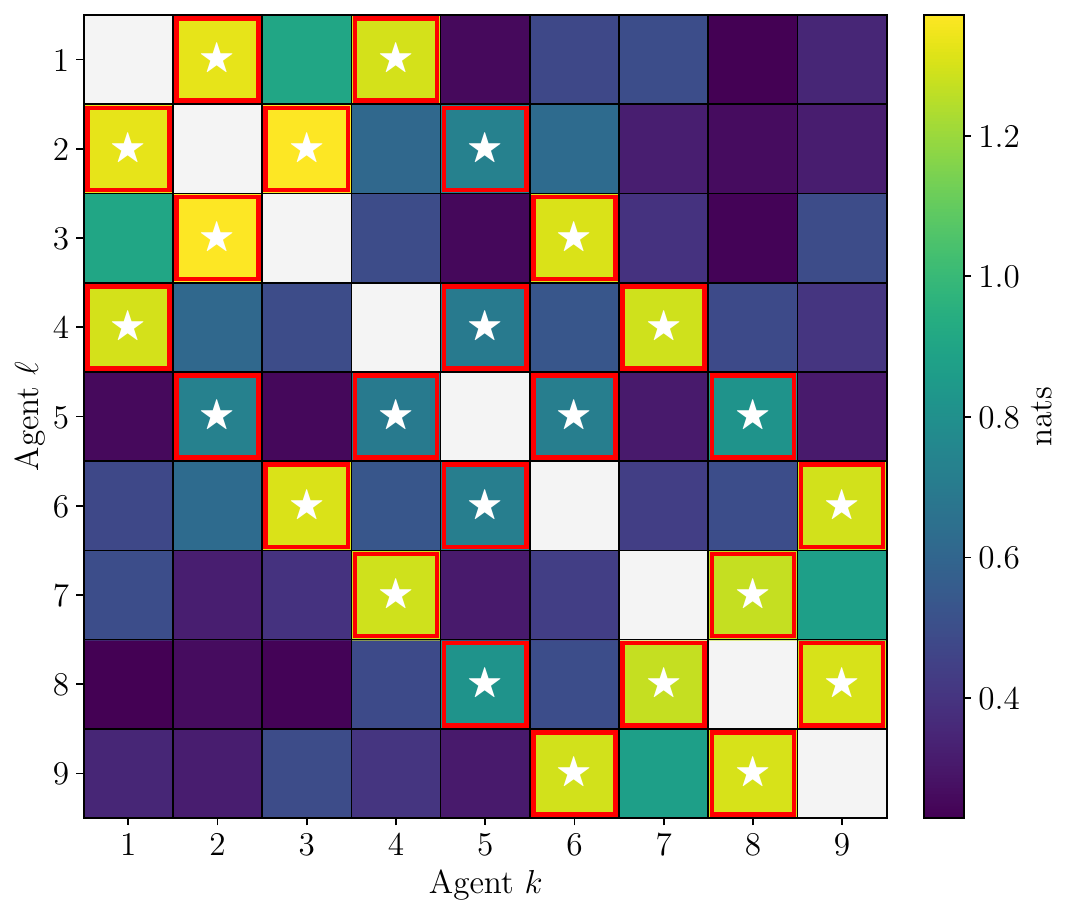}\label{fig:cifar_patch_cmi_heatmap}}
    \hfill
    \subfloat[Grouped conditional mutual information]{\includegraphics[width=0.85\linewidth]{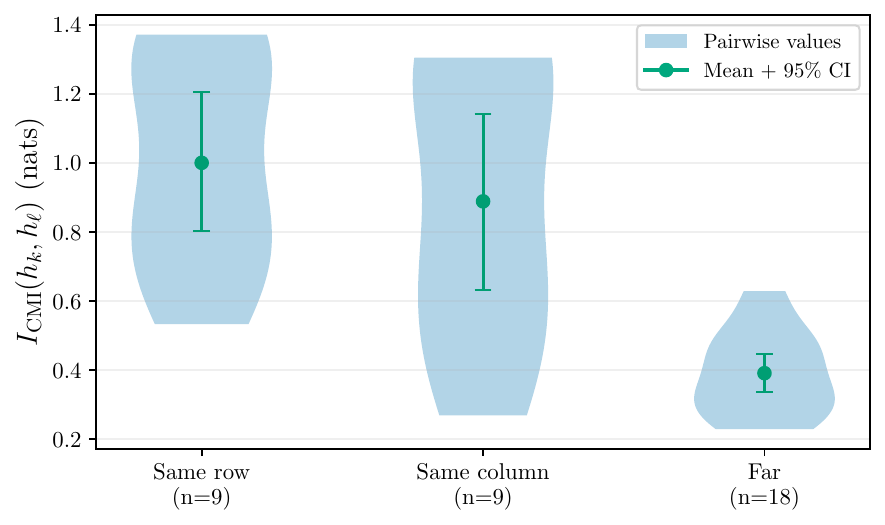}\label{fig:cifar_patch_cmi_bucket}}
    \caption{Conditional mutual information between CIFAR-10 patches. Panel~(a) shows the estimated pairwise conditional mutual information matrix across the nine patch agents. White stars mark patch pairs that are adjacent on the $3\times3$ grid shown in Fig.~\ref{fig: feature map}, while red boxes indicate, in each row, the largest off-diagonal entries whose number matches the number of grid neighbors of the corresponding patch. Panel~(b) shows the grouped summary over same-row, same-column, and far patch pairs.}
    \label{fig:cifar_patch_cmi}
\end{figure}

Since all patches are extracted from the same natural image, nearby patches may remain statistically dependent even after conditioning on the class label. To quantify this effect, we estimate the class-conditional mutual information for each label value and then average according to the class probabilities:
\begin{equation}\label{eq:cifar_cmi}
    I_{\mathrm{CMI}}(h_k,h_\ell)\triangleq\sum_{\gamma\in\Gamma}\P(\bgamma=\gamma)\,I(h_k;h_\ell\givensmall\bgamma=\gamma)
\end{equation}
where $h_k$ denotes the raw pixel vector of the patch observed by agent $k$, $\bgamma$ denotes the class label, and $I(h_k;h_\ell\givensmall\bgamma=\gamma)$ is the mutual information under the conditional distribution given $\bgamma=\gamma$. The reported estimates are computed by first reducing the patch dimension through principal component analysis (PCA)~\cite{sayed2022inference} and then applying a $k$-nearest-neighbor (kNN) mutual information estimator~\cite{kraskov2004estimating}. More details about this estimation process are provided in Appendix~\ref{app:cifar_cmi_details}.

Figure~\ref{fig:cifar_patch_cmi} shows that the estimated conditional mutual information is clearly nonzero for many patch pairs and tends to be strongest among spatially related patches. Thus, the CIFAR-10 patch-partition benchmark induces a structured observation model in which nearby local views exhibit stronger statistical dependence than distant ones. This observation motivates the dependence-aware prediction analysis in Section~\ref{sec: main results}. While the CMI diagnostic is not an estimate of the ATE coefficient in Assumption~\ref{assump: class conditional ATE}, it does provide empirical evidence that the local views are not independent in this benchmark.

\subsubsection{Heterogeneous local model families}
\label{sec:cifar_hetero}

We next relax the assumption that all agents use the same local architecture. Different agents are then assigned different local model families, while the communication graph and the collaboration rule remain unchanged. The local models include the patch-level convolutional network used in the main CIFAR-10 benchmark, a logistic-regression classifier on raw patch pixels, and a small residual convolutional network~\cite{he2016deep}. The detailed protocol and the precise model specifications are given in Appendix~\ref{app:cifar_hetero_details}.

Figure~\ref{fig:cifar_patch_hetero} reports a balanced assignment study in which the model families are rotated across the $3\times3$ grid so that each family appears equally often at every spatial location. The results indicate that collaboration remains beneficial across all three model families and across all patch locations, although the magnitude of the improvement varies with both the family and the position. Furthermore, under one representative heterogeneous assignment, the main performance results summarized in Fig.~\ref{fig:cifar_patch_hetero_main} show that the same qualitative behavior observed in the homogeneous benchmark is preserved here: the empirical decision margins become more separated as $N_0$ increases, the collaborative classifier improves over non-cooperative inference, and most of the performance gain is realized within the first few communication rounds.

\begin{figure*}[t]
  \centering
  \includegraphics[width=0.85\linewidth]{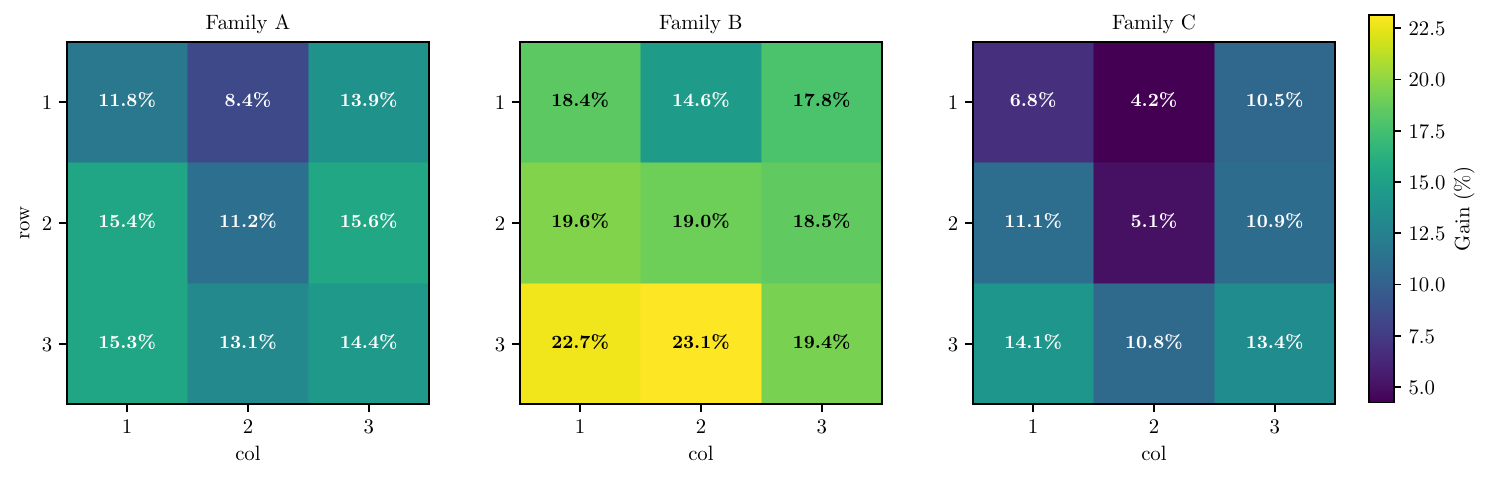}
  \caption{Collaboration gain under heterogeneous local model families on CIFAR-10. Each cell shows the relative gain when the corresponding model family is assigned to that patch location under the balanced rotation protocol.}
  \label{fig:cifar_patch_hetero}
\end{figure*}

\begin{figure*}[t]
    \centering
    \subfloat[Decision margin]{\includegraphics[width=0.33\linewidth]{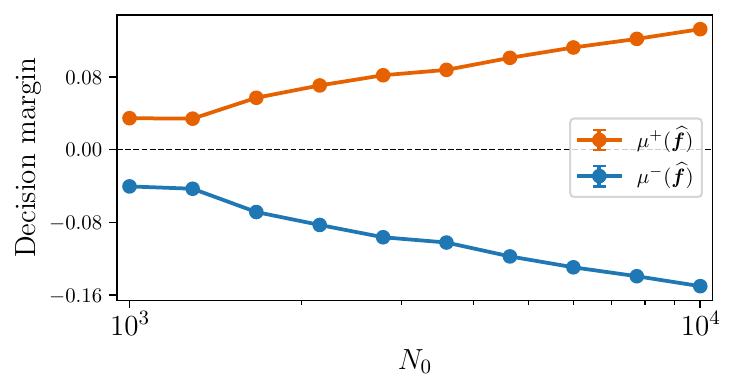}\label{fig:cifar_patch_hetero_margin}}
    \hfill
    \subfloat[Probability of error]{\includegraphics[width=0.33\linewidth]{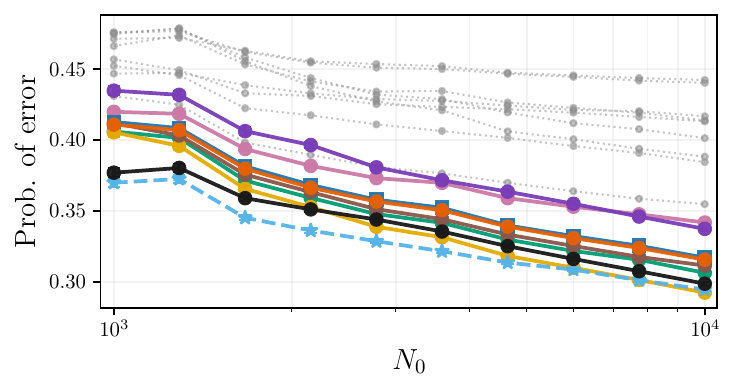}\label{fig:cifar_patch_hetero_pe}}
    \hfill
    \subfloat[Instantaneous probability of error]{\includegraphics[width=0.33\linewidth]{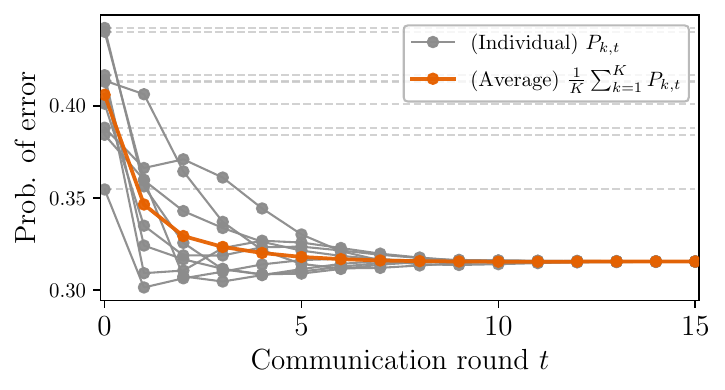}\label{fig:cifar_patch_hetero_process}}
    \caption{Performance of heterogeneous local model families on CIFAR-10. The curve labels in panel~(b) and the value of $N_0$ in panel~(c) are the same as those in Fig.~\ref{fig:cifar_patch_main}.}
    \label{fig:cifar_patch_hetero_main}
\end{figure*}

\subsection{ModelNet40 multi-view benchmark}
\label{subsec:modelnet_real}

In this section, we consider a multi-view benchmark derived from ModelNet40~\cite{wu20153d}. Unlike the CIFAR-10 patch-partition setting, where the agents observe different \emph{spatial} regions of the same image, here each agent receives one rendered \emph{view} of the same 3D object. The heterogeneity across agents therefore arises from viewpoint rather than patch location.

\begin{figure}[h]
    \centering
    \subfloat[]{\includegraphics[width=0.42\linewidth]{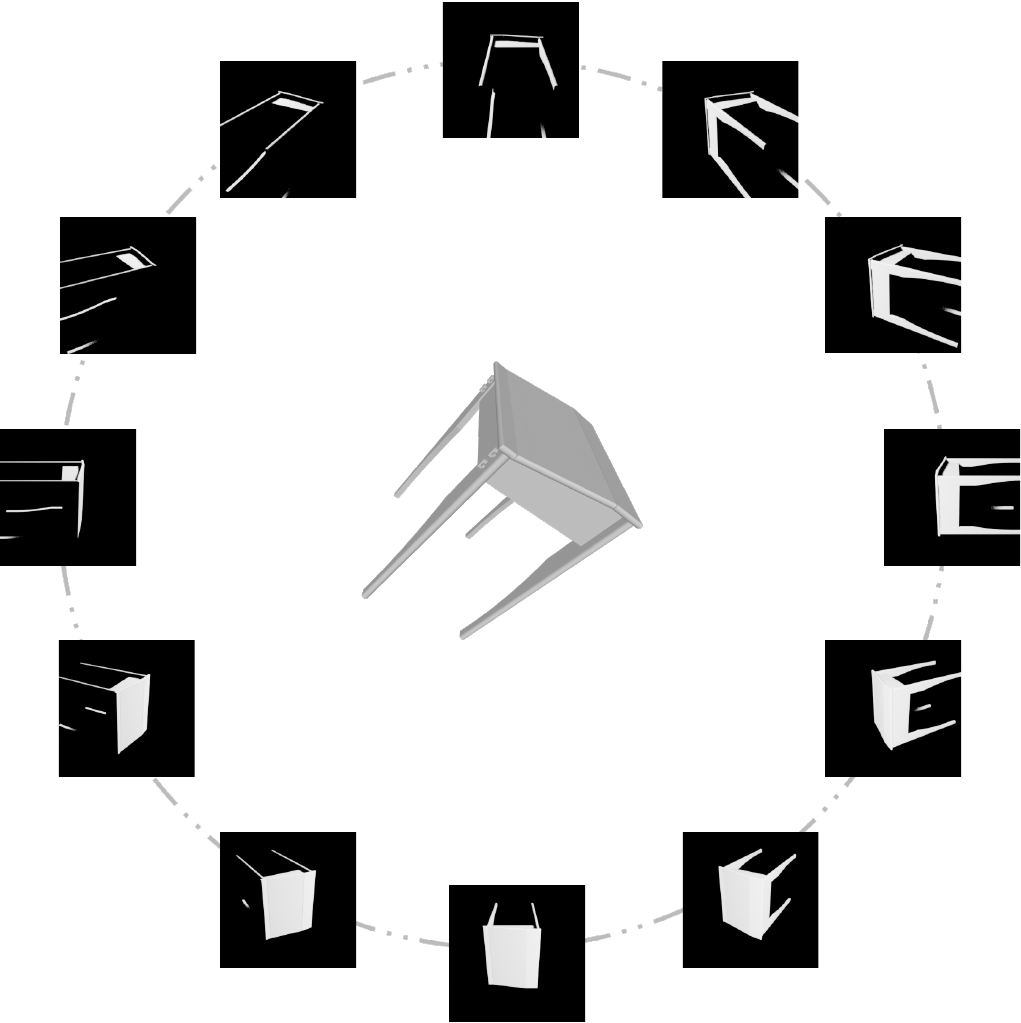}\label{fig:modelnet_views}}
    \hfil
    \subfloat[]{\includegraphics[width=0.49\linewidth]{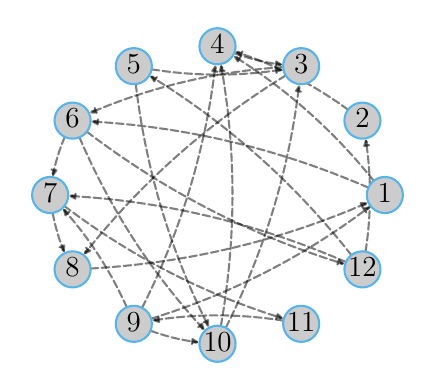}\label{fig:modelnet_net}}
    \caption{ModelNet40 multi-view benchmark. Panel~(a) shows the 12 rendered views associated with a nightstand object, and panel~(b) shows the communication topology used for collaboration.}
    \label{fig:modelnet_setup}
\end{figure}

For each object, we generate $K=12$ RGB renderings from fixed camera positions whose azimuth angles are uniformly spaced around the object, with a common elevation angle of $30^\circ$. One rendered view is then assigned to each agent. In the experiments, we consider the binary task pair of \emph{dresser} and \emph{nightstand}. The multi-view construction and a randomly-generated communication topology used in our experiments are illustrated in Fig.~\ref{fig:modelnet_setup}. Each local classifier is a lightweight convolutional network, referred to as \emph{MicroCNN}; further experimental details are given in Appendix~\ref{app:modelnet_details}.

\begin{figure}[!h]
\centering
\subfloat[Decision margin]{\includegraphics[width=0.75\linewidth]{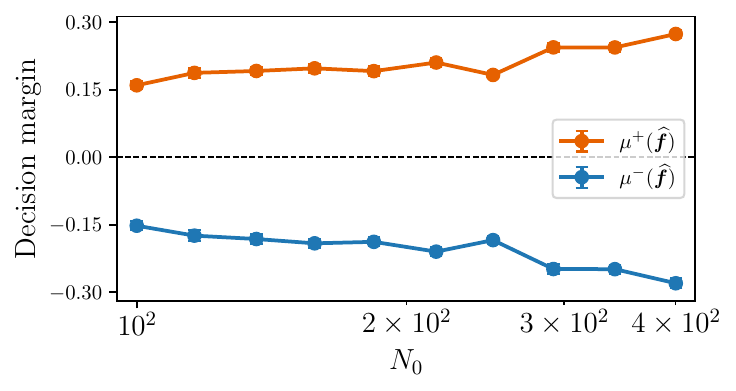}\label{fig:modelnet_margin}}
\hfill
\subfloat[Probability of error]{\includegraphics[width=0.75\linewidth]{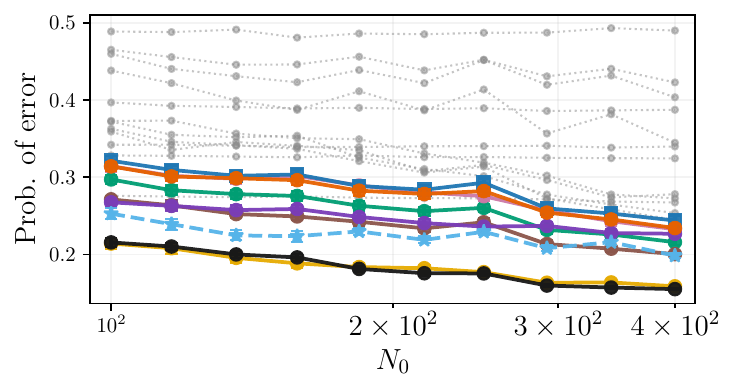}\label{fig:modelnet_pe}}
\hfill
\subfloat[Instantaneous probability of error]{\includegraphics[width=0.75\linewidth]{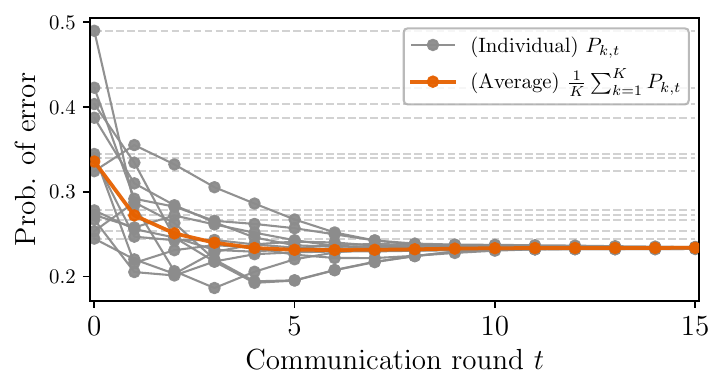}\label{fig:modelnet_process}}
\caption{Main ModelNet40 multi-view results. Panel~(a) reports the empirical class-conditional decision statistics, panel~(b) compares the probability of error under different values of $N_0$, and panel~(c) shows the error evolution over the communication rounds for $N_0=400$.}
\label{fig:modelnet_main}
\end{figure}

Figure~\ref{fig:modelnet_main} summarizes the main performance results for this ModelNet40 benchmark. In panel~(a), the empirical statistics $\bmu^{+}(\widehat{\f})$ and $\bmu^{-}(\widehat{\f})$ remain separated across  $N_0$, indicating that the classifier network remains informative under the viewpoint heterogeneity. Compared with the CIFAR-10 benchmark in Fig.~\ref{fig:cifar_patch_main}, their variation with $N_0$ is milder and exhibits more visible fluctuations, which is consistent with the smaller amount of training data available in this benchmark. The corresponding probability of error curves in Fig.~\ref{fig:modelnet_pe} also decrease overall as $N_0$ increases, in agreement with the larger decision margins attained at larger values of $N_0$. Meanwhile, the improvement of the proposed method is more limited at smaller $N_0$, where some local training sets may be too small for the agents to learn informative classifiers from their views. In this regime, some non-cooperative agents can achieve lower error when acting alone than under collaboration. The relative ordering of the score-fusion baselines is also more sensitive than in the CIFAR-10 benchmark, which may reflect the fact that classifiers trained on different viewpoints need not produce scores on comparable scales. In particular, since \emph{Avg-stat} corresponds to a special case of the proposed method under sufficient communication when $A$ is doubly-stochastic, the gap between \emph{Avg-stat} and \emph{Ours} indicates that the choice of the combination policy $A$ can have a visible effect on collaborative performance. The instantaneous error curves in Fig.~\ref{fig:modelnet_process} further show that most of the performance gain is achieved within the first few communication rounds.

\subsubsection{Temperature scaling and score calibration}\label{sec:temp_scale}
The sensitivity of the score fusion baselines in Fig.~\ref{fig:modelnet_main} suggests that score calibration across viewpoints deserves closer examination. To this end, we consider temperature scaling as a post-hoc calibration step in our experiments~\cite{guo2017calibration}. For each agent, a positive temperature is fitted on the local validation set by minimizing the cross-entropy loss of the rescaled logits, as detailed in Appendix~\ref{app:modelnet_temp_details}. The local decision statistic is then formed from the difference between the two calibrated logits, followed by the same centering step~\eqref{eq: trained classifier} as in the uncalibrated setting. Since temperature scaling does not change the predicted class of an individual local classifier, its role here is not to alter the local decisions themselves, but to improve the comparability of the real-valued decision statistics exchanged during collaboration.

\begin{figure*}[t]
\centering
\subfloat[Decision margin]{\includegraphics[width=0.33\linewidth]{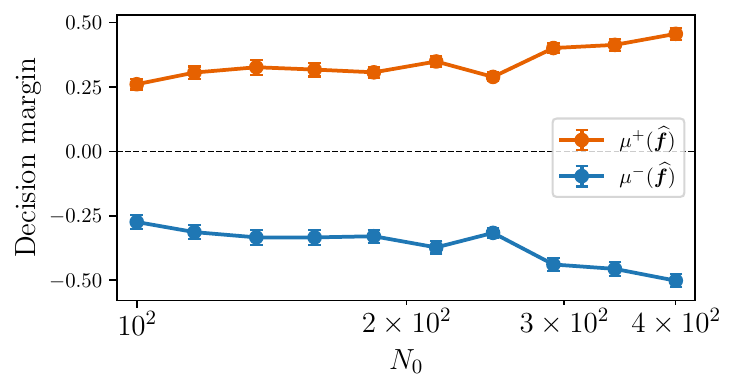}\label{fig:modelnet_margin_temp}}
\hfill
\subfloat[Probability of error]{\includegraphics[width=0.33\linewidth]{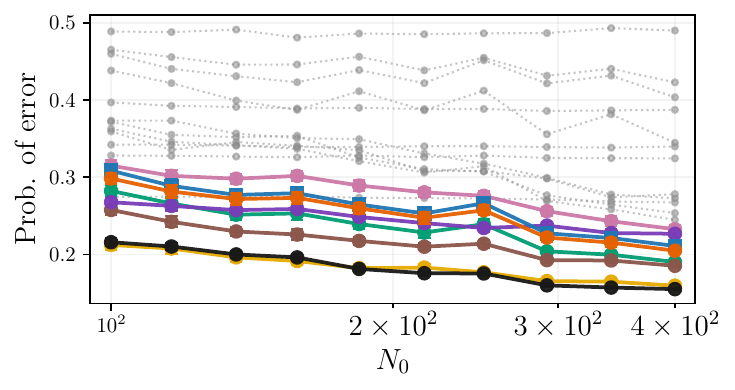}\label{fig:modelnet_pe_temp}}
\hfill
\subfloat[Instantaneous probability of error]{\includegraphics[width=0.33\linewidth]{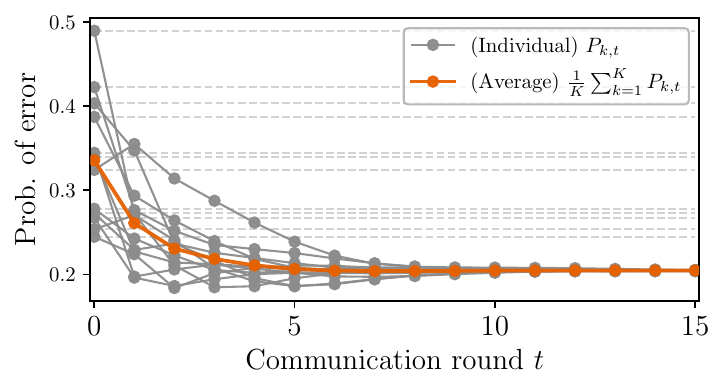}\label{fig:modelnet_process_temp}}
\caption{Performance of temperature scaling on ModelNet40. The curve labels in panel~(b) and the value of $N_0$ in panel~(c) are the same as those in Fig.~\ref{fig:modelnet_main}.}
\label{fig:modelnet_temp}
\end{figure*}

Figure~\ref{fig:modelnet_temp} shows that, relative to Fig.~\ref{fig:modelnet_main}, the class-conditional decision statistics become more clearly separated after calibration, and the probability of error curve of the proposed method is slightly improved over most of the tested range of $N_0$. This suggests that, under viewpoint heterogeneity, score calibration can improve the alignment of the exchanged decision statistics across agents and thereby improve collaborative prediction. VFL-JT is not included in this study because temperature scaling of its local scores is coupled with the jointly optimized fusion head. Additional figures are provided in Appendix~\ref{app:modelnet_temp_details}.

\subsubsection{Controlled correlation stress test}\label{sec:corr_stress_test}
In the multi-view setting, different agents observe the same underlying object and can therefore produce statistically dependent local decision statistics. To examine how the collaborative classifier behaves as this dependence becomes stronger, we perform a controlled stress test directly on the local decision statistics, with more details provided in Appendix~\ref{app:modelnet_corr_details}. Let $x_{i,k}$ denote the clean local decision statistic produced by agent $k$ for the $i$-th testing sample. To control the correlation between $x_{i,k}$ across agents, we consider a perturbed local decision statistic at each agent, denoted by $x'_{i,k}$:
\begin{equation}\label{eq:corr_stress}
x'_{i,k}=x_{i,k}+\sigma_{\mathrm{pert}}\!\left(\omega_c(r)\,u_i + \omega_i(r)\,e_{i,k}\right)
\end{equation}
where $u_i\sim\mathcal N(0,1)$ is a shared standard Gaussian perturbation for sample $i$, $e_{i,k}\sim\mathcal N(0,1)$ are independent standard Gaussian perturbations across samples and agents, and
\begin{equation}\label{eq:corr_weights}
\omega_c(r)=\frac{r}{\sqrt{1+r^2}},\quad\omega_i(r)=\frac{1}{\sqrt{1+r^2}}.
\end{equation}
Here, $r\geq 0$ controls the relative strength of the common-mode component, while $\sigma_{\mathrm{pert}}>0$ controls the overall perturbation scale. Since $\omega_c(r)^2+\omega_i(r)^2=1$, varying $r$ changes the balance between common and idiosyncratic perturbations without substantially changing the total injected variance. This construction follows a standard common-factor decomposition with shared and idiosyncratic components. At test time, the agents employ the same rule~\eqref{eq: distributed learning rule}, but with the perturbed local decision statistics $x'_{i,k}$ in place of the clean statistics $x_{i,k}$.

\begin{figure*}[t]
\centering
\subfloat[Average pairwise correlation]{\includegraphics[width=0.48\linewidth]{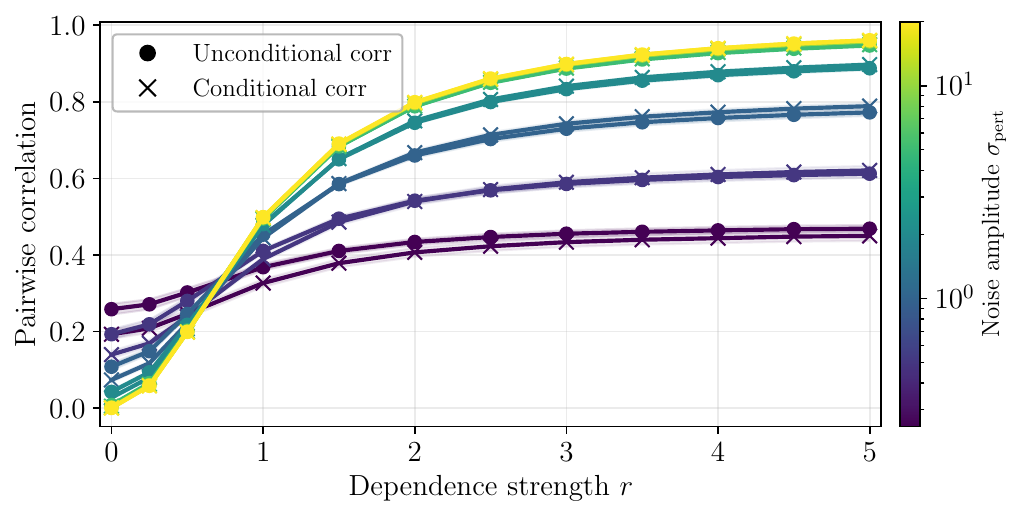}\label{fig:modelnet_corr_vs_alpha}}
\hfill
\subfloat[Collaboration gain]{\includegraphics[width=0.46\linewidth]{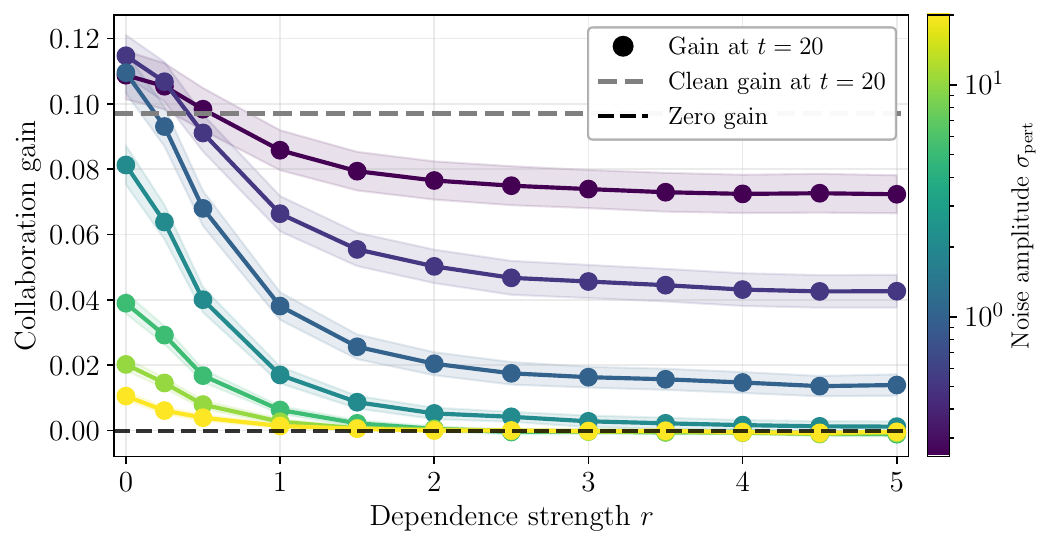}\label{fig:modelnet_gain_vs_alpha}}
\caption{Controlled correlation stress test on ModelNet40. Panel~(a) shows the average pairwise Pearson correlation across agent decision statistics, both unconditional and label-conditional, and panel~(b) shows the collaboration gain after $t=20$ communication rounds.}
\label{fig:modelnet_corr_stress}
\end{figure*}

Figure~\ref{fig:modelnet_corr_stress} shows that the average pairwise correlations across agents, both unconditional and label-conditional, increase with $r$. The corresponding collaboration gain, which is measured by the error reduction relative to non-cooperative inference after $t=20$ communication rounds, decreases with $r$ for each fixed perturbation amplitude. This is consistent with the fact that stronger dependence makes the local decision statistics more redundant and therefore reduces the benefit of information sharing across the network. At the same time, the degradation is gradual, indicating that the collaborative classifier continues to provide measurable gains over a broad range of dependence levels.

\subsection{Communication design, adaptive stopping, and robustness}
\label{subsec:comm_robust}

In this section, we examine several aspects of the communication protocol using the ModelNet40 benchmark. In particular, we study how collaborative performance is influenced by the graph topology, the combination policy, the communication precision, adaptive stopping rules, and corrupted messages.

\subsubsection{Topology, combination policy, and communication precision}
\label{sec:topology_sweep}

We begin by examining the impact of the communication design across three graph topologies: ring, grid, and directed Erd\"{o}s--R\'{e}nyi. For each topology, the combination policy $A$ is constructed using two different schemes: the uniform averaging rule and the Metropolis rule. For this quantization study only, we represent the already trained probabilistic classifier by the bounded posterior-probability difference
\begin{equation}\label{eq: bounded score transform theory}
    g_k(h)\triangleq\widehat p_k(+1|h)-\widehat p_k(-1|h)=\tanh\!\left(\frac{\widehat{\f}_k(h)}{2}\right)\in[-1,1].
\end{equation}
No retraining process is required for this representation. Under this transformation, the resulting centered local score
\begin{equation}
    c_k^g(h) \triangleq g_k(h)-\frac{1}{N_k}\sum_{n=1}^{N_k}g_k(h_{k,n})
\end{equation}
satisfies $|c_k^g(h)|\leq 2$. We compare the resulting matched full-precision recursion with 4-, 6-, and 8-bit stochastic communication. Further details of the quantization protocol are provided in Appendix~\ref{app:topology_details}.

\begin{figure}
    \centering
    \includegraphics[width=0.9\linewidth]{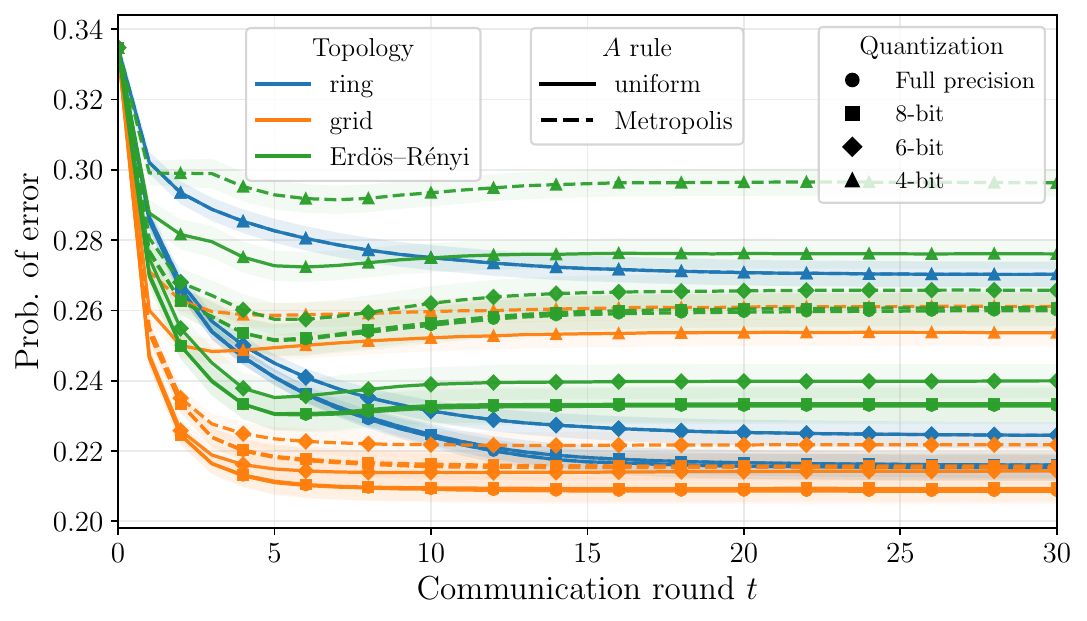}
    \caption{Effect of topology, the construction rule of $A$, and quantization noise.}
    \label{fig:topology_sweep_main}
\end{figure}

Figure~\ref{fig:topology_sweep_main} reports the evolution of the average probability of error under these communication designs. All settings benefit from communication in the early rounds, but both the finite-round improvement and the error level eventually reached depend visibly on the topology, the construction rule of $A$, and the communication precision. Among the graph families considered here, the grid generally attains the lowest error, while the ring also performs competitively. By contrast, the directed Erd\"{o}s--R\'{e}nyi topology is more sensitive to the choice of combination policy and communication precision. The precision hierarchy is systematic: 8-bit stochastic communication is practically indistinguishable from the matched bounded full-precision reference, 6-bit communication produces a modest intermediate degradation, and 4-bit communication produces the largest loss. More detailed grouped views of these results are provided in Appendix~\ref{app:topology_details}.

\subsubsection{Adaptive stopping rules}
\label{sec:stopping_rule}

In the previous experiments, the collaboration framework is evaluated under a fixed communication budget $t$ chosen in advance. In practice, however, selecting this budget is not straightforward: a larger $t$ increases communication cost, while a smaller $t$ may be insufficient for collaboration. This motivates the study of adaptive stopping rules, in which the communication pattern is allowed to depend on the evolution of local decision statistics. We consider two local event-triggered families: a \emph{$\Delta$-trigger} rule based on the change in the communicated statistic, and a \emph{label-stability with confidence} rule based on the local prediction and a confidence threshold (see Appendix~\ref{app:stopping_details}). Since transmissions may stop at different times for different agents and different samples, we measure the resulting communication cost by the total number of transmissions per testing sample.

\begin{figure}[!]
    \centering
    \subfloat[$t_{\max}=2$]{\includegraphics[width=0.49\linewidth]{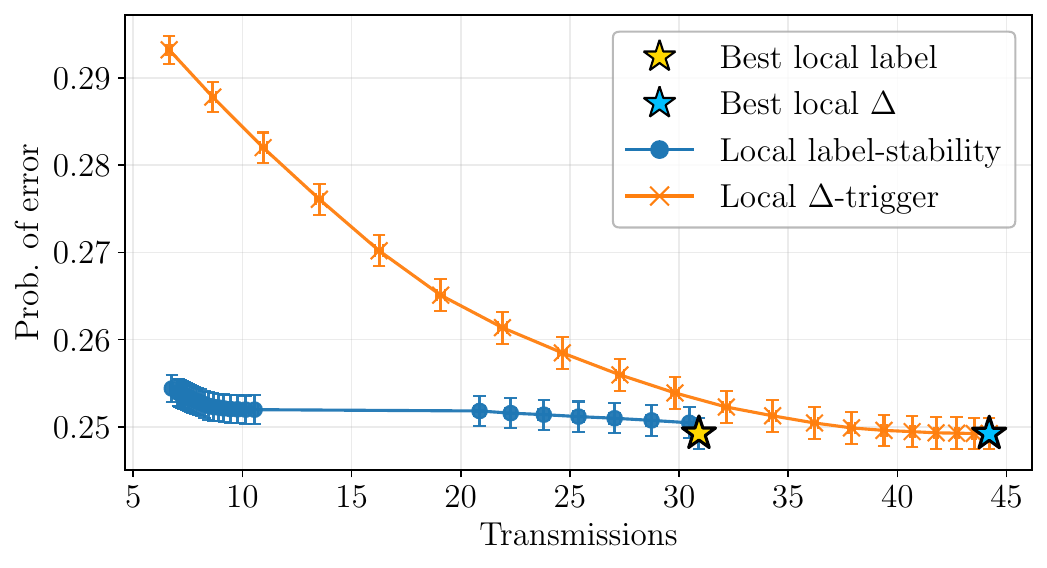}}
    \hfill
    \subfloat[$t_{\max}=5$]{\includegraphics[width=0.49\linewidth]{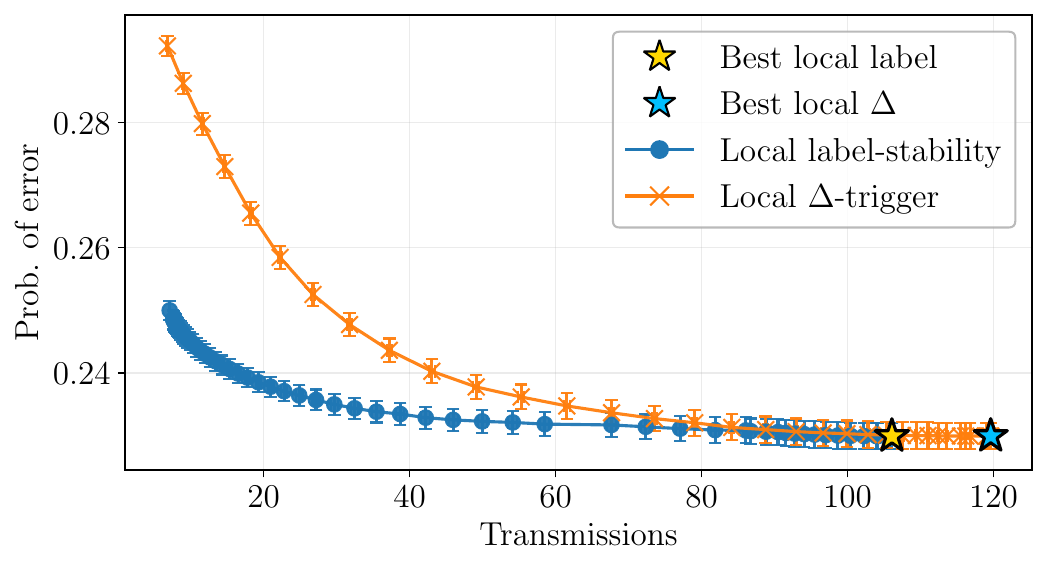}}
    \hfill
    \subfloat[$t_{\max}=10$]{\includegraphics[width=0.49\linewidth]{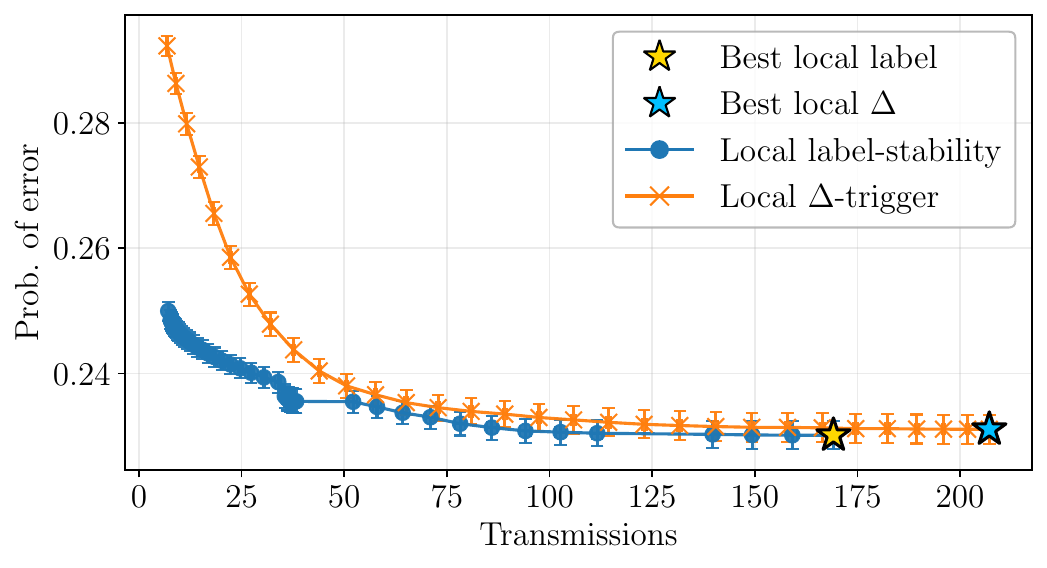}}
    \hfill
    \subfloat[$t_{\max}=20$]{\includegraphics[width=0.49\linewidth]{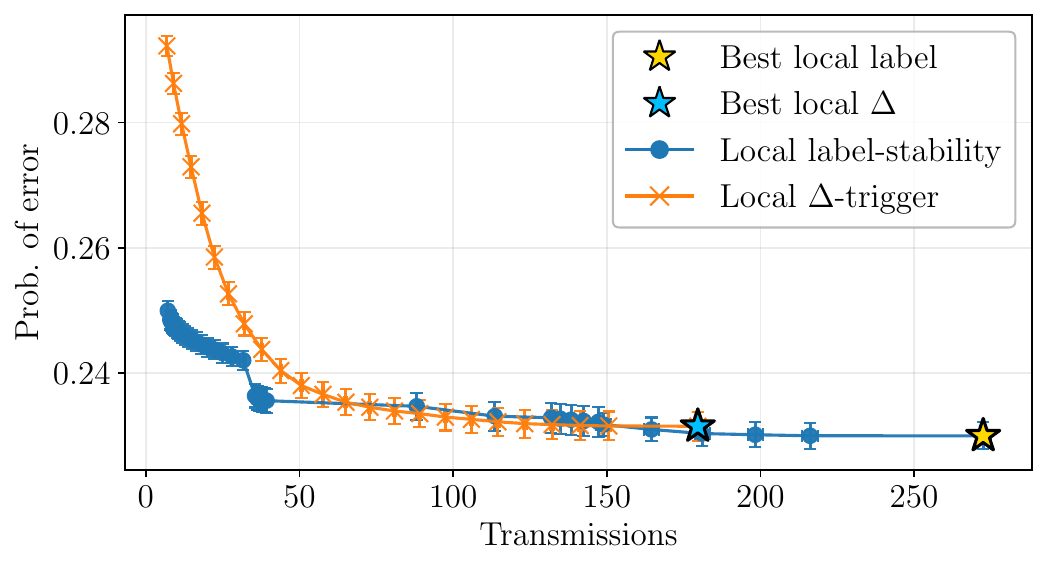}}
    \caption{Tradeoff between the probability of error and the total number of transmissions (per testing sample) under local adaptive stopping rules.}
    \label{fig:stopping_frontier}
\end{figure}

Figure~\ref{fig:stopping_frontier} illustrates the tradeoff between the probability of error and the total communication cost for four values of the maximum communication budget $t_{\max}$. In each panel, only the Pareto-efficient operating points are shown, namely those for which no alternative configuration attains both lower communication cost and lower probability of error. Both stopping rule families exhibit a clear communication--accuracy tradeoff, but the label-stability rule generally yields a better Pareto frontier than the $\Delta$-trigger rule, especially in the low communication regime. The star markers denote the best points selected from the two families of stopping rules. Their positions indicate that communication can often be terminated well before reaching the maximum communication budget.

\subsubsection{Robustness to noisy, faulty, and adversarial agents}
\label{sec:robustness}

We next examine how the collaborative inference mechanism behaves when a subset of agents transmits corrupted messages at test time. The perturbation is introduced only at the level of the local decision statistics before communication begins. For each corruption scenario, we compare the resulting probability of error with that of the corresponding clean run obtained from the same trained models, the same testing set, and the same combination policy. The degradation is measured by the \emph{paired excess error}, namely the increase in probability of error relative to the clean run within the same Monte Carlo repetition. In each repetition, we corrupt exactly $m\in\{1,\dots,K\}$ agents and consider three families of perturbations: benign Gaussian noise, faulty agents (including stuck-at-zero and constant-bias failures), and adversarial reports modeled by scaled sign flips. Additional details on the perturbation models are provided in Appendix~\ref{app:robust_details}.

\begin{figure*}[!t]
    \centering
    \subfloat[Benign Gaussian noise]{\includegraphics[width=0.33\linewidth]{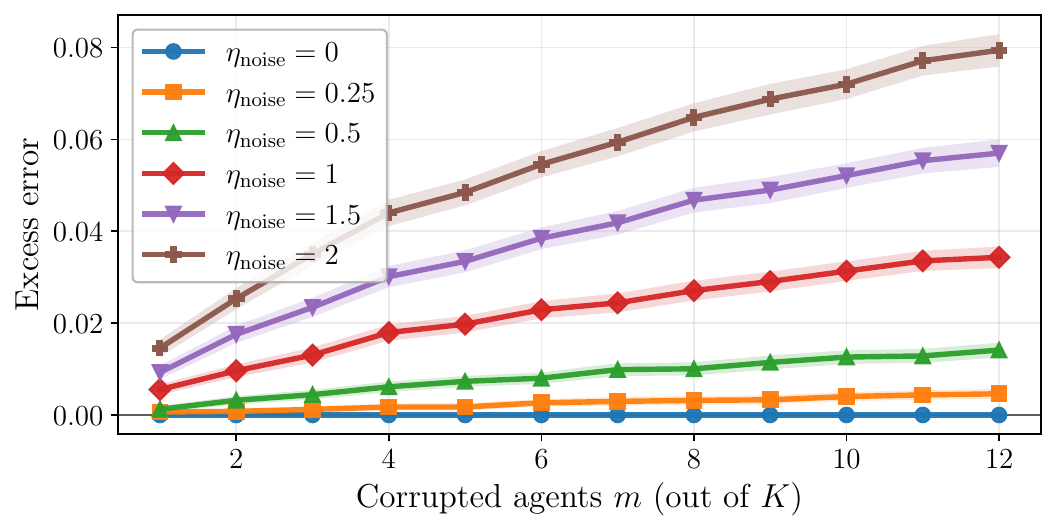}}
    \hfill
    \subfloat[Faulty agents]{\includegraphics[width=0.33\linewidth]{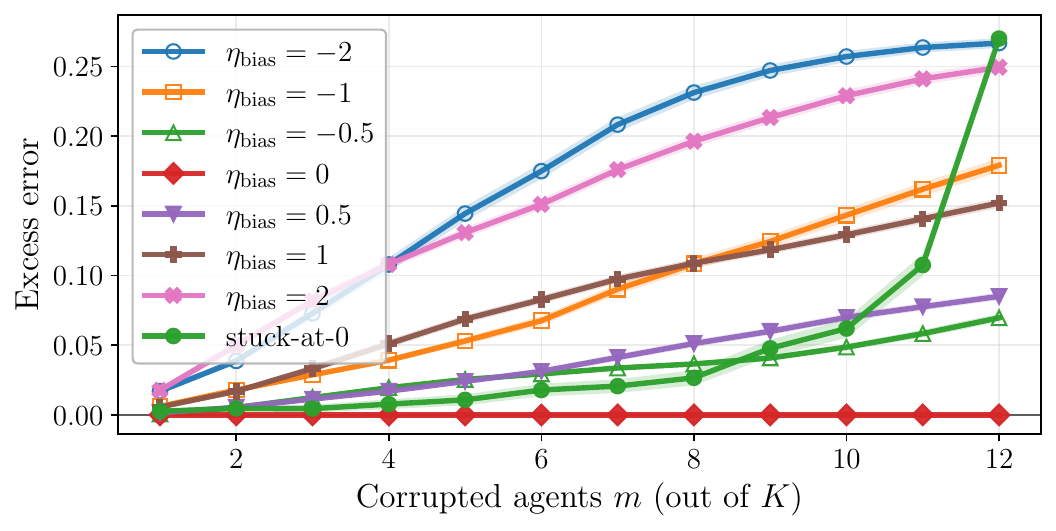}}
    \hfill
    \subfloat[Adversarial reports]{\includegraphics[width=0.33\linewidth]{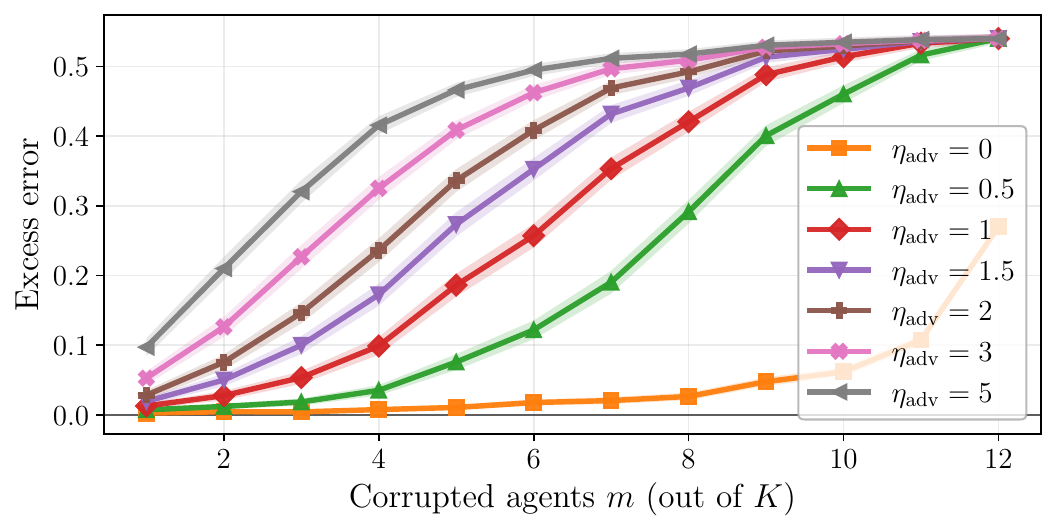}}
    \caption{Robustness of collaborative classification to corrupted agents, measured by the paired excess error at the communication round $t=20$. Panel~(a) varies the Gaussian noise multiplier $\eta_{\mathrm{noise}}$, panel~(b) varies the fault type and the normalized bias level $\eta_{\mathrm{bias}}$, and panel~(c) varies the adversarial scale $\eta_{\mathrm{adv}}$.}
    \label{fig:robust_agents_main}
\end{figure*}

Figure~\ref{fig:robust_agents_main} summarizes the robustness results after 20 rounds of communication. In all three perturbation settings, the excess error increases with both the number of corrupted agents and the perturbation severity. Among the three cases, adversarial scaled sign flips are the most damaging, indicating that persistently misleading decision statistics can have a particularly strong effect on collaborative prediction. By contrast, benign Gaussian noise and faulty agents yield a more gradual degradation, especially when only a few agents are affected.

\section{Concluding Remarks}
\label{sec: conclusions}

In this work, we proposed a distributed test-time collaborative classification framework for multi-agent networks. We considered a group of independently trained agents, potentially with heterogeneous feature spaces and model architectures, that collaborates at inference time by exchanging local beliefs with graph neighbors to produce collective predictions. Our analysis established classification error bounds that characterize the impact of the agent heterogeneity, network topology, combination policies, and communication budgets on collective prediction performance. These bounds provide theoretical insight into the interplay between local training and distributed inference in multi-agent systems, in addition to highlighting the benefits of collaborative classification relative to the non-cooperative scenario. Importantly, this study shows that cooperation can be valuable at test time even when the local models are trained independently, which is becoming an increasingly important paradigm in real-world implementations of distributed machine learning. Future directions include extensions to multi-class settings, theoretical analysis of adaptive and communication-efficient protocols, privacy-preserving implementations, and more general communication settings, including asynchronous updates, time-varying topologies, and lossy communication.

\bibliographystyle{IEEEtran}
\bibliography{refs}

\begin{IEEEbiographynophoto}{Ping Hu} received her Ph.D. in electrical engineering from EPFL, Switzerland in 2026. She is currently a postdoctoral researcher with the Adaptive Systems Laboratory, EPFL. Her research interests include distributed learning, inference and decision-making over networked systems.
\end{IEEEbiographynophoto}

\begin{IEEEbiographynophoto}{Mert Kayaalp} received his Ph.D. in computer and communication sciences from EPFL, Switzerland in 2025. He is currently a research scientist at Dalle Molle Institute for Artificial Intelligence (IDSIA USI-SUPSI), Switzerland, where he is affiliated with the UBS-IDSIA AI Lab. His research focuses on machine learning and signal analysis in networked and dynamical systems.
\end{IEEEbiographynophoto}

\begin{IEEEbiographynophoto}{Ali H. Sayed} is professor and former dean of engineering at EPFL, Switzerland, where he also directs the Adaptive Systems Laboratory. He served before as a distinguished professor and chair of electrical engineering with UCLA. He is a member of the US National Academy of Engineering (NAE) and The World Academy of Sciences (TWAS). He served as a president of the IEEE Signal Processing Society in 2018 and 2019. An author of over 650 scholarly publications and 10 books, his research involves several areas including adaptation and learning theories, statistical inference, and multi-agent systems. His work has been recognized with several awards including the 2022 IEEE Fourier Technical Field Award and the 2020 IEEE Wiener Society Award. He is also a fellow of EURASIP and the American Association for the Advancement of Science (AAAS).
\end{IEEEbiographynophoto}

\end{document}


\maketitle
\section*{Organization of the Supplementary Material}
This supplementary material provides detailed proofs, auxiliary results, and additional experimental information supporting the main paper. The theoretical appendices follow the order of the main development. Appendix~\ref{appendix: learning complexity definitions} presents the network margin and training analysis and collects the auxiliary quantities used in the Theorems and subsequent proofs.
Appendices~\ref{appendix: sufficient communication}--%
\ref{appendix: finite bit communication} establish the sufficient-communication, finite-round, and finite-precision classification guarantees, respectively. Appendix~\ref{appendix: PAC generalization bound} provides the Probably Approximately Correct (PAC)-style generalization analysis, while Appendix~\ref{appendix: classification error using SL rule} develops the relation to social learning. Appendix~\ref{appendix: raw score counterparts} collects the corresponding raw score results when the empirical centering step~\eqref{eq: trained classifier} is omitted. Finally, Appendix~\ref{app:exp_details} presents additional experimental details. Cross-references to theorem, equation, section, and figure numbers in the main paper are retained throughout.

\appendices

\renewcommand{\thesubsection}{\Alph{section}.\arabic{subsection}}
\renewcommand{\thesubsectiondis}{\arabic{subsection}.}
\renewcommand{\thesubsubsection}{\Alph{section}.\arabic{subsection}.\arabic{subsubsection}}
\renewcommand{\thesubsubsectiondis}{\arabic{subsubsection})}
\renewcommand{\theHsubsection}{Appendix.\Alph{section}.\arabic{subsection}}
\renewcommand{\theHsubsubsection}{Appendix.\Alph{section}.\arabic{subsection}.\arabic{subsubsection}}

\numberwithin{equation}{section}
\counterwithin{figure}{section}
\counterwithin{table}{section}
\counterwithin{inlemma}{section}
\counterwithin{proposition}{section}
\counterwithin{corollary}{section}

\section{Network Margin and Training Guarantee}
\label{appendix: learning complexity definitions}

This appendix provides the technical details supporting the network margin and training guarantees in Section~\ref{sec: consistent training} of the main paper. We first collect the complexity and margin quantities used in the finite-sample analysis. We then prove the effect of additive centering and the $\delta$-margin training guarantee under approximate optimization. The final subsection compares the resulting bound with homogeneous centralized training.

\subsection{Definitions of Auxiliary Quantities}
\label{appendix: training auxiliary quantities}

This subsection collects the definitions and relevant properties of several auxiliary quantities used throughout the theoretical analysis.

\paragraph{Rademacher complexity}
We begin with introducing the Rademacher complexity of a general function class $\mathcal G$. Let $\mathcal{Z}\triangleq\{z_1,z_2,\ldots,z_m\}$ be a fixed sample set of size $m$. The empirical Rademacher complexity of $\mathcal G$ with respect to $\mathcal{Z}$ is defined as
\begin{equation}
\label{eq: empirical Rademacher complexity}
    \widehat{\mathfrak R}_{\mathcal{Z}}(\mathcal G) \triangleq \E_{\bm r}\sup_{g\in\mathcal G} \left|\frac{1}{m}\sum_{n=1}^{m}\bm r_n g(z_n)\right|
\end{equation}
where $\bm r_1,\ldots,\bm r_m$ are independent Rademacher random variables. For an i.i.d.\ random sample set $\mathcal{Z}'$ of size $m$, the Rademacher complexity of $\mathcal G$ is
\begin{equation}
\label{eq: definition of Rademacher complexity}
    \mathfrak R_m(\mathcal G)\triangleq \E_{\mathcal{Z}'}\widehat{\mathfrak R}_{\mathcal{Z}'}(\mathcal G).
\end{equation}
For the local function class $\F_k$ and training set size $N_k$, we define the individual Rademacher complexity by
\begin{equation}
\label{eq: individual Rademacher complexity}
    \rho_k\triangleq\mathfrak R_{N_k}(\F_k),
\end{equation}
and the corresponding network Rademacher complexity by the weighted average involving the Perron vector $\pi$:
\begin{equation}
\label{eq: network Rademacher complexity}
    \rho\triangleq\sum_{k=1}^K\pi_k\rho_k.
\end{equation}
Here, $\mathfrak R_{N_k}(\F_k)$ is evaluated for $N_k$ i.i.d.\ local feature samples drawn from the feature marginal induced by $Q_k$.

\paragraph{Sample-size factor}
We define two quantities related to the size of local training sets:
\begin{equation}
\label{eq: alpha and N_max}
    N_{\max}\triangleq\max_k N_k, \quad \alpha\triangleq\sum_{k=1}^K\pi_k\frac{N_{\max}}{N_k}.
\end{equation}
The factor $\alpha$ captures the effect of unequal local training set sizes. Since $N_{\max}/N_k\geq1$ and $\sum_k\pi_k=1$, we have $\alpha\geq1$. In particular, $\alpha=1$ when all local training sets have the same size.

\paragraph{Function $\mathdutchcal E_\Phi(r,\delta)$}
Following Eqs.~(A.16), (A.18), and (A.20) of~\cite{hu2025non-asymptotic}, we use the same construction with a generic risk level $r<\Phi(0)$. For any prescribed margin $\delta\geq0$, let $d_\delta^\star(r)$ denote the unique nonnegative solution of
\begin{equation}
\label{eq: root definition E Phi}
    d-\delta-\frac{\Phi(d)-r}{2L_\Phi}=0.
\end{equation}
The corresponding auxiliary function is
\begin{equation}
\label{eq: definition of E_Phi_R_delta}
    \mathdutchcal E_\Phi(r,\delta) \triangleq \frac{d_\delta^\star(r)-\delta}{4}
    = \frac{\Phi(d_\delta^\star(r))-r}{8L_\Phi}.
\end{equation}
For $r=\mathsf R^o$, this reduces to the quantity $\mathdutchcal E_\Phi(\mathsf R^o,\delta)$ introduced in~\cite{hu2025non-asymptotic}. An important property of the quantity $d_\delta^\star(r)$ is that, for fixed $r$, it increases with the prescribed margin $\delta$. Since $\Phi$ is non-increasing under Assumption~\ref{assump: risk function}, $\mathdutchcal E_\Phi(r,\delta)$ decreases as $\delta$ increases. For a fixed feasible $\delta$, the same construction shows that $\mathdutchcal E_\Phi(r,\delta)$ decreases as the risk level $r$ increases.

\paragraph{Maximum feasible margin $\delta_{\max}(r)$}
According to the definition in Eq.~(A.21) of~\cite{hu2025non-asymptotic}, we have
\begin{equation}
    d_r\triangleq \inf\{x\geq0:\Phi(x)\leq r\},
\end{equation}
with $d_r=+\infty$ if the set is empty. The quantity $\delta_{\max}(r)$ is defined as the largest prescribed margin for which the solution $d_\delta^\star(r)$ remains below $d_r$, namely,
\begin{equation}
\label{eq: definition delta max generic}
    \delta_{\max}(r)\triangleq\sup\{\delta\geq0: d_\delta^\star(r)<d_r\}.
\end{equation}
From~\eqref{eq: definition of E_Phi_R_delta}, $d_\delta^\star(r)<d_r$ is precisely the condition that ensures $\mathdutchcal E_\Phi(r,\delta)>0$. Therefore, $0\leq\delta<\delta_{\max}(r)$ specifies the feasible range of prescribed  margins. Since $d_r$ is non-increasing in $r$ under Assumption~\ref{assump: risk function}, this feasible range becomes more restrictive as the risk level $r$ increases. In Proposition~\ref{prop: P_c_delta}, the generic risk level is $r=\mathsf R^o+\varepsilon_{\opt}$.

\subsection{Proof of Proposition~\ref{prop: effect of additive centering}}
\label{appendix: proof additive centering}

\begin{proof}
Since $\bmu_{\rm raw}^{+}=\mu_{\rm mid}+D$ and $\bmu_{\rm raw}^{-}=\mu_{\rm mid}-D$, the two shifted expected margins in~\eqref{eq: margin under additive shift} are
\begin{equation}
    D+(\mu_{\rm mid}-\chi_\pi) \quad \text{and} \quad D-(\mu_{\rm mid}-\chi_\pi).
\end{equation}
Their minimum is $D-|\mu_{\rm mid}-\chi_\pi|$, which proves \eqref{eq: margin under additive shift}. Since the absolute value term is minimized at zero, $\Delta_\chi(f)$ is maximized when $\chi_\pi=\mu_{\rm mid}$. Setting respectively $\chi_\pi=0$ and $\chi_\pi=\bmu_{\emp}(f)$ gives the two identities in~\eqref{eq: raw margin}, and comparing them yields the condition~\eqref{eq: centering improvement condition}. Finally, $\bmu_{k,\emp}(f'_k)=\bmu_{k,\emp}(f_k)+o_k$, and hence \eqref{eq: empirical centering invariance} follows directly. Therefore, the centered local scores are identical before and after the additive offsets, and substituting them into \eqref{eq: distributed learning rule} gives the same DeGroot iterates at every communication round.
\end{proof}

\noindent\emph{Remarks on centering.}
The achieved two-sided expected margin $\Delta_\chi(f)$ should be distinguished from the prescribed target $\delta\geq0$ used in the subsequent guarantees. In particular, $\Delta_\chi(f)$ may be negative. An additive shift changes only the location of the two class means relative to threshold zero and leaves their separation $2D(f)$ unchanged. Hence, if $D(f)\leq0$, no additive center can produce a positive two-sided expected margin. The proposition concerns this expected margin criterion and does not imply that centering universally improves classification accuracy. For example, if the score is already the exact equal-prior log-likelihood ratio, zero is the Bayes threshold and a nonzero shift can worsen the resulting decision rule.

\subsection{Proof of Proposition~\ref{prop: P_c_delta}}
\label{appendix: proof approximate ERM lemma}

First, we recall the $\delta$-margin consistent training event associated with condition~\eqref{eq: delta-margin consistent training condition}:
\begin{equation}
    \C_\delta \triangleq \left\{\bmu^+(\widehat{\f})>\delta,\,\bmu^-(\widehat{\f})<-\delta\right\},
\end{equation}
and $P_{c,\delta}\triangleq\P(\C_\delta)$. The proof extends the exact-ERM analysis of \cite{hu2025non-asymptotic}, which in turn uses uniform deviation bounds established in~\cite{bordignon2023learning}. We retain the uniform concentration and margin-to-risk arguments from that analysis, while accounting explicitly for the nonzero empirical optimization gaps of the classifiers returned by the local training algorithms from~\eqref{eq: actual trained f_k}. The key new step is to show that these gaps replace the target risk $\mathsf R^o$ in the exact-ERM argument by the effective risk level $\mathsf R^o+\varepsilon_{\opt}$.

\begin{proof}
\medskip
\noindent\emph{Uniform deviations.}
We introduce the empirical and population surrogate risks for the network as
\begin{align}
\label{eq: network empirical risk appendix}
    \R_{\emp}(f) &\triangleq \sum_{k=1}^K\pi_k\R_{k,\emp}(f_k),\\
    \R(f) &\triangleq \sum_{k=1}^K\pi_k \E\Phi(\bgamma f_k(\h_k)).
    \label{eq: network population risk appendix}
\end{align}
Since $\F=\F_1\times\cdots\times\F_K$ and the population risk separates across agents,
\begin{equation}
    \inf_{f\in\F}\R(f) = \sum_{k=1}^K\pi_k\inf_{f_k\in\F_k}\E_{(\h_k,\bgamma)\sim Q_k}\Phi\!\left(\bgamma f_k(\h_k)\right) =\mathsf R^o
\end{equation}
where $\mathsf R^o$ is defined in~\eqref{eq: expected target risk}. Let
\begin{align}
\label{eq: uniform risk deviation appendix}
    U &\triangleq \sup_{f\in\F}|\R_{\emp}(f)-\R(f)|, \\
    V_\mu &\triangleq \sup_{f\in\F}|\bmu_{\emp}(f)-\mu_{\rm mid}(f)|.
\end{align}
The following two uniform deviation bounds are the concentration results used in the exact-ERM analysis of \cite{hu2025non-asymptotic}, recalled there from Theorem~3 of~\cite{bordignon2023learning}. For every $x>4\rho$,
\begin{equation}
\label{eq: uniform centering bound appendix}
    \P(V_\mu\geq x) \leq \exp\left\{ -\frac{N_{\max}(x-4\rho)^2}{2\alpha^2\beta^2} \right\},
\end{equation}
and, for every $x>4L_\Phi\rho$,
\begin{equation}
\label{eq: uniform risk bound appendix}
    \P(U\geq x) \leq \exp\left\{ -\frac{N_{\max}(x-4L_\Phi\rho)^2}{2\alpha^2\beta^2L_\Phi^2} \right\}.
\end{equation}
We use these established bounds directly.

\medskip
\noindent\emph{Effect of approximate empirical optimization.}
From the definition in~\eqref{eq: local optimization error}, condition~\eqref{eq: approximate ERM condition} is equivalent to
\begin{equation}
\label{eq: local approximate ERM appendix}
    \R_{k,\emp}(\widehat{\f}_k)\leq \inf_{f_k\in\F_k}\R_{k,\emp}(f_k) +\varepsilon_{k,\opt}
\end{equation}
for every agent $k$. Multiplying by $\pi_k$ and summing gives
\begin{equation}
    \R_{\emp}(\widehat{\f}) \leq \sum_{k=1}^K \pi_k\inf_{f_k\in\F_k}\R_{k,\emp}(f_k) +\varepsilon_{\opt}.
\end{equation}
Because the coordinates $f_k$ can be selected independently over the Cartesian product class $\F$, the weighted sum of the local infima is the infimum of the network empirical risk. Hence,
\begin{equation}
\label{eq: network approximate ERM}
    \R_{\emp}(\widehat{\f})\leq \inf_{f\in\F}\R_{\emp}(f) +\varepsilon_{\opt}.
\end{equation}
The infimum defining $\mathsf R^o$ need not be attained. By the definition of $U$, for every $f\in\F$, $\R_{\emp}(f)\leq \R(f)+U$. Hence, using \eqref{eq: network approximate ERM},
\begin{align}
    \R(\widehat{\f}) &\leq \R_{\emp}(\widehat{\f})+U
    \nonumber\\
    &\leq \inf_{f\in\F}\R_{\emp}(f) +\varepsilon_{\opt}+U
    \nonumber\\
    &\leq \inf_{f\in\F}\bigl\{\R(f)+U\bigr\} +\varepsilon_{\opt}+U
    \nonumber\\
    &=\mathsf R^o+\varepsilon_{\opt}+2U
\end{align}
where the last equality follows because $U$ does not depend on $f$ and $\mathsf R^o=\inf_{f\in\F}\R(f)$. Therefore,
\begin{equation}
\label{eq: approximate ERM excess risk}
    \R(\widehat{\f}) \leq \mathsf R^o+\varepsilon_{\opt}+2U.
\end{equation}
For convenience, define the effective risk level
\begin{equation}
\label{eq: effective risk appendix}
    r \triangleq \mathsf R^o+\varepsilon_{\opt}.
\end{equation}
Thus,
\begin{equation}
    \R(\widehat{\f})\leq r+2U.
\end{equation}
This is the step that differs from the exact-ERM analysis, for which $r=\mathsf R^o$.

\medskip
\noindent\emph{Reduction of margin failure to the deviation events.}
By Proposition~\ref{prop: effect of additive centering},
\begin{equation}
\label{eq: centered margin identity appendix}
    \Delta_{\rm cen}(\widehat{\f}) = D(\widehat{\f}) - \Bigl|\bmu_{\emp}(\widehat{\f}) - \mu_{\rm mid}(\widehat{\f})
    \Bigr|.
\end{equation}
Fix any constant $d>\delta$. If $D(\widehat{\f})>d$ and $V_\mu<d-\delta$ hold simultaneously, then
\begin{equation}
    \Delta_{\rm cen}(\widehat{\f})> d-(d-\delta) = \delta.
\end{equation}
Consequently,
\begin{equation}
\label{eq: centered failure inclusion concise}
    \overline{\C_\delta}\subseteq\!\bigl\{V_\mu\geq d-\delta\bigr\}\cup\bigl\{D(\widehat{\f})\leq d\bigr\}.
\end{equation}
We next use the same surrogate risk reduction as in the proof of the exact-ERM result in~\cite{hu2025non-asymptotic}. Under the uniform class prior, convexity of $\Phi$ and Jensen's inequality give
\begin{align}
    \R(f)&=\frac12\sum_{k=1}^K\pi_k\biggl[\E^{(+1)}\Phi(f_k(\h_k))+\E^{(-1)}\Phi(-f_k(\h_k))\biggr]
    \nonumber\\
    &\geq\frac12\Bigl[\Phi\!\left(\bmu_{\rm raw}^{+}(f)\right)+\Phi\!\left(-\bmu_{\rm raw}^{-}(f)\right)\Bigr]
    \nonumber\\
    &\geq\Phi\!\left(\frac{\bmu_{\rm raw}^{+}(f)-\bmu_{\rm raw}^{-}(f)}{2}\right)
    \nonumber\\
    &=\Phi(D(f)).
\label{eq: risk separation Jensen appendix}
\end{align}
Thus, since $\Phi$ is non-increasing, $D(\widehat{\f})\leq d$ implies $\R(\widehat{\f})\geq\Phi(d)$. Combining this implication with \eqref{eq: approximate ERM excess risk} yields
\begin{equation}
\label{eq: separation failure approximate ERM}
    \bigl\{D(\widehat{\f})\leq d\bigr\} \subseteq \left\{U\geq\frac{\Phi(d)-r}{2}\right\}.
\end{equation}
Using \eqref{eq: centered failure inclusion concise} and \eqref{eq: separation failure approximate ERM}, followed by the established deviation bounds in~\eqref{eq: uniform centering bound appendix}--\eqref{eq: uniform risk bound appendix}, gives
\begin{align}
\label{eq: approximate ERM two-term inherited bound}
    \P(\overline{\C_\delta}) &\leq \exp\Biggl\{-\frac{N_{\max}(d-\delta-4\rho)^2}{2\alpha^2\beta^2}\Biggr\}
    \nonumber\\
    &\quad+\exp\Biggl\{-\frac{N_{\max}}{2\alpha^2\beta^2L_\Phi^2}\left(\frac{\Phi(d)-r}{2}-4L_\Phi\rho\right)^2    \Biggr\},
\end{align}
for every $d>\delta$ satisfying $d-\delta>4\rho$ and $(\Phi(d)-r)/2>4L_\Phi\rho$.

\medskip
\noindent\emph{Balancing the two deviations.}
We now use the same balancing construction as in \cite{hu2025non-asymptotic}, with the exact target risk $\mathsf R^o$ replaced by the effective risk $r$ defined in~\eqref{eq: effective risk appendix}. Specifically, choose $d=d_\delta^\star(r)$ from Appendix~\ref{appendix: training auxiliary quantities}, which satisfies
\begin{equation}
    d-\delta = \frac{\Phi(d)-r}{2L_\Phi}.
\end{equation}
Under the conditions of Proposition~\ref{prop: P_c_delta}, this choice is admissible. By \eqref{eq: definition of E_Phi_R_delta}, the two terms in \eqref{eq: approximate ERM two-term inherited bound} are then equal, and
\begin{equation}
    \P(\overline{\C_\delta}) \leq 2\exp\Biggl\{ -\frac{8N_{\max}}{\alpha^2\beta^2}\Bigl(\mathdutchcal E_\Phi(r,\delta)-\rho \Bigr)^2 \Biggr\}.
\end{equation}
Since $P_{c,\delta}=1-\P(\overline{\C_\delta})$ and $r=\mathsf R^o+\varepsilon_{\opt}$, we obtain~\eqref{eq: P_c_delta}. When $\varepsilon_{\opt}=0$, the effective risk reduces to $r=\mathsf R^o$, which recovers the exact-ERM result of \cite{hu2025non-asymptotic}.
\end{proof}

\noindent\emph{Interpretation of the optimization-gap condition.}
Proposition~\ref{prop: P_c_delta} is conditional on the empirical optimality gap bounds in~\eqref{eq: approximate ERM condition} for the score functions $\widehat{\f}_k$ returned by local training. It does not provide an algorithm-specific guarantee that a prescribed $\varepsilon_{k,\opt}$ is attained. In particular, when $\F_k$ is parameterized by a neural network, convexity of the scalar surrogate loss $\Phi$ does not imply convexity of the empirical training objective with respect to the network parameters. Any optimization guarantee for a particular training algorithm is therefore separate from the statistical statement of Proposition~\ref{prop: P_c_delta}. Once such tolerances are available, their effect enters the bound through the quantity $\varepsilon_{\opt}$.

\subsection{Homogeneous Centralized-Training Comparison}
\label{appendix: centralized training comparison}

While the bound in Proposition~\ref{prop: P_c_delta} is derived for independent local training, where each agent trains its own classifier without exchanging model updates, it is instructive to compare this setup against a centralized benchmark. For clarity, we focus on the homogeneous case where all agents share the same feature space, model class, and local population distribution. In the heterogeneous case, different agents may observe different feature distributions or spaces and use distinct architectures, so a single centralized model is not directly comparable without introducing additional modeling choices.

For this comparison, suppose a total of $N_{\sf tot}$ i.i.d.\ training samples from the common local population is available. A centralized learner uses all $N_{\sf tot}$ samples to train one classifier, whereas independent training partitions the same sample budget across $K$ agents. For a symmetric comparison, we consider an even split,
\begin{equation}
    N_k=\frac{N_{\sf tot}}{K}, \quad k=1,\ldots,K.
\end{equation}
We next examine the quantities that enter the bound in~\eqref{eq: P_c_delta}.

(i) \emph{Network complexity $\rho$}: For many standard base models, such as feedforward neural networks with norm-constrained weights and kernel methods over bounded balls in reproducing kernel Hilbert spaces, the Rademacher complexity admits the standard $O(1/\sqrt{m})$ dependence on the training set size $m$~\cite{bordignon2023learning,bartlett2002rademacher,neyshabur2015norm}. For the homogeneous setting with an even split, all local complexities are identical and therefore
\begin{equation}
    \rho^{\sf ind} = \mathfrak R_{N_{\sf tot}/K}(\F), \quad \rho^{\sf cen} = \mathfrak R_{N_{\sf tot}}(\F)
\end{equation}
where $\F$ denotes the common function class. Hence, for model classes whose Rademacher complexity decreases with sample size as above,
\begin{equation}
    \rho^{\sf cen}\leq\rho^{\sf ind}.
\end{equation}

(ii) \emph{Function $\mathdutchcal E_\Phi(\mathsf R^o+\varepsilon_{\opt},\delta)$}: Because the local population distribution and model class are common, the individual target risks are identical, and hence the target risk $\mathsf R^o$ is the same for centralized and independent training. The loss function $\Phi$ is also unchanged. Let $\varepsilon_{\opt}^{\sf cen}$ denote the empirical optimization tolerance of the centralized learner and let
\begin{equation}
    \varepsilon_{\opt}^{\sf ind} \triangleq \sum_{k=1}^K\pi_k\varepsilon_{k,\opt}
\end{equation}
denote the corresponding $\pi$-weighted tolerance for independent training. If the two procedures are compared at matched optimization accuracy,
\begin{equation}
    \varepsilon_{\opt}^{\sf cen} = \varepsilon_{\opt}^{\sf ind} \triangleq\varepsilon,
\end{equation}
then $\mathdutchcal E_\Phi(\mathsf R^o+\varepsilon,\delta)$ is identical in the two bounds. If their optimization tolerances differ, the corresponding effective risk levels must instead be retained separately. The theory does not presume that either training procedure solves its empirical objective more accurately.

(iii) \emph{Admissible margin range}: At matched optimization accuracy, the effective risk $\mathsf R^o+\varepsilon$ and the loss $\Phi$ are the same under the two setups. Consequently, $\delta_{\max}(\mathsf R^o+\varepsilon)$, and hence the admissible range of $\delta$, is also identical.

The common model class also yields the same score bound $\beta$. Suppose now that $\delta$ lies in the common admissible range and that
\begin{equation}
    \rho^{\sf ind} < \mathdutchcal E_\Phi(\mathsf R^o+\varepsilon,\delta).
\end{equation}
Under the dependence of training set size described above, $\rho^{\sf cen}\leq\rho^{\sf ind}$, so Proposition~\ref{prop: P_c_delta} applies to both training setups. For centralized training, taking $K=1$, $N_{\max}=N_{\sf tot}$, and $\alpha=1$ gives
\begin{equation}
    P_{c,\delta}^{\sf cen}\geq 1-2\exp\Biggl\{-\frac{8N_{\sf tot}}{\beta^2}\Bigl(\mathdutchcal E_\Phi(\mathsf R^o+\varepsilon,\delta)-\rho^{\sf cen}\Bigr)^2\Biggr\}.
\end{equation}
For independent training with the even split, $N_{\max}=N_{\sf tot}/K$ and $\alpha=1$, which gives
\begin{equation}
    P_{c,\delta}^{\sf ind}\geq 1-2\exp\Biggl\{-\frac{8N_{\sf tot}}{K\beta^2}\Bigl(\mathdutchcal E_\Phi(\mathsf R^o+\varepsilon,\delta)-\rho^{\sf ind}\Bigr)^2\Biggr\}.
\end{equation}
Therefore, under this homogeneous comparison with matched optimization accuracy, the centralized benchmark benefits from both a larger effective size of training set in the exponent and, for the model classes considered above, a Rademacher complexity that is no larger. Consequently, it yields a tighter lower bound on $P_{c,\delta}$ than that obtained by evenly splitting the same sample budget among independently trained models. This comparison makes explicit a statistical cost associated with independent training.

\section{Sufficient Communication}
\label{appendix: sufficient communication}

This appendix provides the prediction argument underlying the sufficient-communication guarantee in Section~\ref{sec: convergence analysis} of the main paper. We first record the bounded-difference concentration inequality obtained by combining the class-conditional approximate tensorization of entropy (ATE) assumption with the entropy method. We then apply this inequality under each class-conditional testing law $P_\gamma$ and combine the resulting prediction bound with the training guarantee from Proposition~\ref{prop: P_c_delta}. The same ATE concentration tool will also be used in the subsequent prediction analyses of other appendices.

We recall the $\sigma$-field generated by the training data and local training randomness in~\eqref{eq: training sigma field}:
\begin{equation}
\label{eq: training sigma field supplement}
    \Ttrain \triangleq \sigma(\D_1,\ldots,\D_K,\mathcal{S}_1,\ldots,\mathcal{S}_K).
\end{equation}

\subsection{Bounded-Difference Inequality under ATE}
\label{appendix: ATE bounded differences}

We first record the bounded-difference concentration consequence of ATE used throughout the prediction analysis. It follows by applying the standard entropy-method bounded-difference argument of~\cite{boucheron2003entropy} with the approximate tensorization inequality of~\cite{caputo2015tensorization} in place of exact tensorization. We use this established consequence without reproving it.

\begin{inlemma}[\textbf{Bounded differences under ATE}]
\label{lemma: ATE bounded differences}
Let $X=(X_1,\ldots,X_K)$ have law $P$ satisfying ATE with coefficient $\tau$, in the sense of~\eqref{eq: class conditional ATE}. Suppose a bounded measurable function $F$ satisfies
\begin{equation}
    |F(x)-F(x')|\leq d_k
\end{equation}
whenever $x$ and $x'$ differ only in coordinate $k$. Then, for every $\theta\in\mathbb R$,
\begin{equation}
\label{eq: ATE MGF bound}
    \log\E\exp\{\theta(F(X)-\E F(X))\} \leq \frac{\tau\theta^2}{8}\sum_{k=1}^K d_k^2.
\end{equation}
Consequently, if $\sum_{k=1}^K d_k^2>0$, then, for every $u\geq0$,
\begin{equation}
\label{eq: ATE bounded differences tail}
    \P\bigl(F(X)-\E F(X)\leq-u\bigr) \leq \exp\Biggl\{-\frac{2u^2}{\tau\sum_{k=1}^K d_k^2} \Biggr\}.
\end{equation}
The same bound holds for the upper tail. If $\sum_{k=1}^K d_k^2=0$, then $F$ is constant and the concentration statement is trivial.
\end{inlemma}

\subsection{Proof of Theorem~\ref{theorem: classification error}}
\label{appendix: classification error}

Fix a realization of the training phase that belongs to $\C_\delta$. Conditional on $\Ttrain$, the learned score functions and empirical training means are fixed, whereas the fresh testing pair remains independent of the training phase. Conditional on $\bgamma^*=\gamma$, the joint testing observation $\h^*$ therefore has law $P_\gamma$. We bound the prediction error separately under the two class-conditional laws and then average over the uniform class prior. To lighten notation, throughout this proof we write $\widehat f_k$, $c_k$, and $c_\ave$ for the corresponding realized training-dependent quantities; the same convention will be used in the later prediction proofs.
\begin{proof}
Suppose first that $\bgamma^*=+1$. From the definition of the aggregate classifier $c_\ave$ in~\eqref{eq: single-sample ensemble classifier},
\begin{align}
    \E_{\h^*}^{(+1)}c_\ave(\h^*) &=\sum_{k=1}^K\pi_k\E_{\h_k^*}^{(+1)}c_k(\h_k^*)
    \nonumber\\
    &=\mu^{+}(\widehat f) > \delta
    \label{eq: expectation-single-sample case}
\end{align}
where the last inequality follows because the fixed training realization belongs to $\C_\delta$.

Now consider another observation collection $\widehat{\h}$ that differs from $\h^*$ only in coordinate $k$. From \eqref{eq: trained classifier} and Assumption~\ref{assump: bound},
\begin{align}
    \bigl|c_\ave(\h^*)-c_\ave(\widehat{\h})\bigr| &= \pi_k \bigl|c_k(\h_k^*)-c_k(\widehat{\h}_k)\bigr|
    \nonumber\\
    &=\pi_k\bigl|\widehat f_k(\h_k^*) - \widehat f_k(\widehat{\h}_k) \bigr|
    \nonumber\\
    &\leq 2\beta\pi_k.
    \label{eq: bound of variation for each agent-single-sample case}
\end{align}
This reveals that the coordinate oscillations of $c_\ave$ are $d_k=2\beta\pi_k$. Conditional on $\bgamma^*=+1$ and $\Ttrain$, the testing observation has law $P_{+1}$, which satisfies ATE with coefficient $\tau_{+1}$. Hence, following~\eqref{eq: expectation-single-sample case}, Lemma~\ref{lemma: ATE bounded differences} gives
\begin{align}
    &\;\P\Bigl( c_\ave(\h^*)\leq 0 \givenbig \bgamma^*=+1,\Ttrain \Bigr)
    \nonumber\\
    &= \P\Bigl(c_\ave(\h^*) - \E_{\h^*}^{(+1)}c_\ave(\h^*)\leq -\mu^+(\widehat f) \givenbig \bgamma^*=+1,\Ttrain \Bigr)
    \nonumber\\
    &\leq\exp\Biggl\{ -\frac{\delta^2}{2\tau_{+1}\beta^2\sum_{k=1}^K\pi_k^2} \Biggr\}.
    \label{eq: condition probability for single-sample case +1}
\end{align}
The final inequality follows from $\mu^+(\widehat f)>\delta$ on $\C_\delta$. For $\bgamma^*=-1$, we similarly have
\begin{equation}
    \E_{\h^*}^{(-1)}c_\ave(\h^*) = \mu^-(\widehat f) < -\delta
\end{equation}
on $\C_\delta$, and the same coordinate oscillations $d_k=2\beta\pi_k$ apply. Accordingly, we have
\begin{equation}\label{eq: condition probability for single-sample case -1}
    \P\Bigl(c_\ave(\h^*)\geq0\givenbig\bgamma^*=-1,\Ttrain\Bigr) \leq \exp\Biggl\{ -\frac{\delta^2}{2\tau_{-1}\beta^2\sum_{k=1}^K\pi_k^2}\Biggr\}.
\end{equation}
Since $\bgamma^*$ is independent of $\Ttrain$ and has the uniform prior, the following holds
\begin{align}
    \P(\M\givensmall\Ttrain) &= \frac12\P\Bigl(c_\ave(\h^*)\leq0 \givenbig \bgamma^*=+1,\Ttrain\Bigr)
    \nonumber\\
    &\quad+\frac12\P\Bigl(c_\ave(\h^*)\geq0 \givenbig \bgamma^*=-1,\Ttrain \Bigr)
    \nonumber\\
    &\leq \frac12 \sum_{\gamma\in\Gamma}\exp\Biggl\{ -\frac{\delta^2}{2\tau_\gamma\beta^2\sum_{k=1}^K\pi_k^2}\Biggr\}
    \nonumber\\
    &\leq\exp\Biggl\{-\frac{\delta^2}{2\tau_{\max}\beta^2\sum_{k=1}^K\pi_k^2}\Biggr\}
\label{eq: sufficient communication conditional appendix}
\end{align}
almost surely on $\C_\delta$. This proves the conditional bound \eqref{eq: P_e for single-sample case under delta-margin consistent training} in Theorem~\ref{theorem: classification error}. Since $\C_\delta$ is $\Ttrain$-measurable and the preceding conditional bound holds almost surely on $\C_\delta$, we have
\begin{align}
    \E\!\left[\mathsf{1}[\C_\delta]\P(\M\givensmall\Ttrain)\right]
    &\leq \exp\Biggl\{-\frac{\delta^2}{2\tau_{\max}\beta^2\sum_{k=1}^K\pi_k^2}\Biggr\}\P(\C_\delta)
    \nonumber\\
    &\leq\exp\Biggl\{-\frac{\delta^2}{2\tau_{\max}\beta^2\sum_{k=1}^K\pi_k^2}\Biggr\}.
\label{eq: conditional prediction contribution appendix}
\end{align}
Combining~\eqref{eq: conditional prediction contribution appendix} with the decomposition in~\eqref{eq: P_e decomposition} gives
\begin{equation}
\label{eq: sufficient communication unconditional appendix}
    P_e\leq\exp\Biggl\{-\frac{\delta^2}{2\tau_{\max}\beta^2\sum_{k=1}^K\pi_k^2}\Biggr\}+\P(\overline{\C_\delta}).
\end{equation}
If the assumptions of Proposition~\ref{prop: P_c_delta} also hold, substituting the bound on $\P(\overline{\C_\delta})$ from~\eqref{eq: P_c_delta} into \eqref{eq: sufficient communication unconditional appendix} gives
\begin{align}
    P_e &\leq\exp\Biggl\{ -\frac{\delta^2}{2\tau_{\max}\beta^2\sum_{k=1}^K\pi_k^2}\Biggr\}
    \nonumber\\
    &\quad+2\exp\Biggl\{-\frac{8N_{\max}}{\alpha^2\beta^2}\Bigl(\mathdutchcal E_\Phi(\mathsf R^o+\varepsilon_{\opt},\delta)-\rho \Bigr)^2 \Biggr\}.
\label{eq: sufficient communication explicit Pe appendix}
\end{align}
This is the explicit classification error bound under sufficient communication.
\end{proof}

\section{Finite-Round Communication}
\label{appendix: finite round communication}

Appendix~\ref{appendix: sufficient communication} establishes the classification guarantee for the limiting case of sufficient communication. At a finite communication round, a similar concentration argument applies. The proof below reuses the conditioning convention and ATE bounded-difference inequality from Appendix~\ref{appendix: sufficient communication}

\begin{proof}
Fix a realization of the training phase that belongs to $\C_\delta$. As in Appendix~\ref{appendix: classification error}, the learned functions, their empirical training means, and therefore the centered score functions are fixed conditional on $\Ttrain$. We proceed first with the case $\bgamma^*=+1$. From~\eqref{eq: finite round nonpositive margin event}, we have
\begin{equation}
    \P\Bigl(\M_{k,t}\givenbig\bgamma^*=+1,\Ttrain\Bigr)= \P\Bigl(\blambda_{k,t}(\h^*)\leq 0 \givenbig \bgamma^*=+1,\Ttrain\Bigr)
    \label{eq: P_e equation at t-th communication round for case +1}
\end{equation}
where the argument of $\h^*$ in the notation $\blambda_{k,t}(\h^*)$ emphasizes the dependence on the testing observation collection. Consider another observation collection $\widehat{\h}$ that differs from $\h^*$ only in coordinate $\ell$, \eqref{eq: decision statistic at time t} gives
\begin{align}
    \bigl|\blambda_{k,t}(\h^*)-\blambda_{k,t}(\widehat{\h})\bigr|
    &=\Bigl|[A^t]_{\ell k}\bigl(c_\ell(\h_\ell^*)-c_\ell(\widehat{\h}_\ell)\bigr)\Bigr|
    \nonumber\\
    &\leq 2\beta[A^t]_{\ell k}.
    \label{eq: bound of variation for t-th round of single-sample case}
\end{align}
Thus, the coordinate oscillation associated with the $\ell$th testing view is $d_\ell=2\beta[A^t]_{\ell k}$.

We next control the deviation in the class-conditional mean caused by incomplete mixing. Since $|\widehat f_\ell|\leq\beta$ and its empirical training mean lies in $[-\beta,\beta]$, the realized centered class mean satisfies $|\mu_\ell^+(\widehat f_\ell)|\leq2\beta$. Hence, using \eqref{eq: convergence constant main},
\begin{align}
    &\Bigl|\E_{\h^*}^{(+1)}\blambda_{k,t}(\h^*) - \mu^+(\widehat f)\Bigr|
    \nonumber\\
    &=\biggl|\sum_{\ell=1}^K \bigl([A^t]_{\ell k}-\pi_\ell\bigr)\mu_\ell^+(\widehat f_\ell)\biggr|
    \nonumber\\
    &\leq 2\beta\sum_{\ell=1}^K\Bigl|[A^t]_{\ell k}-\pi_\ell\Bigr|
    \nonumber\\
    &\leq2\beta K C(A,\sigma)\sigma^t.
    \label{eq: C3}
\end{align}
Therefore, we have $\mu^+(\widehat f)>\delta$ when $\C_\delta$ is true. Consequently, for every $t$ satisfying $r_t(\delta)>0$, the following holds:
\begin{align}
    \E_{\h^*}^{(+1)}\blambda_{k,t}(\h^*) &>\delta-2\beta K C(A,\sigma)\sigma^t
    \nonumber\\
    &= r_t(\delta)>0.
    \label{eq: C4}
\end{align}
Conditional on $\bgamma^*=+1$ and $\Ttrain$, the testing observation has law $P_{+1}$. Therefore, using the similar arguments used in \eqref{eq: condition probability for single-sample case +1} yields
\begin{equation}
    \P\Bigl(\M_{k,t}\givenbig\bgamma^*=+1,\Ttrain\Bigr) \leq \exp\Biggl\{-\frac{r_t^2(\delta)}{2\tau_{+1}\beta^2
        \sum_{\ell=1}^K\bigl([A^t]_{\ell k}\bigr)^2}\Biggr\}.
    \label{eq: condition probability for t-th round single-sample case +1}
\end{equation}
Repeating the steps in~\eqref{eq: C3}--\eqref{eq: condition probability for t-th round single-sample case +1} for the case $\bgamma^*=-1$, we can show that almost surely on $\C_\delta$,
\begin{equation}
    \P(\M_{k,t}\givensmall\Ttrain) \leq \exp\Biggl\{ -\frac{r_t^2(\delta)}{2\tau_{\max}\beta^2\sum_{\ell=1}^K\bigl([A^t]_{\ell k}\bigr)^2}\Biggr\},
\end{equation}
which proves the bound in~\eqref{eq: P_e for t-th round single-sample case under delta-margin consistent training}. If the conditions of Proposition~\ref{prop: P_c_delta} also hold, substituting the bound on $\P(\overline{\C_\delta})$ from \eqref{eq: P_c_delta} gives the following explicit upper bound for classification after $t$ communication rounds:
\begin{align}
    P_{k,t} &\leq\exp\Biggl\{-\frac{r_t^2(\delta)}{2\tau_{\max}\beta^2\sum_{\ell=1}^K\bigl([A^t]_{\ell k}\bigr)^2}
    \Biggr\}
    \nonumber\\
    &\quad+2\exp\Biggl\{-\frac{8N_{\max}}{\alpha^2\beta^2}\Bigl(\mathdutchcal E_\Phi(\mathsf R^o+\varepsilon_{\opt},\delta) - \rho\Bigr)^2\Biggr\}.
\end{align}
This completes the proof of Theorem~\ref{theorem: finite-time classification error}.
\end{proof}

\section{Finite-Precision Communication}
\label{appendix: finite bit communication}

Appendix~\ref{appendix: finite round communication} establishes the finite-round classification guarantee under exact real-valued message exchange. Theorem~\ref{theorem: finite bit communication} additionally accounts for stochastic finite-precision communication. Relative to the finite-round analysis, two new issues must be controlled: i) quantization errors accumulate as they propagate through the network, and ii) their distributions are input-dependent, so that the resulting errors need not be independent across rounds or from the fresh testing views. We first derive a conditional sub-Gaussian bound for the accumulated quantization perturbation by unrolling the recursion and conditioning on the quantization history. We then combine this bound with the finite-round concentration bound established in Appendix~\ref{appendix: finite round communication}.

\begin{proof}
For $t=0$, no quantization has yet occurred: $x_{k,0}^{(b)}=c_k(\h_k^*)=\blambda_{k,0}$ and $V_{k,0}^{\rm q}=0$. The claimed conditional bound therefore reduces to the corresponding real-valued finite-round bound. We henceforth consider $t\geq1$.

To isolate the effect of stochastic quantization, we condition on the training realization and the fresh testing example, so that only the rounding randomness remains. Let
\begin{equation}
    \x_t^{(b)} \triangleq(x_{1,t}^{(b)},\ldots,x_{K,t}^{(b)})^\top, \quad
    \x_t^{\rm fp} \triangleq (x_{1,t}^{\rm fp},\ldots,x_{K,t}^{\rm fp})^\top
\end{equation}
where $\x_t^{\rm fp}$ corresponds to the matched full-precision recursion obtained from the same initialization by replacing $\mathsf Q_b$ with the identity map. In particular,
\begin{equation}
    x_{k,t}^{\rm fp} = \blambda_{k,t}
\end{equation}
where $\blambda_{k,t}$ is defined in \eqref{eq: decision statistic at time t}. Before analyzing the rounding errors, note that the prescribed quantization range is preserved by the recursion. From~\eqref{eq: centered score quantization range}, $x_{k,0}^{(b)}=c_k(\h_k^*)\in[-2\beta,2\beta]$. The outputs of the quantizer $\mathsf Q_b$ remain in the same interval, and \eqref{eq: finite bit recursion theory} forms each subsequent decision statistic as a convex combination of the transmitted values. Hence, by induction,
\begin{equation}
    x_{k,s}^{(b)}\in[-2\beta,2\beta]  \quad  \text{for every agent $k$ and round $s$}.
\end{equation}
This indicates that the input to the stochastic quantizer remains within its prescribed range at every round.

Define the sender-level rounding errors
\begin{equation}
    \varepsilon_{\ell,s}^{\rm q} \triangleq q_{\ell,s}^{(b)}-x_{\ell,s}^{(b)}, \quad
    \bm{\varepsilon}_s^{\rm q}\triangleq(\varepsilon_{1,s}^{\rm q},\ldots,\varepsilon_{K,s}^{\rm q})^\top.
\end{equation}
Let $W=A^\top$. The finite-precision and matched full-precision recursions satisfy
\begin{equation}
    \x_{s+1}^{(b)} = W\x_s^{(b)} + W\bm{\varepsilon}_s^{\rm q}, \quad
    \x_{s+1}^{\rm fp} = W\x_s^{\rm fp},
\end{equation}
with the same initialization. Expanding the recursion therefore gives
\begin{equation}
\label{eq: finite bit unrolling theory appendix}
    \x_t^{(b)}-\x_t^{\rm fp} = \sum_{s=0}^{t-1} W^{t-s}\bm{\varepsilon}_s^{\rm q}.
\end{equation}
Therefore, a rounding error generated at round $s$ is propagated through $W^{t-s}$ before contributing to the decision statistic at round $t$.

Let $\mathcal Q_s$ denote the $\sigma$-field generated by $\Ttrain$, the fresh testing pair $(\h^*,\bgamma^*)$, and all quantizer randomness through round $s-1$. Conditional on $\mathcal Q_s$, the current decision statistics of each agent and therefore the active quantization bins are fixed. Following the stochastic rounding protocol~\eqref{eq: stochastic rounding protocol}, the sender-level rounding errors at round $s$ are conditionally independent across senders, satisfy
\begin{equation}
    \E\!\left[\varepsilon_{\ell,s}^{\rm q}\givenbig\mathcal Q_s\right] = 0,
\end{equation}
and each takes values in an interval of length at most $\Delta_b$. Hence, for any $\mathcal Q_s$-measurable weights $w_1,\ldots,w_K$, Hoeffding's lemma~\cite{mohri2018foundations} gives, for every $\theta\in\mathbb R$,
\begin{equation}
    \E\!\left[ \exp\!\left( \theta\sum_{\ell=1}^K w_\ell\varepsilon_{\ell,s}^{\rm q} \right) \middle|\mathcal Q_s
    \right]\leq \exp\!\left\{\frac{\theta^2\Delta_b^2}{8}\sum_{\ell=1}^K w_\ell^2 \right\}.
\label{eq: finite bit one round conditional mgf theory appendix}
\end{equation}
Given a fixed terminal pair $(k,t)$, the $k$th coordinate of \eqref{eq: finite bit unrolling theory appendix} can be written as
\begin{equation}
    x_{k,t}^{(b)}-x_{k,t}^{\rm fp} = \sum_{s=0}^{t-1}U_s,  \quad
    U_s \triangleq \sum_{j=1}^K[A^{t-s}]_{jk}\varepsilon_{j,s}^{\rm q}.
\end{equation}
Applying \eqref{eq: finite bit one round conditional mgf theory appendix} with $w_j=[A^{t-s}]_{jk}$ gives
\begin{equation}
    \E\!\left[ e^{\theta U_s}\givenbig\mathcal Q_s\right] \leq\exp\!\Biggl\{\frac{\theta^2\Delta_b^2}{8}
        \sum_{j=1}^K([A^{t-s}]_{jk})^2 \Biggr\}.
\label{eq: finite bit increment conditional mgf appendix}
\end{equation}
It is worth noting that the errors at different rounds need not be independent because the quantization bin at a later round depends on the preceding randomized updates. For this reason, in the following, we accumulate the bounds through repeated conditioning rather than by multiplying unconditional moment generating functions (MGFs). According to the law of total expectation, conditioning first on $\mathcal Q_{t-1}$ gives
\begin{align}
    &\;\E\!\left[\exp\!\left(\theta\sum_{s=0}^{t-1}U_s\right)\middle|\Ttrain,\h^*,\bgamma^*\right]
    \nonumber\\
    &=\E\!\left[\exp\!\left(\theta\sum_{s=0}^{t-2}U_s\right)\E\!\left[e^{\theta U_{t-1}}\givenbig\mathcal Q_{t-1}\right] \middle| \Ttrain,\h^*,\bgamma^* \right]
    \nonumber\\
    &\leq \exp\!\left\{\frac{\theta^2\Delta_b^2}{8} \sum_{j=1}^K([A]_{jk})^2 \right\}\E\!\left[\exp\!\left(\theta\sum_{s=0}^{t-2}U_s \right) \middle| \Ttrain,\h^*,\bgamma^* \right].
\end{align}
Repeating the same conditioning step for $U_{t-2},\ldots,U_0$ yields
\begin{align}
    &\;\E\!\left[\exp\!\left(\theta[x_{k,t}^{(b)}-x_{k,t}^{\rm fp}]\right) \middle| \Ttrain,\h^*,\bgamma^*\right]
    \nonumber\\
    &\leq\exp\!\left\{\frac{\theta^2\Delta_b^2}{8}\sum_{r=1}^{t}\sum_{j=1}^K([A^r]_{jk})^2\right\}
    \nonumber\\
    &=\exp\!\left\{ \frac{\theta^2V_{k,t}^{\rm q}}{2}\right\}
\label{eq: finite bit perturbation mgf theory appendix}
\end{align}
where $V_{k,t}^{\rm q}$ is defined in~\eqref{eq: quantization proxy theory}. This derivation requires conditional independence only among the sender-level rounding draws within each round; temporal independence of the rounding errors is not needed.

We now combine the accumulated quantization perturbation with the fluctuation induced by the fresh testing views. Fix $\gamma\in\Gamma$ and a training realization in $\C_\delta$, and define
\begin{equation}
    \mu_{k,t}^{(\gamma)}\triangleq\E\!\left[\gamma x_{k,t}^{\rm fp} \givenbig \Ttrain,\bgamma^*=\gamma \right],
    \quad Y_{k,t}^{(\gamma)} \triangleq\gamma x_{k,t}^{\rm fp} - \mu_{k,t}^{(\gamma)}.
\end{equation}
Since $x_{k,t}^{\rm fp}=\blambda_{k,t}$, the finite-round mean argument leading to~\eqref{eq: C4} in Appendix~\ref{appendix: finite round communication}, applied to either class, gives
\begin{equation}
\label{eq: finite bit ideal mean lower bound appendix}
    \mu_{k,t}^{(\gamma)} \geq r_t(\delta).
\end{equation}
For fixed $\Ttrain$ and $\bgamma^*=\gamma$, changing only the $\ell$th testing view changes $\gamma x_{k,t}^{\rm fp}$ by at most $2\beta[A^t]_{\ell k}$. Applying the MGF form of Lemma~\ref{lemma: ATE bounded differences}, namely~\eqref{eq: ATE MGF bound}, under $P_\gamma$ therefore gives
\begin{equation}
    \E\!\left[ e^{\theta Y_{k,t}^{(\gamma)}} \givenbig \Ttrain,\bgamma^*=\gamma \right]\leq\exp\!\left\{\frac{\theta^2}{2} \tau_\gamma\beta^2 \sum_{\ell=1}^K([A^t]_{\ell k})^2\right\}
\label{eq: finite bit ATE mgf appendix}
\end{equation}
for any $\theta\in\mathbb R$. Define the signed quantization perturbation:
\begin{equation}
    Z_{k,t}^{\rm q,\gamma}\triangleq \gamma \bigl(x_{k,t}^{(b)}-x_{k,t}^{\rm fp}\bigr).
\end{equation}
Although the quantizer uses fresh auxiliary randomness, the distribution of $Z_{k,t}^{\rm q,\gamma}$ depends on the realized testing example through the input-dependent quantization bins. Thus, $Z_{k,t}^{\rm q,\gamma}$ and $Y_{k,t}^{(\gamma)}$ are not assumed to be independent. Instead, we condition first on the complete fresh testing example. Since $Y_{k,t}^{(\gamma)}$ is then fixed, the tower property and \eqref{eq: finite bit perturbation mgf theory appendix} give
\begin{align}
    &\;\E\!\left[ e^{\theta( Y_{k,t}^{(\gamma)} + Z_{k,t}^{\rm q,\gamma})}\givenbig\Ttrain,\bgamma^*=\gamma\right]
    \nonumber\\
    &=\E\!\left[ e^{\theta Y_{k,t}^{(\gamma)}}\E\!\left[e^{\theta Z_{k,t}^{\rm q,\gamma}} \givenbig\Ttrain,\h^*,\bgamma^*=\gamma\right] \middle| \Ttrain,\bgamma^*=\gamma \right]
    \nonumber\\
    &\leq \exp\!\left\{ \frac{\theta^2V_{k,t}^{\rm q}}{2}\right\}\E\!\left[e^{\theta Y_{k,t}^{(\gamma)}} \givenbig\Ttrain,\bgamma^*=\gamma \right]
    \nonumber\\
    &\leq\exp\!\left\{ \frac{\theta^2}{2}\left[\tau_\gamma\beta^2\sum_{\ell=1}^K([A^t]_{\ell k})^2 + V_{k,t}^{\rm q}\right]\right\}.
\label{eq: finite bit combined mgf appendix}
\end{align}
Therefore, the sub-Gaussian variance proxies associated with the testing views and quantization errors add in the conditional MGF bound. In particular, no unconditional independence between the two fluctuations is required. Suppose $r_t(\delta)>0$. Since
\begin{equation}
    \gamma x_{k,t}^{(b)} = \mu_{k,t}^{(\gamma)} + Y_{k,t}^{(\gamma)} + Z_{k,t}^{\rm q,\gamma},
\end{equation}
the event $\M_{k,t}^{(b)}=\{\bgamma^*x_{k,t}^{(b)}\leq0\}$ defined in~\eqref{eq: M_k_t for finite precision}, conditional on $\bgamma^*=\gamma$, together with \eqref{eq: finite bit ideal mean lower bound appendix}, implies
\begin{equation}
    Y_{k,t}^{(\gamma)} + Z_{k,t}^{\rm q,\gamma} \leq -r_t(\delta).
\end{equation}
Let
\begin{equation}
    V_{k,t}^{(\gamma)} \triangleq \tau_\gamma\beta^2 \sum_{\ell=1}^K([A^t]_{\ell k})^2 + V_{k,t}^{\rm q}.
\end{equation}
For any $\lambda>0$, Chernoff's method and \eqref{eq: finite bit combined mgf appendix} give
\begin{align}
    &\;\P\Bigl(Y_{k,t}^{(\gamma)} + Z_{k,t}^{\rm q,\gamma} \leq -r_t(\delta) \givenbig \Ttrain,\bgamma^*=\gamma\Bigr)
    \nonumber\\
    &\leq\exp\!\left\{ -\lambda r_t(\delta) +\frac{\lambda^2}{2}V_{k,t}^{(\gamma)}\right\}.
\end{align}
Optimizing at $\lambda=r_t(\delta)/V_{k,t}^{(\gamma)}$ yields
\begin{align}
    &\;\P(\M_{k,t}^{(b)}
        \givenbig
        \Ttrain,\bgamma^*=\gamma)
    \nonumber\\
    &\leq
    \exp\!\left\{
        -\frac{r_t^2(\delta)}
        {2\left[
            \tau_\gamma\beta^2
            \sum_{\ell=1}^K([A^t]_{\ell k})^2
            +
            V_{k,t}^{\rm q}
        \right]}
    \right\}.
\label{eq: finite bit class conditional appendix}
\end{align}
Since the fresh label is independent of $\Ttrain$ and has the uniform prior, averaging the two class-conditional bounds gives, almost surely on $\C_\delta$,
\begin{align}
    \P(\M_{k,t}^{(b)}\givenbig\Ttrain)
    &\leq
    \frac12
    \sum_{\gamma\in\Gamma}
    \exp\!\left\{
        -\frac{r_t^2(\delta)}
        {2\left[
            \tau_\gamma\beta^2
            \sum_{\ell=1}^K([A^t]_{\ell k})^2
            +
            V_{k,t}^{\rm q}
        \right]}
    \right\}
    \nonumber\\
    &\leq
    \exp\!\left\{
        -\frac{r_t^2(\delta)}
        {2\left[
            \tau_{\max}\beta^2
            \sum_{\ell=1}^K([A^t]_{\ell k})^2
            +
            V_{k,t}^{\rm q}
        \right]}
    \right\},
\end{align}
which proves~\eqref{eq: finite bit classification theory}. Using the decomposition in \eqref{eq: P_e decomposition} with $\M$ replaced by $\M_{k,t}^{(b)}$, if the conditions of Proposition~\ref{prop: P_c_delta} also hold, substituting its bound on $\P(\overline{\C_\delta})$ gives
\begin{align}
    P_{k,t}^{(b)}
    &\leq
    \exp\!\left\{
        -\frac{r_t^2(\delta)}
        {2\left[
            \tau_{\max}\beta^2
            \sum_{\ell=1}^K([A^t]_{\ell k})^2
            +
            V_{k,t}^{\rm q}
        \right]}
    \right\}
    \nonumber\\
    &\quad+
    2\exp\Biggl\{
        -\frac{8N_{\max}}
        {\alpha^2\beta^2}
        \Bigl(
            \mathdutchcal E_\Phi(
                \mathsf R^o+\varepsilon_{\opt},\delta
            )
            -
            \rho
        \Bigr)^2
    \Biggr\}.
\end{align}
This is the explicit unconditional finite-precision classification bound referred to in Theorem~\ref{theorem: finite bit communication} and completes the proof.
\end{proof}

\section{PAC-Style Generalization Guarantee}
\label{appendix: PAC generalization bound}

This appendix proves Theorem~\ref{theorem: generalization bound for c_k and c_ave} of the main paper. Unlike the prediction guarantees in Appendices~\ref{appendix: sufficient communication}--\ref{appendix: finite bit communication}, the PAC-style result does not rely on the expected margin event $\C_\delta$ or the class-conditional ATE assumption. Its sampling requirement is instead that the network examples $(\h_n,\bgamma_n)$ are i.i.d.\ across $n$, while the $K$ local views within one example may remain statistically dependent.

The main technical issue is the empirical centering operation in~\eqref{eq: trained classifier}, which leads the resulting classifier class to be sample-dependent. The proof handles this issue in three steps. We first bound the sensitivity of the uniform population--empirical risk gap to replacing one network example. We then embed every realized centered aggregate in a fixed offset-augmented function class and control its expected uniform deviation by Rademacher complexity. Finally, McDiarmid's inequality yields a high-probability bound, which is specialized to the $\eta$-margin loss $\Phi_\eta$ in Section~\ref{sec: PAC-style generalization bounds}.

\begin{proof}
We prove the collaborative bound in~\eqref{eq: margin-based generalization bound for c_ave}, with the local bound following from the same argument by taking a single agent. Recall that
\begin{equation}
     \D=\{(\h_n,\bgamma_n)\}_{n=1}^{N}
\end{equation}
contains $N$ i.i.d.\ complete network examples drawn from $Q$, and that $\D_k=\{(\h_{k,n},\bgamma_n)\}_{n=1}^{N}$ is its projection at agent $k$. The proof treats each complete network example as one independent sample coordinate and therefore never requires independence among the local views within an example.

For $f_k\in\F_k$, we define its empirical training mean on $\D_k$ and centered score by
\begin{align}
    \bar f_k(\D_k)
    &\triangleq\frac{1}{N}\sum_{n=1}^{N}f_k(\h_{k,n}),
    \label{eq: empirical mean for agent k}\\
    \c_{f_k,\D_k}(h_k)
    &\triangleq f_k(h_k)-\bar f_k(\D_k).
    \label{eq: centered classifier for agent k}
\end{align}
For $f=(f_1,\ldots,f_K)\in\F$, let
\begin{equation}
    \c_{\ave,f}^{\D}(h)
    \triangleq
    \sum_{k=1}^{K}\pi_k\c_{f_k,\D_k}(h_k),
    \label{eq: definition of c_ave_f}
\end{equation}
and, for a generic $L_\Phi$-Lipschitz loss $\Phi$, define
\begin{align}
    \R_{\emp}^{\D}(f)
    &\triangleq\frac{1}{N}\sum_{n=1}^{N}
    \Phi\!\left(\bgamma_n\c_{\ave,f}^{\D}(\h_n)\right),
    \label{eq: empirical risk of c_ave}\\
    \R^{\D}(f)
    &\triangleq\E_{(\h,\bgamma)\sim Q}
    \Phi\!\left(\bgamma\c_{\ave,f}^{\D}(\h)\right).
    \label{eq: expected risk of c_ave}
\end{align}
The superscript $\D$ emphasizes that, even for fixed $f$, the deployed centered aggregate score associated with $\c_\ave$ depends on the training sample. We therefore study
\begin{equation}
    \Omega(\D)
    \triangleq
    \sup_{f\in\F}
    \left|\R^{\D}(f)-\R_{\emp}^{\D}(f)\right|.
    \label{eq: definition of uniform bound for c_ave}
\end{equation}

\paragraph{Bounded difference of $\Omega(\D)$}
Let $\D'$ be obtained from $\D$ by replacing only the $j$-th complete network example $(\h_j,\bgamma_j)$ by $(\h'_j,\bgamma'_j)$. Using
\begin{equation}
    \Bigl|
        \sup_f|a_f|-\sup_f|b_f|
    \Bigr|
    \leq
    \sup_f|a_f-b_f|
\end{equation}
and the triangle inequality, we obtain
\begin{align}
    \bigl|\Omega(\D)-\Omega(\D')\bigr|
    &\leq
    \sup_{f\in\F}\bigl|\R^{\D}(f)-\R^{\D'}(f)\bigr|
    \nonumber\\
    &\quad+
    \sup_{f\in\F}
    \bigl|\R_{\emp}^{\D}(f)-\R_{\emp}^{\D'}(f)\bigr|.
    \label{eq: risk difference decomposition for c_ave}
\end{align}
Since only one observation changes in each projection $\D_k$ and $|f_k|\leq\beta$ by Assumption~\ref{assump: bound}, one gets from~\eqref{eq: empirical mean for agent k} that
\begin{equation}
    \bigl|\bar f_k(\D_k)-\bar f_k(\D'_k)\bigr|
    \leq\frac{2\beta}{N}.
    \label{eq: center replacement bound appendix}
\end{equation}
For any $f\in\F$ and any $h$, it follows that
\begin{align}
    \bigl|
        \c_{\ave,f}^{\D}(h)
        -
        \c_{\ave,f}^{\D'}(h)
    \bigr|
    &=
    \left|
        \sum_{k=1}^K
        \pi_k
        \bigl(
            \bar f_k(\D'_k)
            -
            \bar f_k(\D_k)
        \bigr)
    \right|
    \nonumber\\
    &\leq
    \sum_{k=1}^K
    \pi_k
    \bigl|
        \bar f_k(\D'_k)
        -
        \bar f_k(\D_k)
    \bigr|
    \nonumber\\
    &\leq
    \frac{2\beta}{N}
    \label{eq: aggregate center replacement bound appendix}
\end{align}
where we used $\sum_{k=1}^K\pi_k=1$. By the triangle inequality and the Lipschitz property of $\Phi$,
\begin{equation}
    \sup_{f\in\F}\bigl|\R^{\D}(f)-\R^{\D'}(f)\bigr|
    \leq\frac{2\beta L_\Phi}{N}.
    \label{eq: bound for term (I)}
\end{equation}
For the empirical risk, we distinguish the unchanged samples from the replaced one. For an unchanged coordinate $n\neq j$, we have $(\h_n,\bgamma_n)=(\h'_n,\bgamma'_n)$, and therefore
\begin{align}
    &\left|
        \bgamma_n\c_{\ave,f}^{\D}(\h_n)
        -
        \bgamma'_n\c_{\ave,f}^{\D'}(\h'_n)
    \right|
    \nonumber\\
    &=
    \left|
        \c_{\ave,f}^{\D}(\h_n)
        -
        \c_{\ave,f}^{\D'}(\h_n)
    \right|
    \nonumber\\&\leq
    \frac{2\beta}{N}.
    \label{eq: unchanged PAC coordinate appendix}
\end{align}
At the replaced coordinate, $|f_k|\leq\beta$ and $|\bar f_k(\D_k)|\leq\beta$ imply
\begin{equation}
    |\c_{\ave,f}^{\D}(h)|\leq2\beta.
    \label{eq: aggregate centered range appendix}
\end{equation}
Hence, we have
\begin{align}
    &\left|
        \bgamma_j\c_{\ave,f}^{\D}(\h_j)
        -
        \bgamma'_j\c_{\ave,f}^{\D'}(\h'_j)
    \right|
    \nonumber\\
    &\leq
    \left|
        \c_{\ave,f}^{\D}(\h_j)
    \right|
    +
    \left|
        \c_{\ave,f}^{\D'}(\h'_j)
    \right|
    \nonumber\\&\leq4\beta.
    \label{eq: replaced PAC coordinate appendix}
\end{align}
Applying the Lipschitz property of $\Phi$ to these two cases yields
\begin{align}
    \bigl|\R_{\emp}^{\D}(f)-\R_{\emp}^{\D'}(f)\bigr|
    &\leq
    \frac{L_\Phi}{N}\biggl[
        (N-1)\frac{2\beta}{N}
        +4\beta
    \biggr]
    \nonumber\\
    &\leq
    \frac{6\beta L_\Phi}{N},
    \label{eq: empirical risk replacement bound appendix}
\end{align}
for every $f\in\F$. Combining the two terms in~\eqref{eq: bound for term (I)} and~\eqref{eq: empirical risk replacement bound appendix} gives
\begin{equation}
    \bigl|\Omega(\D)-\Omega(\D')\bigr|
    \leq\frac{8\beta L_\Phi}{N}.
    \label{eq: bounded difference for c_ave}
\end{equation}
This replacement argument uses only bounded local scores and independence across complete network examples. No factorization of the joint law $Q$ across agents is required, which implies that the conditional independence assumption of the local views across agents is not required.

\paragraph{Expectation of $\Omega(\D)$}
A direct symmetrization of the function class consisting of $\c_{\ave,f}^{\D}$, $\forall f\in\F$ would incorrectly treat a sample-dependent class as fixed. To address this issue, we instead embed every realized centered classifier $\c_{\ave,f}^{\D}$ in a \emph{deterministic} augmented class. Define
\begin{equation}
    a_f(\D)
    \triangleq
    \sum_{k=1}^{K}\pi_k\bar f_k(\D_k).
    \label{eq: aggregate empirical offset appendix}
\end{equation}
Since $|a_f(\D)|\leq\beta$ and
\begin{equation}
    \c_{\ave,f}^{\D}(h)
    =
    f_{\ave}(h)-a_f(\D),
    \quad
    f_{\ave}(h)
    \triangleq
    \sum_{k=1}^{K}\pi_k f_k(h_k),
\end{equation}
every realized centered classifier belongs to the fixed class
\begin{equation}
    \mathcal G_{\mathrm{off}}
    \triangleq
    \Bigl\{h\mapsto f_{\ave}(h)-a:
      f\in\F,\ a\in[-\beta,\beta]\Bigr\}.
    \label{eq: augmented offset class appendix}
\end{equation}
Define the corresponding fixed-class uniform deviation
\begin{equation}
    \Delta_{\mathrm{off}}(\D)
    \triangleq
    \sup_{g\in\mathcal G_{\mathrm{off}}}
    \biggl|
        \E\Phi(\bgamma g(\h))
        -
        \frac{1}{N}\sum_{n=1}^{N}
        \Phi(\bgamma_n g(\h_n))
    \biggr|.
    \label{eq: offset uniform deviation appendix}
\end{equation}
Then
\begin{equation}
    \Omega(\D)
    \leq
    \Delta_{\mathrm{off}}(\D).
    \label{eq: omega embedded in offset class appendix}
\end{equation}
To apply the contraction inequality in its standard form, define
\begin{equation}
    \widetilde{\Phi}(u)
    \triangleq
    \Phi(u)-\Phi(0).
    \label{eq: centered loss for contraction appendix}
\end{equation}
Then $\widetilde{\Phi}(0)=0$, $\widetilde{\Phi}$ remains $L_\Phi$-Lipschitz, and replacing $\Phi$ by $\widetilde{\Phi}$ does not change any population--empirical risk difference. Let
\begin{equation}
    \D^{\rm g}
    \triangleq
    \bigl\{(\h_n^{\rm g},\bgamma_n^{\rm g})\bigr\}_{n=1}^{N}
\end{equation}
be an independent ghost sample drawn from the same joint distribution $Q$. Since $\mathcal G_{\mathrm{off}}$ is independent of both $\D$ and $\D^{\rm g}$, Jensen's inequality gives
\begin{align}
    &\;\E_{\D}\Delta_{\mathrm{off}}(\D)\nonumber\\
    &=
    \E_{\D}
    \sup_{g\in\mathcal G_{\mathrm{off}}}
    \left|
        \E_{\D^{\rm g}}
        \frac{1}{N}\sum_{n=1}^{N}
        \left[
            \widetilde{\Phi}
            \!\left(
                \bgamma_n^{\rm g}g(\h_n^{\rm g})
            \right)
            -
            \widetilde{\Phi}
            \!\left(
                \bgamma_n g(\h_n)
            \right)
        \right]
    \right|
    \nonumber\\
    &\leq
    \E_{\D,\D^{\rm g}}
    \sup_{g\in\mathcal G_{\mathrm{off}}}
    \left|
        \frac{1}{N}\sum_{n=1}^{N}
        \left[
            \widetilde{\Phi}
            \!\left(
                \bgamma_n^{\rm g}g(\h_n^{\rm g})
            \right)
            -
            \widetilde{\Phi}
            \!\left(
                \bgamma_n g(\h_n)
            \right)
        \right]
    \right|.
    \label{eq: ghost sample symmetrization appendix}
\end{align}
Introducing $N$ independent Rademacher variables $\bm r_1,\ldots,\bm r_N$ and using the usual symmetrization argument, we obtain
\begin{align}
    \E\Omega(\D)
    &\leq
    \E\Delta_{\mathrm{off}}(\D)
    \nonumber\\
    &\leq
    2\E_{\D,\bm r}
    \sup_{g\in\mathcal G_{\mathrm{off}}}
    \left|
        \frac{1}{N}\sum_{n=1}^{N}
        \bm r_n
        \widetilde{\Phi}
        \!\left(\bgamma_n g(\h_n)\right)
    \right|
    \nonumber\\
    &\leq
    4L_\Phi
    \E_{\D,\bm r}
    \sup_{g\in\mathcal G_{\mathrm{off}}}
    \left|
        \frac{1}{N}\sum_{n=1}^{N}
        \bm r_n\bgamma_n g(\h_n)
    \right|
    \label{eq: augmented class symmetrization}
\end{align}
where the last step follows from the contraction inequality for Rademacher averages~\cite{bartlett2002rademacher,mohri2018foundations}. For $g(h)=f_{\ave}(h)-a\in\mathcal G_{\mathrm{off}}$, using the definition of $f_\ave$ in~\eqref{eq: f_ave}, we further have
\begin{align}
    &\sup_{g\in\mathcal G_{\mathrm{off}}}
    \left|
        \frac{1}{N}\sum_{n=1}^{N}
        \bm r_n\bgamma_n g(\h_n)
    \right|
    \nonumber\\
    &\leq
    \sum_{k=1}^K\pi_k
    \sup_{f_k\in\F_k}
    \left|
        \frac{1}{N}\sum_{n=1}^{N}
        \bm r_n\bgamma_n f_k(\h_{k,n})
    \right|
    \nonumber\\
    &\quad+
    \beta
    \left|
        \frac{1}{N}\sum_{n=1}^{N}
        \bm r_n\bgamma_n
    \right|.
    \label{eq: offset complexity decomposition appendix}
\end{align}
Conditional on the sample set $\D$, the variables $\bm r_n\bgamma_n$ are again independent Rademacher signs. Therefore, using the definition $\rho=\sum_k\pi_k\rho_k$ from Appendix~\ref{appendix: training auxiliary quantities},
\begin{align}
    &\;\E_{\D,\bm r}
    \sup_{g\in\mathcal G_{\mathrm{off}}}
    \left|
        \frac{1}{N}\sum_{n=1}^{N}
        \bm r_n\bgamma_n g(\h_n)
    \right|
    \nonumber\\
    &\leq
    \rho
    +
    \beta\E_{\bm r}
    \left|
        \frac{1}{N}\sum_{n=1}^{N}\bm r_n
    \right|
    \nonumber\\
    &\leq
    \rho+\frac{\beta}{\sqrt N}
    \label{eq: augmented offset complexity}
\end{align}
Thus,
\begin{equation}
    \E\Omega(\D)
    \leq
    4L_\Phi\rho
    +
    \frac{4\beta L_\Phi}{\sqrt N}.
    \label{eq: bound for the expectation for c_ave}
\end{equation}

\paragraph{High-probability bound}
Applying McDiarmid's inequality~\cite{mcdiarmid1989on} to \eqref{eq: bounded difference for c_ave} gives, for every $x>0$,
\begin{equation}
    \P\bigl(\Omega(\D)-\E\Omega(\D)\geq x\bigr)
    \leq
    \exp\left\{
        -\frac{Nx^2}
        {32\beta^2L_\Phi^2}
    \right\}.
    \label{eq: McDiarmid inequality for c_ave under general Phi}
\end{equation}
Taking
\begin{equation}
    x=
    4\beta L_\Phi
    \sqrt{\frac{2\log(1/\epsilon)}{N}}
\end{equation}
and combining with \eqref{eq: bound for the expectation for c_ave}, we obtain, with probability at least $1-\epsilon$, uniformly for all $f\in\F$,
\begin{align}
    \bigl|\R^{\D}(f)-\R_{\emp}^{\D}(f)\bigr|
    &\leq
    4L_\Phi\rho
    +
    \frac{4\beta L_\Phi}{\sqrt N}
    \nonumber\\
    &\quad+
    4\beta L_\Phi
    \sqrt{\frac{2\log(1/\epsilon)}{N}}.
    \label{eq: margin bound for c_ave under general Phi}
\end{align}
We now take $\Phi=\Phi_\eta$, the $\eta$-margin loss in \eqref{eq: definition of eta-margin loss function}, for which $L_\Phi=1/\eta$. Using the definition of $\mathcal R(\epsilon,N)$ in~\eqref{eq: definition of R(e,N)}, the preceding display becomes
\begin{equation}
    \bigl|\R^{\D}(f)-\R_{\emp}^{\D}(f)\bigr|
    \leq
    \frac{4\rho}{\eta}
    +
    \mathcal R(\epsilon,N).
    \label{eq: margin bound for Phi_eta}
\end{equation}
For a given training set $\D$, $\R^{\D}(f)$ is the expected $\eta$-margin loss of the deployed classifier $\c_{\ave,f}^{\D}$, while $\R_{\emp}^{\D}(f)$ denotes its empirical value. Hence~\eqref{eq: property of eta-margin loss} gives
\begin{equation}
    P_e(\c_{\ave,f}^{\D})
    \leq
    P_{e,\emp}^{\eta}(\c_{\ave,f}^{\D})
    +
    \frac{4\rho}{\eta}
    +
    \mathcal R(\epsilon,N).
    \label{eq: PAC bound for general f}
\end{equation}
The high-probability event is uniform over $f\in\F$ and depends only on the aligned network training set $\D$. Therefore, given a fixed $\D$, \eqref{eq: PAC bound for general f} may be evaluated at any data-dependent outputs $\widehat{\f}_k=\mathcal A_k(\D_k,\mathcal{S}_k)\in\F_k$, for every realization of the local training randomness, without any empirical optimality condition. This establishes \eqref{eq: margin-based generalization bound for c_ave}. Applying the same argument to the projected sample $\D_k$ for a fixed agent $k$ gives~\eqref{eq: margin-based generalization bound for c_k} and completes the proof.
\end{proof}

\noindent\emph{Remark on a refinement.}
The preceding proof is formulated for a generic $L_\Phi$-Lipschitz loss and controls the effect of replacing one complete network example through the bounded score range. This produces the factor $\beta L_\Phi$ in the confidence term. For the $\eta$-margin loss $\Phi_\eta$, a sharper concentration argument is available because $0\leq\Phi_\eta\leq1$.

Specifically, after embedding the empirically centered classifiers in the fixed offset-augmented class $\mathcal G_{\mathrm{off}}$, define the one-sided population--empirical deviation
\begin{equation}
    \Delta_{\mathrm{off},+}(\D)
    \triangleq
    \sup_{g\in\mathcal G_{\mathrm{off}}}
    \left\{
        \E\Phi_\eta(\bgamma g(\h))
        -
        \frac{1}{N}\sum_{n=1}^{N}
        \Phi_\eta(\bgamma_n g(\h_n))
    \right\}.
    \label{eq: one sided offset deviation appendix}
\end{equation}
Standard one-sided symmetrization and the non-absolute contraction inequality in~\cite{mohri2018foundations} give
\begin{equation}
    \E\Delta_{\mathrm{off},+}(\D)
    \leq
    \frac{2}{\eta}
    \left(
        \rho+\frac{\beta}{\sqrt N}
    \right).
\end{equation}
Moreover, replacing one complete network example changes this fixed-class uniform deviation by at most $1/N$, since the $\eta$-margin loss takes values in $[0,1]$. Thus, the bounded-difference constant for $\Delta_{\mathrm{off},+}(\D)$ is $1/N$, and McDiarmid's inequality yields the alternative confidence term
\begin{equation}
    \sqrt{\frac{\log(1/\epsilon)}{2N}}.
\end{equation}
Consequently, the collaborative bound associated with the $\eta$-margin loss can be refined to
\begin{equation}
    P_e(\c_\ave)
    \leq
    P_{e,\emp}^{\eta}(\c_\ave)
    +
    \frac{2\rho}{\eta}
    +
    \frac{2\beta}{\eta\sqrt N}
    +
    \sqrt{\frac{\log(1/\epsilon)}{2N}}.
\end{equation}
The analogous refinement for a local classifier is
\begin{equation}
    P_e(\c_k)
    \leq
    P_{e,\emp}^{\eta}(\c_k)
    +
    \frac{2\rho_k}{\eta}
    +
    \frac{2\beta}{\eta\sqrt N}
    +
    \sqrt{\frac{\log(1/\epsilon)}{2N}}.
\end{equation}
These refinements use the boundedness of $\Phi_\eta$ and a one-sided fixed-class argument. Theorem~\ref{theorem: generalization bound for c_k and c_ave} retains the bound obtained from the generic Lipschitz-loss argument above.

\section{Relation to Social Learning}
\label{appendix: classification error using SL rule}

This appendix returns to the prediction dynamics and relates the DeGroot collaboration rule~\eqref{eq: distributed learning rule} used in the main paper to social learning (SL) based on geometric pooling. We will first formulate the corresponding static-observation specialization. Then, we establish its finite-round classification guarantee, and compare the two update mechanisms.

\subsection{Static Observation Specialization}

A defining feature of the DeGroot model~\eqref{eq: distributed learning rule} is that agents update their opinions $\blambda_{k,t}$ solely by averaging those of their neighbors, without acquiring new observations over time. In contrast, SL rules are designed for inference in \emph{streaming} environments, where agents continuously receive new observations and combine them with information aggregated from their neighbors. To connect these two inference mechanisms, we consider the geometric-pooling SL mechanism~\cite{lalitha2018social}, specialized to the learned local decision statistics used in social machine learning~\cite{bordignon2023learning,hu2025non-asymptotic}:
\begin{equation}\label{eq: geometric Sl rule}
    \blambda_{k,t}^{\sf SL}
    =
    \sum_{\ell=1}^{K}a_{\ell k}
    \bigl(
        \blambda_{\ell,t-1}^{\sf SL}
        +
        \c_\ell(\h_{\ell,t})
    \bigr)
\end{equation}
where $\h_{\ell,t}$ denotes the private observation received by agent $\ell$ at round $t$. The linear recursion in~\eqref{eq: geometric Sl rule} is the log-domain form of the corresponding geometric SL update. Unrolling it gives
\begin{equation}\label{eq: SL recursion}
    \blambda^{\sf SL}_{k,t}
    =
    \sum_{\ell=1}^K[A^{t}]_{\ell k}\blambda^{\sf SL}_{\ell,0}
    +
    \sum_{i=1}^t\sum_{\ell=1}^K
    [A^{t+1-i}]_{\ell k}\c_\ell(\h_{\ell,i}).
\end{equation}
We assume no prior information about the class label, so that $\blambda^{\sf SL}_{k,0}=0$. In the static observation specialization, $\h_{\ell,t}=\h_\ell^*$ for every $t$, and therefore
\begin{equation}\label{eq: decision statistic at time t under SL}
    \blambda_{k,t}^{\sf SL}
    =
    \sum_{\ell=1}^K\sum_{i=1}^{t}
    [A^i]_{\ell k}\c_\ell(\h_\ell^*).
\end{equation}
As the same testing example is reused at every round, the repeated terms in~\eqref{eq: decision statistic at time t under SL} do not represent new temporal evidence. Instead, they repeatedly accumulate the statistics associated with the same $K$ local views.

Since $[A^i]_{\ell k}\to\pi_\ell$, its Ces\`aro averages satisfy
\begin{align}
\nonumber
    \lim_{t\to\infty}\frac{1}{t}\blambda_{k,t}^{\sf SL}
    &=
    \sum_{\ell=1}^K
    \lim_{t\to\infty}
    \frac{1}{t}\sum_{i=1}^{t}
    [A^i]_{\ell k}\c_\ell(\h_\ell^*)\\
    \label{eq: convergence of decision statistic under SL rule}
    &=
    \sum_{\ell=1}^K\pi_\ell \c_\ell(\h_\ell^*)
    =
    \c_\ave(\h^*).
\end{align}
Hence, whenever $\c_\ave(\h^*)\neq0$, the prediction of every agent eventually agrees with the collective prediction based on the aggregate classifier $\c_\ave$. The sufficient-communication guarantee of Theorem~\ref{theorem: classification error} therefore applies to the limiting SL prediction as well. For the finite-round analysis, we define
\begin{equation}
    \widehat{\bgamma}_{k,t}^{\sf SL}
    =
    \sign(\blambda_{k,t}^{\sf SL}),
    \quad
    \M_{k,t}^{\sf SL}
    =
    \{\bgamma^*\blambda_{k,t}^{\sf SL}\leq0\}.
\end{equation}
Similar to Theorems~\ref{theorem: classification error}--\ref{theorem: finite bit communication}, we analyze the conditional probability $\P[\M_{k,t}^{\sf SL}\givensmall\Ttrain]$. To quantify the cumulative deviation of the finite-round weights from the Perron vector, define
\begin{equation}\label{eq: cumulative mixing constant}
    \Bmix
    \triangleq
    \max_{k\in\K}
    \sum_{i=1}^{\infty}
    \sum_{\ell=1}^K
    |[A^i]_{\ell k}-\pi_\ell|,
\end{equation}
which is finite under Assumption~\ref{assump: network}. For every $\sigma\in(\sigma_A,1)$ and corresponding $C(A,\sigma)$ in~\eqref{eq: convergence constant main}, it holds that
\begin{equation}\label{eq: cumulative mixing constant upper bound}
    \Bmix
    \leq
    \frac{K C(A,\sigma)\sigma}{1-\sigma}.
\end{equation}
With this definition, we establish the classification guarantee for the SL rule.
\begin{proposition}[\textbf{Static-observation SL classification guarantee}]
\label{prop: classification error using SL rule}
Suppose Assumptions~\ref{assump: bound},~\ref{assump: network}, and~\ref{assump: class conditional ATE} hold, and let $\delta>0$. Define
\begin{equation}\label{eq: kappa}
    \kappa
    \triangleq
    2\beta\Bmix.
\end{equation}
Almost surely on $\C_\delta$, for every agent $k\in\K$ and every integer $t>\kappa/\delta$,
\begin{align}
    \P(\M_{k,t}^{\sf SL}\given\Ttrain)
    \leq
    \exp\Biggl\{
        -\frac{(t\delta-\kappa)^2}
        {2\tau_{\max}\beta^2
        \sum_{\ell=1}^K
        \bigl(
            \sum_{i=1}^t[A^i]_{\ell k}
        \bigr)^2}
    \Biggr\}.
    \label{eq: P_e for t-th round using SL rule}
\end{align}
This bound converges to the sufficient-communication bound for $\c_\ave$ as $t\to\infty$.
\end{proposition}

\subsection{Proof of Proposition~\ref{prop: classification error using SL rule}}

\begin{proof}
The proof reuses the conditioning convention and ATE bounded-difference inequality from Appendix~\ref{appendix: sufficient communication}. Relative to the finite-round DeGroot analysis presented in Appendix~\ref{appendix: finite round communication}, the only change is that the SL statistic uses the cumulative weights $\sum_{i=1}^{t}[A^i]_{\ell k}$ rather than the single-round weights $[A^t]_{\ell k}$.

Fix a realization of the training phase in $\C_\delta$. Conditional on $\Ttrain$, the trained classifiers are fixed, while the fresh testing vector follows $P_\gamma$ under $\bgamma^*=\gamma$. Since the same testing view is reused at every round, its repeated appearances are not treated as independent coordinates. Instead, changing view $\ell$ changes all of its accumulated contributions simultaneously. From~\eqref{eq: decision statistic at time t under SL}, changing only coordinate $\ell$ of the testing vector changes the statistic by at most
\begin{align}
\nonumber
    \Bigl|
        \blambda_{k,t}^{\sf SL}(\h^*)
        -
        \blambda_{k,t}^{\sf SL}(\widehat{\h})
    \Bigr|
    &\leq
    \sum_{i=1}^{t}[A^i]_{\ell k}
    \bigl|
        c_\ell(\h_\ell^*)
        -
        c_\ell(\widehat{\h}_\ell)
    \bigr|\\
    &\leq
    2\beta
    \sum_{i=1}^{t}[A^i]_{\ell k}.
\end{align}
Thus Lemma~\ref{lemma: ATE bounded differences} applies with coordinate oscillations
\begin{equation}
    d_{\ell,k,t}^{\sf SL}
    =
    2\beta
    \sum_{i=1}^{t}[A^i]_{\ell k}.
\end{equation}
We next control the class-conditional mean. For either $\gamma\in\Gamma$,
\begin{align}
\nonumber
    &\;\Bigl|
        \E_{\h^*}^{(\gamma)}
        \blambda_{k,t}^{\sf SL}(\h^*)
        -
        t\mu^\gamma(\widehat f)
    \Bigr|\\
\nonumber
    &=
    \biggl|
        \sum_{i=1}^{t}
        \sum_{\ell=1}^K
        \bigl(
            [A^i]_{\ell k}-\pi_\ell
        \bigr)
        \mu^\gamma_\ell(\widehat f_\ell)
    \biggr|\\
     \nonumber
    &\leq
    2\beta
    \sum_{i=1}^{t}
    \sum_{\ell=1}^K
    \bigl|
        [A^i]_{\ell k}-\pi_\ell
    \bigr|
    \\
    & \leq
    2\beta\Bmix
    =
    \kappa
    \label{eq: SL classification mean appendix}
\end{align}
where $|\mu^\gamma_\ell(\widehat f_\ell)|\leq2\beta$ follows from $|c_\ell|\leq2\beta$. For $\bgamma^*=+1$, the event $\C_\delta$ implies $\mu^+(\widehat f)>\delta$, so for $t>\kappa/\delta$,
\begin{equation}
    \E_{\h^*}^{(+1)}
    \blambda_{k,t}^{\sf SL}(\h^*)
    >
    t\delta-\kappa
    >
    0.
\end{equation}
The lower-tail form of Lemma~\ref{lemma: ATE bounded differences} therefore gives
\begin{align}
    &\;\P\bigl(
        \M_{k,t}^{\sf SL}
        \givenbig
        \bgamma^*=+1,\Ttrain
    \bigr)
    \nonumber\\
    &\leq
    \exp\Biggl\{
        -\frac{(t\delta-\kappa)^2}
        {2\tau_{+1}\beta^2
        \sum_{\ell=1}^K
        \bigl(
            \sum_{i=1}^{t}[A^i]_{\ell k}
        \bigr)^2}
    \Biggr\}.
\label{eq: condition probability for t-th round single-sample case +1---using SL}
\end{align}
A similar bound is obtained for the case $\bgamma^*=-1$ by replacing $\tau_{+1}$ with $\tau_{-1}$. The inequality in~\eqref{eq: P_e for t-th round using SL rule} is then established using the definition of $\tau_{\max}$ in~\eqref{eq: tau_max}. Finally,
\begin{equation}
    \lim_{t\to\infty}
    \frac{1}{t}
    \sum_{i=1}^{t}[A^i]_{\ell k}
    =
    \pi_\ell,
\end{equation}
which implies
\begin{equation}
    \lim_{t\to\infty}
    \frac{1}{t^2}
    \sum_{\ell=1}^K
    \biggl(
        \sum_{i=1}^{t}[A^i]_{\ell k}
    \biggr)^2
    =
    \sum_{\ell=1}^K\pi_\ell^2.
\end{equation}
Since $(t\delta-\kappa)^2/t^2\to\delta^2$, we obtain
\begin{equation}\label{eq: convergence to the bound for c_ave under SL}
    \lim_{t\to\infty}
    \frac{(t\delta-\kappa)^2}
    {2\tau_{\max}\beta^2
    \sum_{\ell=1}^K
    \bigl(
        \sum_{i=1}^{t}[A^i]_{\ell k}
    \bigr)^2}
    =
    \frac{\delta^2}
    {2\tau_{\max}\beta^2
    \sum_{\ell=1}^K\pi_\ell^2},
\end{equation}
which is the exponent that appears in the bound~\eqref{eq: P_e for single-sample case under delta-margin consistent training} for the aggregate classifier $\c_\ave$ under sufficient communication. This completes the proof.
\end{proof}

\subsection{Comparison with the DeGroot Rule}

By Jensen's inequality, the following holds:
\begin{align}
\nonumber
    \sum_{\ell=1}^K
    \biggl(
        \sum_{i=1}^{t}[A^i]_{\ell k}
    \biggr)^2
    &\leq
    t
    \sum_{\ell=1}^K
    \sum_{i=1}^{t}
    \bigl([A^i]_{\ell k}\bigr)^2\\
    &\leq
    t
    \sum_{\ell=1}^K
    \sum_{i=1}^{t}
    [A^i]_{\ell k}
    =
    t^2
\end{align}
where the last equality uses left-stochasticity of $A$. Hence, for every agent,
\begin{equation}
    \P\bigl(\M_{k,t}^{\sf SL}\given\Ttrain\bigr)
    \leq
    \exp\Biggl\{
        -\frac{(t\delta-\kappa)^2}
        {2\tau_{\max}\beta^2t^2}
    \Biggr\}
\end{equation}
almost surely on $\C_\delta$ whenever $t>\kappa/\delta$.

The finite-round DeGroot and static observation SL statistics are generally different: the former uses the weights $[A^t]_{\ell k}$ in~\eqref{eq: decision statistic at time t}, whereas the latter accumulates $\sum_{i=1}^{t}[A^i]_{\ell k}$ in~\eqref{eq: decision statistic at time t under SL}. Accordingly, neither finite-round analytical bound is asserted to vary monotonically with $t$, because its denominator depends on the corresponding finite-round mixing weights. Nevertheless,~\eqref{eq: common decision statistic} and~\eqref{eq: convergence of decision statistic under SL rule} show that, whenever $\c_\ave(\h^*)\neq0$, both mechanisms eventually yield the same collective prediction $\widehat{\bgamma}$ in~\eqref{eq: label of single-sample classification}. Their error bounds likewise converge to the sufficient-communication bound.

A particularly transparent connection arises for the fully connected consensus matrix $A=\pi\mathbbm{1}^\top$. In this case, $A^m=A$ for every integer $m\geq 1$, and for every $t\geq1$,
\begin{equation}
    \blambda_{k,t}
    =
    \sum_{\ell=1}^K
    \pi_\ell\c_\ell(\h_\ell^*)
    =
    \c_\ave(\h^*),
\end{equation}
whereas
\begin{equation}
    \blambda_{k,t}^{\sf SL}
    =
    t
    \sum_{\ell=1}^K
    \pi_\ell\c_\ell(\h_\ell^*)
    =
    t\c_\ave(\h^*).
\end{equation}
Therefore, the two statistics have the same sign and produce the same prediction after one communication round. Moreover, exact mixing is achieved after one communication round, since
\begin{equation}
    [A^t]_{\ell k}=\pi_\ell,
    \quad
    t\geq1.
\end{equation}
Hence, for $t\geq1$, using the exact identity $[A^t]_{\ell k}=\pi_\ell$ directly in the finite-round DeGroot argument makes the transient mean error identically zero, and the resulting specialized concentration bound coincides with the sufficient-communication bound. For the SL rule, $\Bmix=0$ and hence $\kappa=0$, while
\begin{equation}
    \sum_{i=1}^{t}[A^i]_{\ell k}
    =
    t\pi_\ell.
\end{equation}
Its finite-round bound therefore reduces to the same sufficient-communication bound for every $t\geq1$. Finally, the two recursions use different initial conditions because they represent different inference semantics. For SL, $\blambda_{k,0}^{\sf SL}=0$ is the state before the first observation is incorporated. For the DeGroot rule, $\blambda_{k,0}=\c_k(\h_k^*)$ is the local decision statistic available before communication. This discrepancy at $t=0$ does not affect the preceding finite-round or limiting comparison.

\section{Raw-Score Counterparts}
\label{appendix: raw score counterparts}

The main analysis uses the empirically centered local scores $\c_k$ in~\eqref{eq: trained classifier}. As discussed in Proposition~\ref{prop: effect of additive centering}, empirical centering acts as a threshold alignment mechanism rather than a structural requirement for collaboration. This appendix provides the corresponding results when the learned scores $\widehat{\f}_k$ are used directly, without empirical centering.

The raw aggregate score $f_{\ave}$ and its class-conditional means $\bmu_{\rm raw}^{\gamma}$ are already defined in \eqref{eq: raw aggregate classifier}. Moreover, the raw expected zero-threshold margin is $\Delta_{\rm raw}(f)$ in~\eqref{eq: raw margin}. Since the arguments closely parallel those developed for the centered scores $c_{\ave}$ in Appendices~\ref{appendix: learning complexity definitions}--\ref{appendix: classification error using SL rule}, we state only the changes needed for the raw-score counterparts.

\subsection{Raw-Score Training Guarantee}
\label{appendix: raw margin training proof}

For $0\leq\delta<\beta$ and a generic risk level $r$, we define
\begin{equation}
\label{eq: raw E feasibility quantity}
    \mathdutchcal E_\Phi^{\mathrm{raw}}(r,\delta)
    \triangleq
    \frac{
        \Phi(\delta)+\Phi(\beta)-2r
    }{16L_\Phi}.
\end{equation}
For the learned functions $\widehat{\f}=\{\widehat{\f}_1,\dots,\widehat{\f}_K\}$, we denote the raw-score $\delta$-margin training event by
\begin{align}
\label{eq: raw training event}
    \C_\delta^{\mathrm{raw}}
    &\triangleq
    \bigl\{
        \Delta_{\rm raw}(\widehat{\f})>\delta
    \bigr\}
    \nonumber\\
    &=
    \bigl\{
        \bmu_{\rm raw}^{+}(\widehat{\f})>\delta,\,
        \bmu_{\rm raw}^{-}(\widehat{\f})<-\delta
    \bigr\}
\end{align}
where the equality follows from~\eqref{eq: mu_mid},~\eqref{eq: raw midpoint and signed half difference} and \eqref{eq: raw margin}. This is the raw-score counterpart of the centered condition \eqref{eq: delta-margin consistent training condition}. Similar to Proposition~\ref{prop: P_c_delta}, we can establish a bound on the achievability of~\eqref{eq: raw training event} as follows.

\begin{corollary}[\textbf{Raw-score $\delta$-margin consistent training}]
\label{corollary: raw margin training}
Suppose Assumptions~\ref{assump: risk function}--\ref{assump: network} hold and the approximate optimization condition \eqref{eq: approximate ERM condition} is satisfied. Let
\begin{equation}
    r\triangleq\mathsf R^o+\varepsilon_{\opt}
\end{equation}
where $\varepsilon_{\opt}$ is defined in \eqref{eq: network optimization error}. If
\begin{equation}
    0\leq\delta<\beta,
    \quad
    r<
    \frac{\Phi(\delta)+\Phi(\beta)}{2},
    \quad
    \rho<
    \mathdutchcal E_\Phi^{\mathrm{raw}}(r,\delta)
\end{equation}
then
\begin{equation}
\label{eq: raw margin training probability}
    \P(\C_\delta^{\mathrm{raw}})
    \geq
    1-
    \exp\Biggl\{
        -\frac{8N_{\max}}{\alpha^2\beta^2}
        \Bigl(
            \mathdutchcal E_\Phi^{\mathrm{raw}}(r,\delta)-\rho
        \Bigr)^2
    \Biggr\}.
\end{equation}
\end{corollary}
\begin{proof}
The analysis on the approximate optimization remains unchanged from the proof of Proposition~\ref{prop: P_c_delta}. With the argument in Appendix~\ref{appendix: proof approximate ERM lemma}, we have
\begin{equation}
\label{eq: raw approximate risk reduction}
    \R(\widehat{\f})
    \leq
    r+2U,
    \quad
    r=\mathsf R^o+\varepsilon_{\opt}
\end{equation}
where $U$ is the uniform risk deviation in \eqref{eq: uniform risk deviation appendix}. It remains to relate failure of the $\delta$-margin condition~\eqref{eq: raw training event} to the population surrogate risk. We introduce the following two class-conditional risks:
\begin{align}
    \R^{+}(f)
    &\triangleq
    \sum_{k=1}^K
    \pi_k
    \E^{(+1)}
    \Phi(f_k(\h_k)),
    \\
    \R^{-}(f)
    &\triangleq
    \sum_{k=1}^K
    \pi_k
    \E^{(-1)}
    \Phi(-f_k(\h_k)).
\end{align}
Under the uniform class prior, $\R(f)=[\R^{+}(f)+\R^{-}(f)]/2$. By the same Jensen argument used in \eqref{eq: risk separation Jensen appendix},
\begin{equation}
    \R^{+}(f)
    \geq
    \Phi\bigl(\bmu_{\rm raw}^{+}(f)\bigr),
    \quad
    \R^{-}(f)
    \geq
    \Phi\bigl(-\bmu_{\rm raw}^{-}(f)\bigr).
\end{equation}
Moreover, Assumption~\ref{assump: bound} and the non-increasing property of $\Phi$ imply
\begin{equation}
    \R^{+}(f)\geq\Phi(\beta),
    \quad
    \R^{-}(f)\geq\Phi(\beta).
\end{equation}
If $\C_\delta^{\mathrm{raw}}$ fails, one of the two class-conditional risks is at least $\Phi(\delta)$, while the other is at least $\Phi(\beta)$. Therefore,
\begin{equation}
\label{eq: raw failure risk inclusion}
    \overline{\C_\delta^{\mathrm{raw}}}
    \subseteq
    \left\{
        \R(\widehat{\f})
        \geq
        \frac{\Phi(\delta)+\Phi(\beta)}{2}
    \right\}.
\end{equation}
Combining~\eqref{eq: raw failure risk inclusion} with \eqref{eq: raw approximate risk reduction} gives, on $\overline{\C_\delta^{\mathrm{raw}}}$,
\begin{align}
    U
    \geq
    \frac{
        \Phi(\delta)+\Phi(\beta)-2r
    }{4}=
    4L_\Phi
    \mathdutchcal E_\Phi^{\mathrm{raw}}(r,\delta).
\end{align}
Applying the uniform deviation bound \eqref{eq: uniform risk bound appendix} at this threshold proves \eqref{eq: raw margin training probability}.
\end{proof}
\noindent Compared with Proposition~\ref{prop: P_c_delta}, the change in the bound~\eqref{eq: raw margin training probability} occurs in the reduction from margin event failure to population risk. The centered analysis additionally controls the empirical training mean, whereas the raw margin in~\eqref{eq: raw margin} is determined directly by the two raw class-conditional means.

\subsection{Sufficient-Communication Counterpart}
\label{appendix: raw sufficient communication}

Set $f=\widehat{\f}$ in the raw aggregate score \eqref{eq: raw aggregate classifier}, and define
\begin{equation}
\label{eq: raw sufficient event}
    \M^{\mathrm{raw}}
    \triangleq
    \{
        \bgamma^*
        f_{\ave}(\h^*)\leq0
    \}.
\end{equation}
This is the raw-score counterpart of the event $\M$ in~\eqref{eq: P_e decomposition}. When $\C_\delta^{\mathrm{raw}}$ is true, the class-conditional mean of $f_{\ave}(\h^*)$ lies on the correct side of zero by more than $\delta$. Moreover, changing only testing view $k$ changes $f_{\ave}(\h^*)$ by at most $2\beta\pi_k$, which is the same coordinate oscillation used for the centered aggregate in Appendix~\ref{appendix: classification error}. Applying Lemma~\ref{lemma: ATE bounded differences} under the two class-conditional laws $P_\gamma$ and using the definition of $\tau_{\max}$ gives
\begin{equation}
\label{eq: raw sufficient conditional error}
    \P(
        \M^{\mathrm{raw}}\givensmall\Ttrain
    )
    \leq
    \exp\Biggl\{
        -\frac{\delta^2}
        {2\tau_{\max}\beta^2
        \sum_{k=1}^K\pi_k^2}
    \Biggr\}
\end{equation}
almost surely on $\C_\delta^{\mathrm{raw}}$. Thus, for the same prescribed margin $\delta$, the prediction-side bound~\eqref{eq: raw sufficient conditional error} has exactly the same form as its centered counterpart in~\eqref{eq: P_e for single-sample case under delta-margin consistent training}. Applying the same tower-property decomposition as in \eqref{eq: P_e decomposition}, with $\M$ and $\C_\delta$ replaced by $\M^{\mathrm{raw}}$ and $\C_\delta^{\mathrm{raw}}$, respectively, gives
\begin{align}
    \P(\M^{\mathrm{raw}})
    &=
    \E\!\left[
        \P(\M^{\mathrm{raw}}\givensmall\Ttrain)
    \right]
    \nonumber\\
    &\leq
    \E\!\left[
        \mathsf{1}[\C_\delta^{\mathrm{raw}}]
        \P(\M^{\mathrm{raw}}\givensmall\Ttrain)
    \right]
    +
    \P(\overline{\C_\delta^{\mathrm{raw}}}).
\label{eq: raw sufficient decomposition}
\end{align}
Combining~\eqref{eq: raw sufficient decomposition} with \eqref{eq: raw sufficient conditional error} and the raw-score training guarantee~\eqref{eq: raw margin training probability} gives the corresponding unconditional guarantee.

\subsection{Finite-Round Counterpart}
\label{appendix: raw finite round}

Without empirical centering, the finite-round statistic corresponding to \eqref{eq: decision statistic at time t} becomes
\begin{equation}
\label{eq: raw finite round statistic}
    \blambda_{k,t}^{\mathrm{raw}}
    \triangleq
    \sum_{\ell=1}^K
    [A^t]_{\ell k}
    \widehat{\f}_\ell(\h_\ell^*),
\end{equation}
with
\begin{equation}
\label{eq: raw finite round event}
    \M_{k,t}^{\mathrm{raw}}
    \triangleq
    \{
        \bgamma^*
        \blambda_{k,t}^{\mathrm{raw}}
        \leq0
    \}.
\end{equation}
The finite-round mixing argument pertaining to $\blambda_{k,t}^{\mathrm{raw}}$ follows Appendix~\ref{appendix: finite round communication}. The coordinate oscillation remains $2\beta[A^t]_{\ell k}$. The only change in the expectation term follows from
\begin{equation}
    \left|
        \E^{(\gamma)}
        \widehat{\f}_\ell(\h_\ell)
    \right|
    \leq\beta,
\end{equation}
instead of the bound $2\beta$ for empirically centered scores $\c_\ell(\h_\ell)$. Using the matrix-power estimate in~\eqref{eq: convergence constant main},
\begin{equation}
\label{eq: raw finite mean transient}
    \left|
        \E^{(\gamma)}
        \blambda_{k,t}^{\mathrm{raw}}
        -
        \bmu_{\rm raw}^{\gamma}(\widehat{\f})
    \right|
    \leq
    \beta K C(A,\sigma)\sigma^t.
\end{equation}
Accordingly, define the raw-score counterpart of the residual margin in~\eqref{eq: finite round residual margin} by
\begin{equation}
\label{eq: raw finite round residual}
    r_t^{\mathrm{raw}}(\delta)
    \triangleq
    \delta-\beta K C(A,\sigma)\sigma^t.
\end{equation}
Repeating the argument used to establish Theorem~\ref{theorem: finite-time classification error} gives, for every agent $k$ and every $t\geq0$ such that $r_t^{\mathrm{raw}}(\delta)>0$,
\begin{equation}
\label{eq: raw finite-time conditional error}
    \P(
        \M_{k,t}^{\mathrm{raw}}\givensmall\Ttrain
    )
    \leq
    \exp\Biggl\{
        -\frac{
            [r_t^{\mathrm{raw}}(\delta)]^2
        }
        {2\tau_{\max}\beta^2
        \sum_{\ell=1}^K([A^t]_{\ell k})^2}
    \Biggr\}
\end{equation}
almost surely on $\C_\delta^{\mathrm{raw}}$. Therefore, relative to the bound in~\eqref{eq: P_e for t-th round single-sample case under delta-margin consistent training}, only the residual margin changes:
\begin{equation}
    r_t(\delta)
    =
    \delta-2\beta K C(A,\sigma)\sigma^t
    \longrightarrow
    r_t^{\mathrm{raw}}(\delta)
    =
    \delta-\beta K C(A,\sigma)\sigma^t.
\end{equation}
The corresponding unconditional guarantee follows by combining \eqref{eq: raw finite-time conditional error} with \eqref{eq: raw margin training probability}, in the same manner as the centered finite-round decomposition \eqref{eq: P_e equation at t-th communication round}.

\subsection{Finite-Precision Counterpart}
\label{appendix: raw finite precision}

We next consider the raw-score counterpart of the finite-precision recursion~\eqref{eq: finite bit recursion theory}. Since $|\widehat{\f}_k(h)|\leq\beta$ by Assumption~\ref{assump: bound}, the raw inputs can be quantized over $[-\beta,\beta]$ rather than the centered range $[-2\beta,2\beta]$ in~\eqref{eq: centered score quantization range}. Accordingly, we employ a narrower raw-score alphabet
\begin{equation}
    \mathcal V_b^{\mathrm{raw}}
    \triangleq
    \left\{
        v_j^{\mathrm{raw}}
        =
        -\beta+j\Delta_b^{\mathrm{raw}}
        :
        j=0,\ldots,L_b-1
    \right\}
\end{equation}
where
\begin{equation}
\label{eq: raw quantization step}
    \Delta_b^{\mathrm{raw}}
    \triangleq
    \frac{2\beta}{L_b-1}
    =
    \frac{\Delta_b}{2}.
\end{equation}
We apply the same adjacent-level stochastic rounding protocol as in \eqref{eq: stochastic rounding protocol} using this raw-score alphabet. The initialization and all quantizer outputs lie in $[-\beta,\beta]$, and the update in \eqref{eq: finite bit recursion theory} is a convex combination. Therefore, this interval is preserved throughout the recursion. The accumulated perturbation term corresponding to \eqref{eq: quantization proxy theory} is
\begin{align}
\label{eq: raw quantization proxy}
    V_{k,t}^{\mathrm{q,raw}}
    &\triangleq
    \frac{(\Delta_b^{\mathrm{raw}})^2}{4}
    \sum_{r=1}^{t}
    \sum_{j=1}^K
    ([A^r]_{jk})^2
    \nonumber\\
    &=
    \frac{\beta^2}{(L_b-1)^2}
    \sum_{r=1}^{t}
    \sum_{j=1}^K
    ([A^r]_{jk})^2
    \nonumber\\
    &=
    \frac14V_{k,t}^{\rm q}.
\end{align}
Let $x_{k,t}^{(b),\mathrm{raw}}$ denote the recursion \eqref{eq: finite bit recursion theory} initialized by
\begin{equation}
    x_{k,0}^{(b),\mathrm{raw}}
    =
    \widehat{\f}_k(\h_k^*)
\end{equation}
and using the raw quantization range above. Define
\begin{equation}
\label{eq: raw finite precision event}
    \M_{k,t}^{(b),\mathrm{raw}}
    \triangleq
    \{
        \bgamma^*
        x_{k,t}^{(b),\mathrm{raw}}
        \leq0
    \}.
\end{equation}
The conditional-MGF argument in Appendix~\ref{appendix: finite bit communication} is unchanged. Replacing $r_t(\delta)$ and $V_{k,t}^{\rm q}$ in~\eqref{eq: finite bit classification theory} respectively by $r_t^{\mathrm{raw}}(\delta)$ and $V_{k,t}^{\mathrm{q,raw}}$ therefore gives
\begin{align}
\nonumber
    &\;\P(
        \M_{k,t}^{(b),\mathrm{raw}}
        \given\Ttrain
    )
    \\
    &\leq
    \exp\Biggl\{
        -\frac{
            [r_t^{\mathrm{raw}}(\delta)]^2
        }
        {2\left[
            \tau_{\max}\beta^2
            \sum_{\ell=1}^K([A^t]_{\ell k})^2
            +
            V_{k,t}^{\mathrm{q,raw}}
        \right]}
    \Biggr\}
    \label{eq: raw finite precision conditional error}
\end{align}
almost surely on $\C_\delta^{\mathrm{raw}}$ whenever $r_t^{\mathrm{raw}}(\delta)>0$. Thus, relative to the centered result \eqref{eq: finite bit classification theory}, the raw-score counterpart changes only the residual finite-round margin $r_t^{\mathrm{raw}}(\delta)$ and the quantization proxy $V_{k,t}^{\mathrm{q,raw}}$. The unconditional guarantee follows by combining \eqref{eq: raw finite precision conditional error} with \eqref{eq: raw margin training probability}, as in Theorem~\ref{theorem: finite bit communication}.

\subsection{PAC-Style Generalization Counterpart}
\label{appendix: raw PAC generalization}

We next consider the PAC-style result under the aligned network training set~\eqref{eq: training set D for the network} and its local projections~\eqref{eq: local training set D_k}. The raw aggregate $f_{\ave}$ in~\eqref{eq: raw aggregate classifier} does not contain the sample-dependent empirical training mean that appears in the centered aggregate. Therefore, the fixed-class refinement technique described at the end of Appendix~\ref{appendix: PAC generalization bound} applies directly without the additional offset variable.

For the raw local and aggregate scores, respectively, define
\begin{align}
    P_e(f_k)
    &\triangleq
    \P\bigl(
        \sign(f_k(\h_k^*))\neq\bgamma^*
        \given\Ttrain
    \bigr),
    \\
    P_e(f_{\ave})
    &\triangleq
    \P\bigl(
        \sign(f_{\ave}(\h^*))\neq\bgamma^*
        \given\Ttrain
    \bigr).
\end{align}
We establish the analogous PAC-style bounds for $f_\ave$ and $f_k$ as follows.
\begin{corollary}[\textbf{Raw-score margin generalization bound}]
\label{corollary: raw PAC bound}
Under the sampling setup of Theorem~\ref{theorem: generalization bound for c_k and c_ave}, fix $\eta>0$ and $\epsilon\in(0,1)$. Then, with probability at least $1-\epsilon$, uniformly for all $f\in\F$,
\begin{equation}\label{eq: raw PAC bound}
    P_e(f_{\ave})
    \leq
    P_{e,\emp}^{\eta}(f_{\ave})
    +\frac{2\rho}{\eta}
    +\sqrt{\frac{\log(1/\epsilon)}{2N}}.
\end{equation}
Moreover, for every fixed agent $k$, with probability at least $1-\epsilon$, uniformly for all $f_k\in\F_k$,
\begin{equation}\label{eq: raw local PAC bound}
    P_e(f_k)
    \leq
    P_{e,\emp}^{\eta}(f_k)
    +\frac{2\rho_k}{\eta}
    +\sqrt{\frac{\log(1/\epsilon)}{2N}}.
\end{equation}
\end{corollary}

\begin{proof}
Following~\eqref{eq: one sided offset deviation appendix}, we define the one-sided deviation
\begin{equation}
    \Delta_+^{\mathrm{raw}}(\D)
    \triangleq
    \sup_{f\in\F}
    \Bigl\{
        \E
        \Phi_\eta(
            \bgamma f_{\ave}(\h)
        )
        -
        P_{e,\emp}^{\eta}(f_{\ave})
    \Bigr\}
\end{equation}
where $\Phi_\eta$ and $P_{e,\emp}^{\eta}$ are defined in \eqref{eq: definition of eta-margin loss function} and \eqref{eq: empirical eta-margin loss}, respectively. Since the function class consisting of $f_{\ave}$, $\forall f\in\F$ is fixed, the one-sided symmetrization and contraction argument used in the refinement of Appendix~\ref{appendix: PAC generalization bound} gives
\begin{equation}
    \E\Delta_+^{\mathrm{raw}}(\D)
    \leq
    \frac{2}{\eta}
    \sum_{k=1}^K\pi_k\rho_k
    =
    \frac{2\rho}{\eta}.
\end{equation}
Moreover, $0\leq\Phi_\eta\leq1$, so replacing one complete network example changes $\Delta_+^{\mathrm{raw}}(\D)$ by at most $1/N$. McDiarmid's inequality therefore gives, with probability at least $1-\epsilon$,
\begin{equation}
    \Delta_+^{\mathrm{raw}}(\D)
    \leq
    \frac{2\rho}{\eta}
    +
    \sqrt{\frac{\log(1/\epsilon)}{2N}}.
\end{equation}
Using~\eqref{eq: property of eta-margin loss} proves \eqref{eq: raw PAC bound}. Since this high-probability event is uniform over $f\in\F$, it may be evaluated at $f=\widehat{\f}$ for every realization of the local training randomness. Applying the same argument to the projected sample $\D_k$ gives \eqref{eq: raw local PAC bound}.
\end{proof}
\noindent Relative to the centered refinement in Appendix~\ref{appendix: PAC generalization bound}, the offset-complexity term is absent in Corollary~\ref{corollary: raw PAC bound} because the raw-score classifier does not contain an empirical training mean. As already characterized by~\eqref{eq: raw margin} and~\eqref{eq: raw centered margin comparison}, this difference does not imply a universal ordering of the resulting classifiers, since empirical centering also changes their margins relative to the zero decision threshold.

\subsection{Static-Observation Social Learning Counterpart}
\label{appendix: raw static social learning}

Replacing the centered decision statistic $\blambda_{k,t}^{\sf SL}$ in \eqref{eq: decision statistic at time t under SL} by the learned raw scores gives
\begin{equation}
\label{eq: raw SL statistic}
    \blambda_{k,t}^{\sf SL,raw}
    \triangleq
    \sum_{\ell=1}^K
    \sum_{i=1}^{t}
    [A^i]_{\ell k}
    \widehat{\f}_\ell(\h_\ell^*).
\end{equation}
We define
\begin{equation}
    \M_{k,t}^{\sf SL,raw}
    \triangleq
    \{
        \bgamma^*
        \blambda_{k,t}^{\sf SL,raw}
        \leq0
    \}.
\end{equation}
Using the cumulative mixing constant $\Bmix$ in \eqref{eq: cumulative mixing constant}, we denote
\begin{equation}
\label{eq: raw kappa}
    \kappa_{\mathrm{raw}}
    \triangleq
    \beta\Bmix
    =
    \frac{\kappa}{2}
\end{equation}
where $\kappa=2\beta\Bmix$ is defined in~\eqref{eq: kappa}. The proof of Proposition~\ref{prop: classification error using SL rule} carries over directly. The cumulative coordinate oscillations are unchanged, while $|\E^{(\gamma)}\widehat{\f}_\ell(\h_\ell)|\leq\beta$ reduces the cumulative mean-transient term in~\eqref{eq: SL classification mean appendix} from $\kappa$ to $\kappa_{\mathrm{raw}}$. Consequently, almost surely on $\C_\delta^{\mathrm{raw}}$, for every agent $k$ and every integer $t>\kappa_{\mathrm{raw}}/\delta$,
\begin{equation}
\label{eq: raw SL conditional bound}
    \P(
        \M_{k,t}^{\sf SL,raw}
        \given\Ttrain
    )
    \leq
    \exp\Biggl\{
        -\frac{
            (t\delta-\kappa_{\mathrm{raw}})^2
        }
        {2\tau_{\max}\beta^2
        \sum_{\ell=1}^K
        \left(
            \sum_{i=1}^{t}[A^i]_{\ell k}
        \right)^2}
    \Biggr\}.
\end{equation}
This is the direct raw-score counterpart of \eqref{eq: P_e for t-th round using SL rule}. Since
\begin{equation}
    \lim_{t\to\infty}\frac1t\sum_{i=1}^{t}[A^i]_{\ell k}
    =
    \pi_\ell,
\end{equation}
the bound~\eqref{eq: raw SL conditional bound} converges to the bound associated with raw scores under sufficient communication~\eqref{eq: raw sufficient conditional error}.

\section{Additional Experimental Details}
\label{app:exp_details}

This appendix provides additional implementation details and supplementary figures for the experiments conducted in Section~\ref{sec:simulations}. The presentation is grouped into the CIFAR-10 benchmark, the ModelNet40 benchmark, and the communication and robustness studies.

\subsection{CIFAR-10 Benchmark}
\label{app:cifar_details}

\subsubsection{Experimental Protocol}
\label{app:cifar_patch_train}

This subsection provides additional implementation details for the CIFAR-10 patch-partition benchmark introduced in Section~\ref{subsec:cifar_patch}. Each agent trains a convolutional neural network (CNN) independently on the patch associated with its grid location. The network consists of two convolutional layers followed by fully connected layers and produces a two-dimensional output $z_k(\h_k)\in\mathbb{R}^2$ for each input patch $\h_k$~\cite{hu2025non-asymptotic}. Its two entries, denoted by $z_{k,+1}(\h_k)$ and $z_{k,-1}(\h_k)$, are the logits associated with the two classes, which define the class-posterior probabilities through the softmax rule:
\begin{equation}\label{eq: posterior probability from logits}
    \widehat p_k(\gamma| \h_k)
=
\frac{\exp(z_{k,\gamma}(\h_k))}
{\exp(z_{k,+1}(\h_k))+\exp(z_{k,-1}(\h_k))}
\end{equation}
for $\gamma\in\{+1,-1\}$. Here, these logits refer to the raw outputs of the neural network, and should not be confused with the theoretical logit function in~\eqref{eq: logit function}. In the binary case, their difference satisfies
\begin{equation}\label{eq: logit gap}
    z_{k,+1}(\h_k)-z_{k,-1}(\h_k)
=
\log\frac{\widehat p_k(+1| \h_k)}{\widehat p_k(-1| \h_k)}.
\end{equation}
Therefore, in this implementation, the scalar quantity used for collaboration is first formed by taking the difference between the two logits. This logit difference is then centered according to~\eqref{eq: trained classifier} to obtain the local decision statistic $\c_k(\h_k)$.

For each value of $N_0$, we first sample a balanced set of CIFAR-10 images containing equal numbers of cats and dogs. The same selected images are then used across all agents, with each agent retaining only the patch corresponding to its own grid location. In this way, $N_0$ denotes the size of the local labeled dataset at each agent before the train/validation split. We reserve $20\%$ of each local dataset for validation and use the remaining $80\%$ for training. Figure~\ref{fig:cifar_local_datasets} illustrates this construction: each row corresponds to the local dataset of one agent, and each column shows the local patches extracted from a common image across agents.

\begin{figure}[t]
    \centering
    \includegraphics[width=0.95\linewidth]{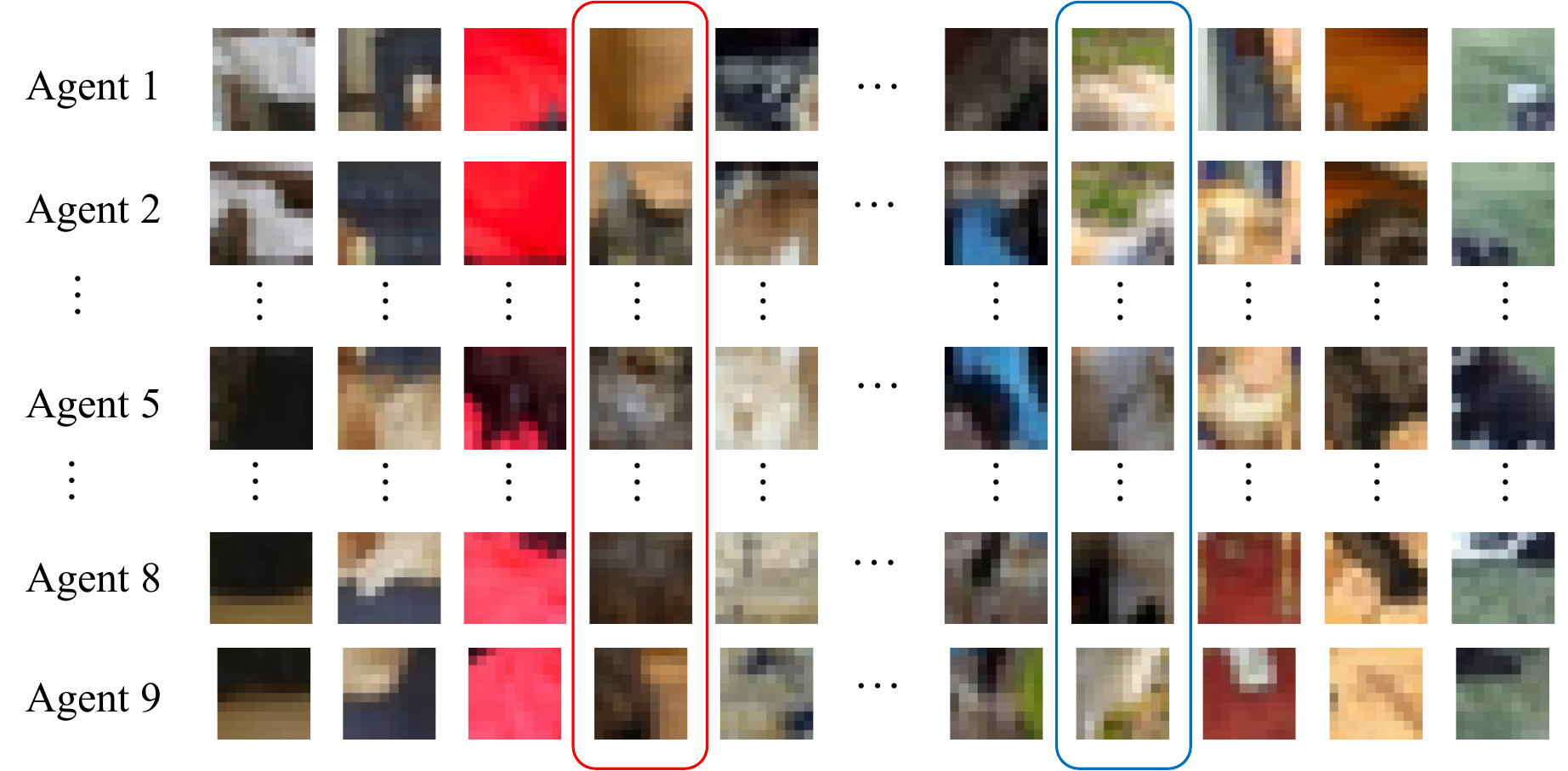}
    \caption{Illustration of the local labeled datasets in the CIFAR-10 patch-partition benchmark. The highlighted columns illustrate two representative images from the two classes contributing different local patches to all agents.}
    \label{fig:cifar_local_datasets}
\end{figure}

Each local CNN classifier is trained independently using the Adam optimizer with learning rate $10^{-4}$, mini-batch size $256$, and cross-entropy loss for at most $100$ epochs. Early stopping is applied in all CIFAR-10 experiments: training is monitored on the validation set and terminated once the validation loss ceases to improve over several consecutive validation checks, and the model parameters corresponding to the best validation loss are retained.

For each value of $N_0$, all reported quantities in Fig.~\ref{fig:cifar_patch_main} are averaged over $200$ Monte Carlo repetitions. In each repetition, the balanced image set is resampled, the train/validation split is regenerated, and all local classifiers are retrained from scratch. In the main baseline comparison, the communication graph is generated once and then kept fixed across repetitions. Each agent is also assumed to have a self-loop in the communication graph, which is omitted from Fig.~\ref{fig: network} for clarity. The corresponding combination policy $A$ is constructed using the uniform averaging rule:
\begin{equation}
    a_{\ell k}=\frac{1}{|\mathcal{N}_k|},\quad \forall \ell\in\mathcal{N}_k
\end{equation}
where $|\mathcal{N}_k|$ denotes the cardinality of $\mathcal{N}_k$.

\subsubsection{Baselines}
\label{app:cifar_baselines}

Figure~\ref{fig:cifar_patch_pe} compares the following methods.

\paragraph{Non-cooperative}
Each agent predicts independently from its own local decision statistic, without any communication or aggregation.

\paragraph{Avg-stat}
The local decision statistics are averaged uniformly across agents, and the final prediction is obtained by thresholding the resulting aggregate. This provides a simple averaging reference, and coincides with the proposed method when the combination policy $A$ is doubly-stochastic and the communication is sufficient.

\paragraph{Vote}
Each agent first generates a hard class prediction based on its local decision statistic. The final output is then determined by majority vote across the agents.

\paragraph{Learned fusion}
A centralized affine fusion rule is fitted on a held-out validation set. If $\bm{c}(h)\in\mathbb{R}^K$ denotes the vector of local decision statistics for sample $h$, the baseline predicts according to the sign of
\begin{equation}
    w^\top \bm{c}(h)+b_{\rm fus}
\end{equation}
where both the fusion weights $w\in\mathbb{R}^K$ and the bias $b_{\rm fus}\in\mathbb{R}$ are learned from validation data.

\paragraph{Learned simplex fusion}
This baseline also learns a centralized linear fusion rule from validation data, but under the structural constraints
\begin{equation}
    w_k\ge 0,\quad \sum_{k=1}^K w_k=1,
\end{equation}
with no additive bias term. It may therefore be interpreted as a learned convex combination of the local decision statistics.

\paragraph{AdaBoost}
A centralized boosting model is trained using the local predictors as constituent learners. Unlike the proposed method, this baseline introduces cooperation during training rather than only at test time.

\paragraph{VFL-JT}
To provide a task-matched joint training comparator for the partitioned CIFAR-10 data, we implement a score-level vertical federated joint-training scheme inspired by the feature-distributed learning framework in~\cite{hu2019fdml}. For an aligned example, branch $k$ receives the same local patch as the corresponding independently trained agent and outputs a scalar binary score $s_k$. The coordinating server forms
\begin{equation}
    u_{\rm VFL}
    =
    b_{\rm VFL}
    +
    \sum_{k=1}^{K} w_k s_k ,
\end{equation}
and the branch parameters together with the fusion head are optimized jointly using the global binary logistic loss. During training, the local scores are sent to the server and the corresponding loss derivatives with respect to the local scores are returned to the branches. Raw patches and local parameter tensors are not exchanged, and no parameter averaging is performed.

Because VFL-JT and the proposed method communicate at different stages, we report their training- and inference-time communication separately. Let $E_{\rm VFL}$ denote the realized number of VFL-JT training epochs, and $N_{\rm tr}$ and $N_{\rm val}$ the numbers of aligned training and validation samples, respectively. Recall that $E_{\rm off}(A)$ defined in~\eqref{eq: off diagonal communication links theory} denotes the number of directed non-self communication links. The resulting communication costs, measured by the number of transmissions, are summarized in Table~\ref{tab:vfl_ttc_comm}.

\begin{table}[t]
\centering
\caption{Communication cost of VFL-JT and the proposed scheme.}
\label{tab:vfl_ttc_comm}
\small
\setlength{\tabcolsep}{6pt}
\renewcommand{\arraystretch}{1.05}
\begin{tabular}{c|cc}
\hline
\textbf{Method}
&
\textbf{Training}
&
\textbf{Inference per sample}
\\
\hline
VFL-JT
&
$K E_{\rm VFL}(2N_{\rm tr}+N_{\rm val})$
&
$K$
\\
Proposed
&
$0$
&
$tE_{\rm off}(A)$
\\
\hline
\end{tabular}
\end{table}

For VFL-JT, each training sample involves the transmission of $K$ local scores to the server and the return of $K$ corresponding loss derivatives, while each validation sample requires the $K$ forward score transmissions. For the proposed method, communication occurs only during test-time collaboration over the directed non-self links of the network. If each communicated scalar is represented using $b$ bits, the corresponding bit costs are obtained by multiplying the entries in Table~\ref{tab:vfl_ttc_comm} by $b$. 

\paragraph{Central oracle}
A centralized classifier from the same CNN family is trained on the full $32\times32$ CIFAR-10 image. In this implementation, the first convolutional layer is adapted to the full-image input, and the first fully connected layer is resized accordingly to match the resulting feature dimension. This baseline is included only as a full-information reference and is not constrained by the distributed observation model.

\paragraph{Ours (no centering)}
This ablation implements the same collaboration rule as the proposed method, but removes the centering step in~\eqref{eq: trained classifier}. In other words, each agent uses the trained score $\widehat{\f}_k$ rather than the centered decision statistic $\c_k$.

\subsubsection{Conditional Mutual Information}
\label{app:cifar_cmi_details}

In Section~\ref{sec:cifar_patch_cmi}, the estimation of the mutual information $I(h_k;h_\ell\givensmall\gamma=y)$ appearing in~\eqref{eq:cifar_cmi} is carried out in two stages. For the reported results in Figs.~\ref{fig:cifar_patch_cmi_heatmap} and~\ref{fig:cifar_patch_cmi_bucket}, the conditional mutual information is estimated on the training split using all $10{,}000$ selected training samples. We also repeated the same analysis on the testing split and obtained the same qualitative conclusions. First, for each agent $k$, the raw patch vectors are stacked into a data matrix whose rows correspond to samples and whose columns correspond to pixel values, and principal component analysis (PCA) is then applied to this matrix~\cite{sayed2022inference}. A reduced representation of dimension $20$ is retained. The PCA transformation is fitted once per agent on the chosen split, and the resulting reduced coordinates are then partitioned according to the class label. Second, for each class $y\in\{-1,+1\}$, the mutual information between the two reduced patch representations is estimated by a $k$-nearest-neighbor (kNN) estimator with neighborhood parameter $k_{\rm NN}=5$~\cite{kraskov2004estimating}. The two class-conditional estimates are then combined according to~\eqref{eq:cifar_cmi}. All reported values are given in nats. We also evaluated PCA dimensions $10$, $20$, and $30$, together with neighborhood parameters $k_{\mathrm{NN}}=3$, $5$, and $10$. Across these choices, the same qualitative conclusions are found.

\subsubsection{Heterogeneous Models}
\label{app:cifar_hetero_details}

This subsection provides additional details for the heterogeneous version of the CIFAR-10 patch-partition benchmark in Section~\ref{sec:cifar_hetero}. We consider three local model families:
\begin{enumerate}
    \item \textbf{Family} \textsf{A}: the patch-level CNN adopted in the main CIFAR-10 experiments;
    \item \textbf{Family} \textsf{B}: logistic regression on the flattened raw patch pixels;
    \item \textbf{Family} \textsf{C}: a small residual convolutional network with adaptive pooling~\cite{he2016deep}.
\end{enumerate}
All other aspects of the experiment remain unchanged from the homogeneous CIFAR-10 benchmark, including the local datasets, the train/validation splitting protocol, the communication graph, the combination policy, and the training hyperparameters.

Because different patch locations are not equally informative, the placement of a model family on the $3\times3$ grid can affect its apparent performance. To reduce this spatial bias, we use a balanced assignment protocol based on the base layout:
\begin{equation}
\label{eq:app_hetero_base_assign}
{\renewcommand{\arraystretch}{1.5}
\mathbf{F}^{(1)}=
\begin{array}{|@{\hspace{0.5em}}c@{\hspace{0.5em}}|@{\hspace{0.5em}}c@{\hspace{0.5em}}|@{\hspace{0.5em}}c@{\hspace{0.5em}}|}
\hline
{\small \textsf{A}} & {\small \textsf{B}} & {\small \textsf{C}}\\
\hline
{\small \textsf{B}} & {\small \textsf{C}} & {\small \textsf{A}}\\
\hline
{\small \textsf{C}} & {\small \textsf{A}} & {\small \textsf{B}}\\
\hline
\end{array}}
\end{equation}
where $\textsf{A}$, $\textsf{B}$, and $\textsf{C}$ denote the three model families. A cyclic shift of this pattern yields two additional layouts:
\begin{equation}
\label{eq:app_hetero_assignments}
{\renewcommand{\arraystretch}{1.5}
\mathbf{F}^{(2)}=
\begin{array}{|@{\hspace{0.5em}}c@{\hspace{0.5em}}|@{\hspace{0.5em}}c@{\hspace{0.5em}}|@{\hspace{0.5em}}c@{\hspace{0.5em}}|}
\hline
{\small \textsf{B}} & {\small \textsf{C}} & {\small \textsf{A}}\\
\hline
{\small \textsf{C}} & {\small \textsf{A}} & {\small \textsf{B}}\\
\hline
{\small \textsf{A}} & {\small \textsf{B}} & {\small \textsf{C}}\\
\hline
\end{array}\;\;,
\quad
\mathbf{F}^{(3)}=
\begin{array}{|@{\hspace{0.5em}}c@{\hspace{0.5em}}|@{\hspace{0.5em}}c@{\hspace{0.5em}}|@{\hspace{0.5em}}c@{\hspace{0.5em}}|}
\hline
{\small \textsf{C}} & {\small \textsf{A}} & {\small \textsf{B}}\\
\hline
{\small \textsf{A}} & {\small \textsf{B}} & {\small \textsf{C}}\\
\hline
{\small \textsf{B}} & {\small \textsf{C}} & {\small \textsf{A}}\\
\hline
\end{array}\;.}
\end{equation}
In the study of collaboration gain at a fixed $N_0$ illustrated by Fig.~\ref{fig:cifar_patch_hetero}, the values are averaged over these three layouts so that each family appears equally often at each spatial location. In the $N_0$-sweep of Fig.~\ref{fig:cifar_patch_hetero_main}, the base layout $\mathbf{F}^{(1)}$ is kept fixed.

For each agent $k$, we quantify the performance gain from collaboration at round $t$ by
\begin{equation}
\Delta_k(t) \triangleq P_{k,0}-P_{k,t}
\label{eq:app_hetero_gain}
\end{equation}
where we recall that $P_{k,t}$ denotes the probability of error of agent $k$ at communication round $t$. For the heatmap shown in Fig.~\ref{fig:cifar_patch_hetero}, we report the corresponding \emph{relative} error reduction at the representative communication round $t=15$:
\begin{equation}
\Delta_k^{(\%)}(15)
\triangleq
100\cdot \frac{\Delta_k(15)}{\max\{P_{k,0},\epsilon_{\mathrm{safe}}\}}
\label{eq:app_hetero_gain_percent}
\end{equation}
where $\epsilon_{\mathrm{safe}}=10^{-6}$ is a small numerical safeguard.

\subsection{ModelNet40 Benchmark}
\label{app:modelnet_details}

The ModelNet40 experiments follow the same general pipeline as the CIFAR-10 benchmark in Appendix~\ref{app:cifar_details}, except for the differences described below.

The dataset is constructed from CAD models by rendering multiple views of each object using Blender. For each object, we generate $K=12$ RGB views from fixed camera positions whose azimuth angles are uniformly spaced around the object, with a common elevation angle of $30^\circ$. All rendered images are generated at resolution $224\times224$ with transparent background.

The local classifier is a lightweight convolutional network, referred to as \emph{MicroCNN}. It consists of a single convolutional layer with kernel size $7\times7$ and stride $4$, followed by a ReLU nonlinearity, global average pooling, an optional dropout layer, and a final linear layer producing two logits. As in the CIFAR-10 experiments, the local score is formed from the difference between two logits and then centered according to~\eqref{eq: trained classifier}.

For each value of $N_0$, we sample a balanced set of objects containing $N_0/2$ examples from each class and split it into training and validation subsets using the same $20\%$ validation fraction as in the CIFAR-10 study. The local models are trained independently using cross-entropy loss and validation-based early stopping, but with SGD optimizer and learning rate $0.05$. For the VFL-JT comparison, we use the same score-level joint-training construction described in Appendix~\ref{app:cifar_baselines}, with $K=12$ view-specific MicroCNN branches.

\paragraph{Centralized full-observation reference}
In the ModelNet40 benchmark, the central oracle has access to all $12$ views of the same object. The model consists of $K$ view-specific MicroCNN backbones, one for each rendered view, whose feature vectors are concatenated and passed to a centralized multilayer head for final prediction. The oracle is initialized from the independently trained local models and then trained in two stages. In the first stage, all backbone parameters are frozen and only the fusion head is optimized. In the second stage, all parameters are unfrozen and the full model is fine-tuned jointly using a smaller learning rate. Both training stages use cross-entropy loss, the Adam optimizer, and validation-based early stopping.

\subsubsection{Temperature Scaling}
\label{app:modelnet_temp_details}

For each agent $k$, we denote by $z_k(\h_k)\in\mathbb{R}^2$ the two-dimensional output of its local classifier for feature vector $\h_k$, whose entries are the logits associated with the two classes. A positive scalar temperature $T_k>0$ is fitted on the validation set by minimizing the cross-entropy loss of the rescaled logits:
\begin{equation}
\min_{T_k>0}\;
\frac{1}{N_{\rm val}}
\sum_{n=1}^{N_{\rm val}}
\ell_{\rm CE}\!\left(
\frac{z_k(\h_{k,n})}{T_k},\,\gamma_n
\right)
\label{eq:modelnet_temp_scaling}
\end{equation}
where $\ell_{\rm CE}$ denotes the cross-entropy loss, and $(\h_{k,n},\gamma_n)$ denotes the $n$-th labeled example in the validation set. In the implementation, we optimize over $\log T_k$ in order to enforce positivity. After fitting $T_k$, the classifier output is replaced by $z_k(\h_k)/T_k$. The local score used for collaboration is then formed by taking the difference between the two calibrated logits and is subsequently centered according to~\eqref{eq: trained classifier}.

\paragraph{Additional temperature scaling diagnostics}
\label{app:modelnet_temp_extra_figs}

Figure~\ref{fig:modelnet_temps} summarizes the fitted temperature parameters on the ModelNet40 benchmark across training set sizes $N_0$. For each value of $N_0$, we report the average of $T_k$ over the agents together with a $95\%$ confidence interval across $200$ Monte Carlo repetitions, as well as the average run-wise standard deviation of $\{T_k\}_{k=1}^K$. A learned temperature $T_k=1$ leaves the logits unchanged, $T_k>1$ softens them by reducing their magnitude, and $T_k<1$ sharpens them. The figure shows that the fitted temperatures are generally greater than one, and that both their average level and their variability across agents decrease as $N_0$ increases. This trend supports the conclusion that calibration is most critical in data-constrained regimes and becomes more stable as local classifier's reliability improves with additional training data.

Figure~\ref{fig:cifar_temps} reports the corresponding temperature diagnostics for the CIFAR-10 benchmark considered in Fig.~\ref{fig:cifar_patch_main}. Similar to ModelNet40, the fitted temperatures are larger and more variable at smaller training set sizes. As $N_0$ increases, however, they move closer to one, indicating that the logits require less rescaling in the higher-data regime. The corresponding results of decision margin and probability of error with temperature scaling are shown in Figs.~\ref{fig:cifar_patch_temp_margin} and~\ref{fig:cifar_patch_temp_pe}, respectively. Compared with Figs.~\ref{fig:cifar_patch_margin} and~\ref{fig:cifar_patch_pe}, temperature scaling tends to increase the achieved decision margin and slightly lowers the probability of error, with the effect being most visible at smaller $N_0$.

\begin{figure}[t]
    \centering
    \includegraphics[width=0.8\linewidth]{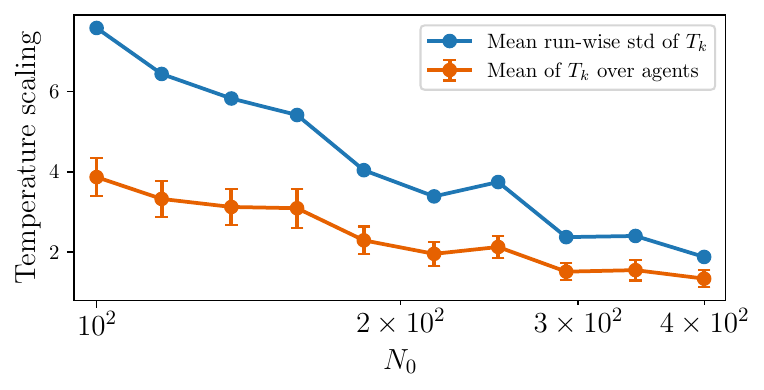}
    \caption{Temperature scaling diagnostics for the ModelNet40 benchmark. The orange curve shows the average of $T_k$ over agents, with error bars indicating $95\%$ confidence intervals across Monte Carlo repetitions. The blue curve shows the average run-wise standard deviation of $\{T_k\}_{k=1}^K$, which quantifies the variability of the fitted temperatures across agents.}
    \label{fig:modelnet_temps}
\end{figure}

\begin{figure}[t]
    \centering
    \includegraphics[width=0.8\linewidth]{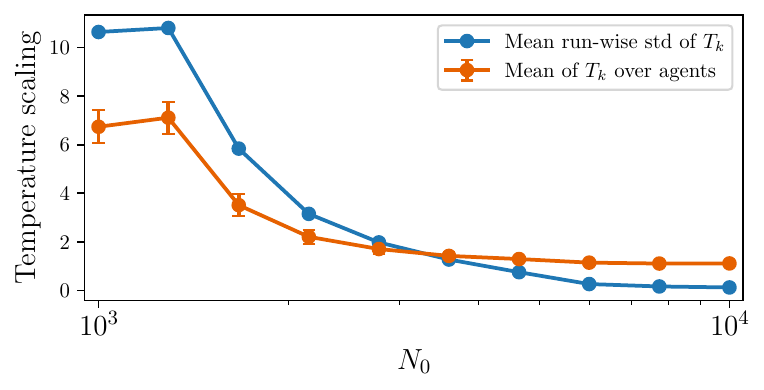}
    \caption{Temperature scaling diagnostics for the CIFAR-10 benchmark.}
    \label{fig:cifar_temps}
\end{figure}

\begin{figure}[t]
    \centering
    \subfloat[Decision margin]{
        \includegraphics[width=0.8\linewidth]
        {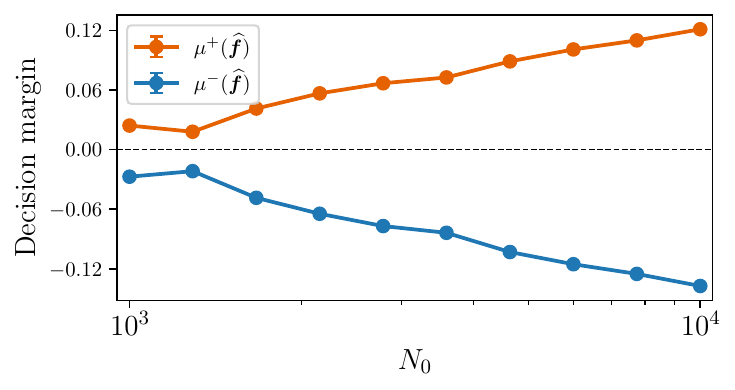}
        \label{fig:cifar_patch_temp_margin}}
    \hfill
    \subfloat[Probability of error]{
        \includegraphics[width=0.8\linewidth]
        {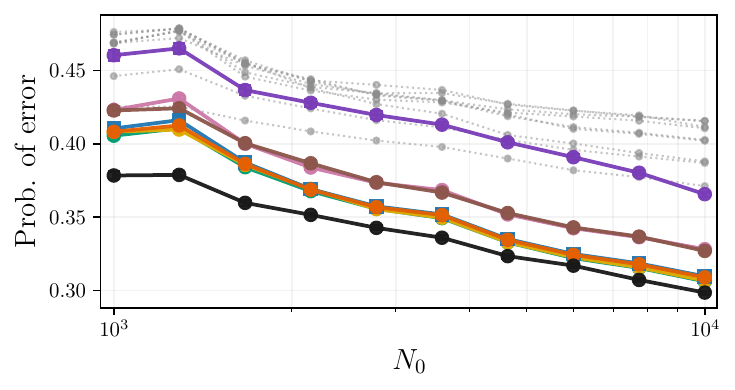}
        \label{fig:cifar_patch_temp_pe}}
    \caption{Performance of temperature scaling on CIFAR-10.}
    \label{fig:cifar_patch_temp_main}
\end{figure}

\subsubsection{Correlation Stress Test}
\label{app:modelnet_corr_details}

This subsection provides additional details for the controlled correlation stress test studied in Section~\ref{sec:corr_stress_test}. The perturbed decision statistics $\{x'_{i,k}\}$ are generated from the clean centered statistics $\{x_{i,k}\}$ using the perturbation model in~\eqref{eq:corr_stress} and \eqref{eq:corr_weights}. After perturbation, the decision statistics over the test set are collected into a matrix
\begin{equation}
    X\in\mathbb{R}^{N_{\mathrm{test}}\times K},
\end{equation}
whose $(i,k)$-th entry is the perturbed decision statistic of agent $k$ for the $i$-th testing sample, where $N_{\mathrm{test}}$ denotes the number of testing samples. To quantify the induced dependence, we compute the unconditional average pairwise Pearson correlation across agents:
\begin{equation}
\mathrm{corr}_{\rm uncond}(X)
\triangleq
\frac{2}{K(K-1)}
\sum_{1\leq k<\ell\leq K}
\widehat{r}_{k\ell}(X)
\end{equation}
where $\widehat{r}_{k\ell}(X)$ denotes the sample Pearson correlation between columns $k$ and $\ell$ of $X$. The label-conditional counterpart is obtained by restricting $X$ to each class separately and then averaging the corresponding pairwise correlations over the two classes.

In addition to these dependence measures, Fig.~\ref{fig:modelnet_corr_stress} reports the collaboration gain after $t$ communication rounds, measured by
\begin{equation}
P_0-P_t=\frac{1}{K}\sum_{k=1}^K (P_{k,0}-P_{k,t})
\end{equation}
where $P_0$ denotes the \emph{average} probability of error across agents in the non-cooperative setting and $P_t$ is the corresponding average probability of error after $t$ rounds of collaboration at the same perturbation level.


For completeness, Fig.~\ref{fig:modelnet_corr_error_overlay} reports the corresponding non-cooperative and collaborative probabilities of error as functions of the dependence strength $r$. For a fixed perturbation amplitude $\sigma_{\mathrm{pert}}$, varying $r$ changes the balance between the shared and idiosyncratic perturbation components while preserving the marginal perturbation distribution at each agent. This explains why the non-cooperative error remains approximately constant as $r$ varies. In contrast, the collaborative error generally increases with $r$, since a stronger shared component makes the perturbations across agents more correlated and reduces the benefit of combining their information. Consequently, the gap between the non-cooperative and collaborative curves narrows as $r$ increases, consistent with the reduction in collaboration gain observed in Fig.~\ref{fig:modelnet_corr_stress}. Increasing $\sigma_{\mathrm{pert}}$, on the other hand, increases the perturbation magnitude and raises the overall error levels. The clean reference curves show the corresponding performance without the injected perturbations.

\begin{figure}[t]
    \centering
    \includegraphics[width=0.95\linewidth]{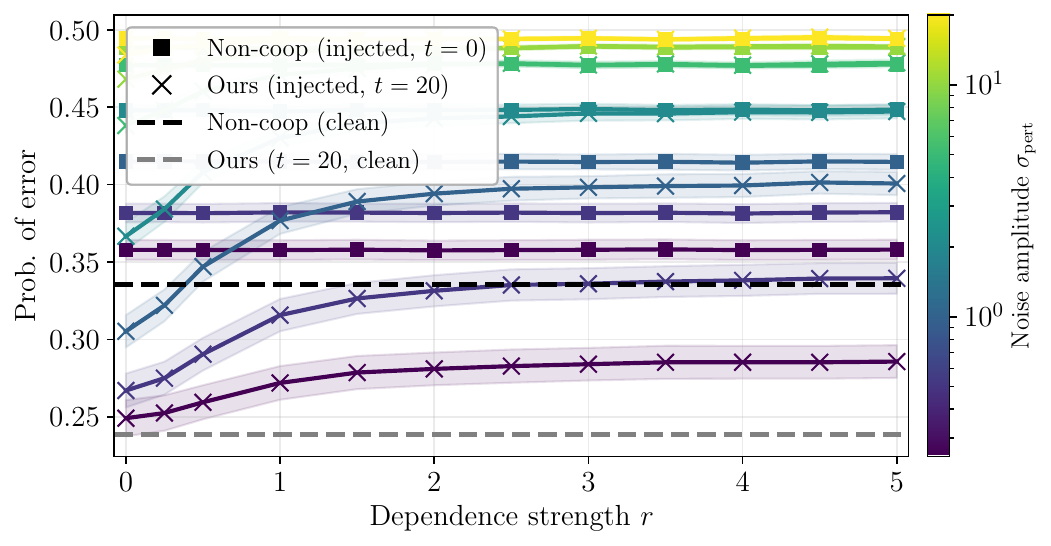}
    \caption{Probability of error under the controlled correlation stress test. For each perturbation amplitude $\sigma_{\mathrm{pert}}$, the figure compares the non-cooperative error at $t=0$ and the collaborative error after $t=20$ communication rounds as functions of the dependence strength $r$. The clean non-cooperative and collaborative error levels are also shown for reference.}
    \label{fig:modelnet_corr_error_overlay}
\end{figure}

\subsection{Communication and Robustness}

This subsection collects the implementation details for the communication design, adaptive stopping, and robustness studies in Section~\ref{subsec:comm_robust} of the main paper.

\subsubsection{Topology and Quantization}
\label{app:topology_details}

This subsection provides additional implementation details for the communication design study in Section~\ref{subsec:comm_robust}. Figure~\ref{fig:topology_examples_appendix} illustrates the ring and grid communication topologies used in the experiments.

\begin{figure}[t]
\centering
\subfloat[Ring]{\includegraphics[width=0.48\linewidth]{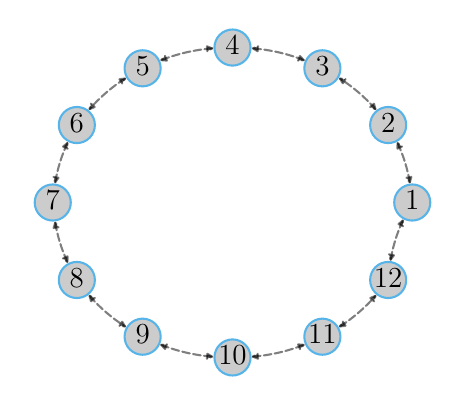}\label{fig:net_ring_appendix}}
\hfill
\subfloat[Grid]{\includegraphics[width=0.48\linewidth]{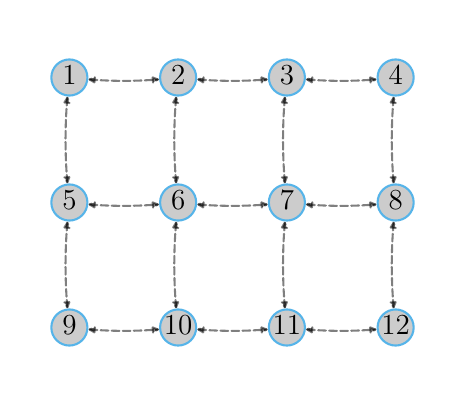}\label{fig:net_grid_appendix}}
\caption{Ring and grid communication topologies used in Fig.~\ref{fig:topology_sweep_main}. Self-loops are omitted from the plots for visual clarity.}
\label{fig:topology_examples_appendix}
\end{figure}

For each topology-rule pair, the graph and the corresponding combination policy $A$ are generated once and then kept fixed throughout the Monte Carlo study. In this study, we choose $N_0=400$, which corresponds to all available training samples for the selected class pair. The training and testing sets are therefore kept fixed across repetitions, while the train/validation split and the training randomness are varied. For each finite-bit condition, the results are additionally averaged over $30$ independent quantizer realizations within each of the $200$ training repetitions. In this way, the comparison isolates the effect of the communication design from variability due to changes in the underlying graph or data.

For the Metropolis rule used in the simulations, the communication weights are defined by
\begin{equation}\label{eq:app_metropolis_rule}
    a_{\ell k}=\begin{cases}
        \frac{1}{\max\{|\mathcal{N}_k|,|\mathcal{N}_\ell|\}}, &\text{if } \ell\neq k,\ell\in\mathcal{N}_k,\\
        1-\sum_{m\in \mathcal{N}_k\setminus\{k\}}a_{mk}, & \text{if } \ell =k,\\
        0, &\text{otherwise.}
    \end{cases}
\end{equation}
For the undirected ring and grid topologies, this is the standard Metropolis rule in the literature~\cite{sayed2014adaptation}, which ensures that the resulting combination policy $A$ is doubly-stochastic. In the case of the directed Erd\"{o}s--R\'{e}nyi topology, the inclusion of self-loops at all agents guarantees that $A$ is well defined and left-stochastic by construction~\eqref{eq:app_metropolis_rule}. However, because the graph is directed, the policy does not generally remain doubly-stochastic.

For the finite-precision study, we use the bounded posterior-probability difference $g_k(h)$ defined in~\eqref{eq: bounded score transform theory}. This statistic is obtained directly from the same trained classifier, and no retraining is performed. Its empirical training mean $\mu_{k,\emp}(g_k)$ is recomputed from the training statistics in this representation, so that the centered statistic lies in $[-2,2]$. For $b\in\{4,6,8\}$, we use $L_b=2^b$ uniformly spaced reconstruction levels on $[-2,2]$, with step size
\begin{equation}
    \Delta_b=\frac{4}{2^b-1}.
\end{equation}
The adjacent-level stochastic rounding protocol $\mathsf Q_b$ and the quantize-before-average recursion~\eqref{eq: finite bit recursion theory} are those defined in Section~\ref{sec: finite bit theory}. Quantization is applied at each sender before averaging, including the first exchange. For each testing sample and communication round, every sender draws one fresh quantization outcome independently across senders, rounds, and samples, and broadcasts the same quantized value to all of its outgoing neighbors.  Since both stochastic quantization and left-stochastic averaging preserve the interval $[-2,2]$, no overload clipping is required.

The matched full-precision reference uses the same trained classifiers, bounded scores, empirical centers, graph, and combination policy, but replaces the stochastic quantizer by the identity map. Thus, the finite-bit branches differ from the full-precision reference only through stochastic quantization. 

For the finite-bit branches, the per-round communication cost is measured in transmitted bits per testing sample:
\begin{equation}\label{eq: per-round communication cost}
B_{\rm round}=bE_{\rm off}(A)
\end{equation}
where $E_{\rm off}(A)$ is defined in~\eqref{eq: off diagonal communication links theory} and excludes self-loops from the communication count.

\paragraph{Additional discussion of Fig.~\ref{fig:topology_sweep_main}}
\label{app:topology_grouped_figs}

Figures~\ref{fig:topo_by_topology}--\ref{fig:topo_by_bits} provide alternative views of the results in Fig.~\ref{fig:topology_sweep_main}, organized according to graph family, the construction rule of $A$, and communication precision. These views help separate the role of each factor in collaborative classification. Across the graph families considered here, the uniform averaging rule and the Metropolis rule generally produce different combination matrices $A$. One exception is the ring graph, for which the two rules generate the same matrix $A$. Accordingly, Fig.~\ref{fig:topo_by_topology_ring} presents coincident trajectories at all tested precisions. For the grid graph, Fig.~\ref{fig:topo_by_topology_grid} shows that the uniform averaging rule performs better than the Metropolis rule, although the gap becomes small in the heavily quantized (4-bit) regime. For the tested directed Erd\"{o}s--R\'{e}nyi graph, Fig.~\ref{fig:topo_by_topology_er} shows that uniform averaging performs better than the Metropolis rule under the matched full-precision reference and all three finite-bit settings.

Figure~\ref{fig:topo_by_A} relates these performance differences to the spectral properties of the corresponding combination matrices. In our experiments, across both construction rules, the second-largest eigenvalue magnitude $\sigma_A$ is largest for the ring graph, intermediate for the grid, and smallest for the directed Erd\"{o}s--R\'{e}nyi graph. Since $\sigma_A$ governs the speed of linear mixing associated with the local decision statistics in~\eqref{eq: distributed learning rule}, this ordering helps explain the convergence behaviors observed in Figs.~\ref{fig:topo_by_A_uniform_col} and~\ref{fig:topo_by_A_metropolis_col}: the ring trajectories approach their limiting regime more slowly than those of the grid or Erd\"{o}s--R\'{e}nyi graphs. This spectral comparison concerns the transient mixing behavior. The classification performance also depends on the Perron vector of $A$ and the distribution of heterogeneous local scores.

The precision-specific views in Fig.~\ref{fig:topo_by_bits} further clarify the role of quantization. Since the ring and grid graphs are undirected, the Metropolis rule yields doubly-stochastic matrices for both graph families, and hence the same uniform Perron vector. Therefore, in the matched bounded full-precision setting, their decision statistics converge to the same limit. This is reflected in Fig.~\ref{fig:topo_by_bits_32bit}, where the corresponding curves approach nearly the same long-run error level. However, the ring graph exhibits a slower transient due to its larger $\sigma_A$. As the communication precision decreases, this coincidence is no longer exact. The finite-bit results exhibit a gradual precision hierarchy: the 8-bit trajectories remain close to full precision, 6-bit communication produces an intermediate degradation, and 4-bit communication produces the largest loss. Finally, although the ring graph converges more slowly in the early rounds, it attains a lower long-run error than the directed Erd\"{o}s--R\'{e}nyi graph in these experiments. In particular, the two graphs have the same number of communication links, and hence the same per-round communication cost at a fixed bit-width following~\eqref{eq: per-round communication cost}. This comparison further illustrates the importance of graph topology in collaborative inference.

\begin{figure*}[t]
    \centering
    \subfloat[Ring topology]{\includegraphics[width=0.32\linewidth]{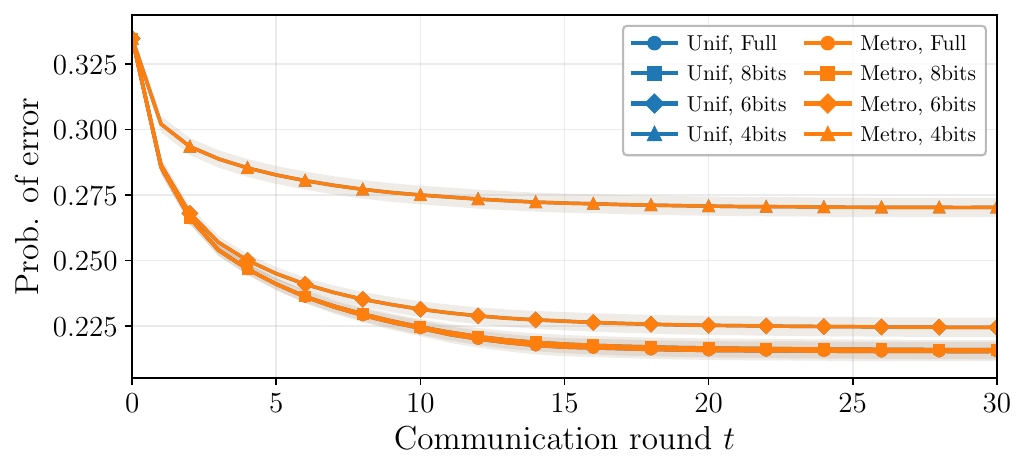}\label{fig:topo_by_topology_ring}}
    \hfill
    \subfloat[Grid topology]{\includegraphics[width=0.32\linewidth]{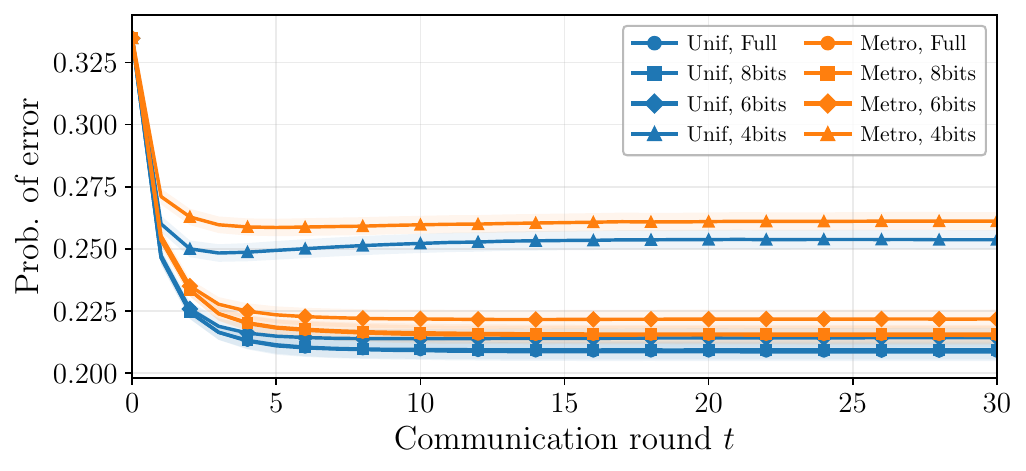}\label{fig:topo_by_topology_grid}}
    \hfill
    \subfloat[Directed Erd\"{o}s--R\'enyi topology]{\includegraphics[width=0.32\linewidth]{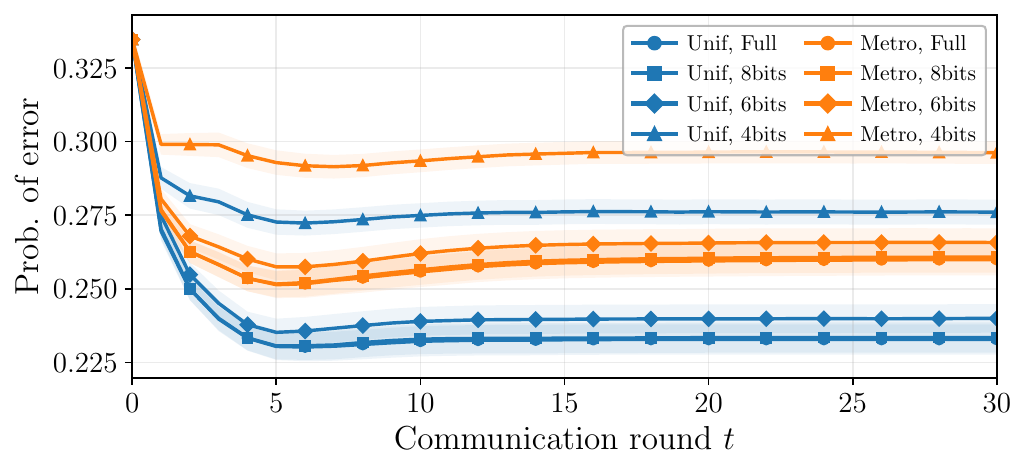}\label{fig:topo_by_topology_er}}
    \caption{Effect of the construction rule of $A$ and quantization on classification performance under different graph topologies.}
    \label{fig:topo_by_topology}
\end{figure*}

\begin{figure*}[t]
    \centering
    \subfloat[Uniform averaging rule]{\includegraphics[width=0.45\linewidth]{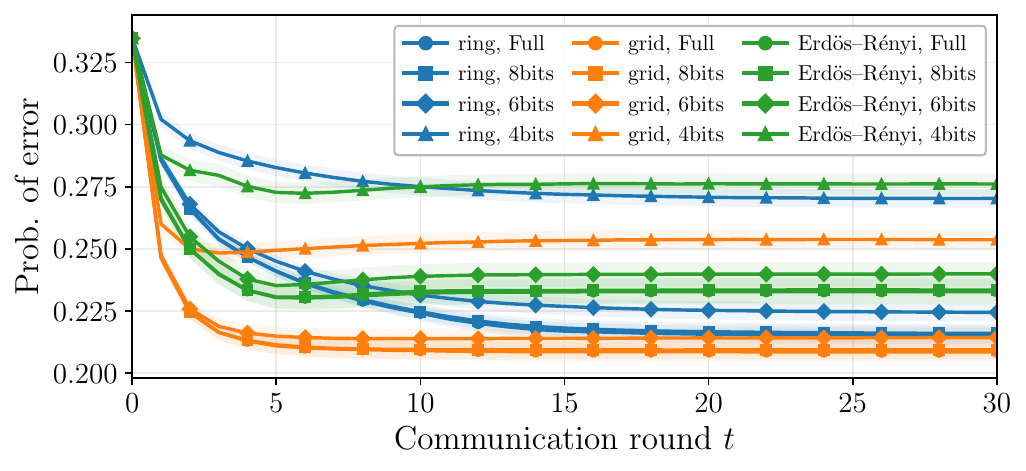}\label{fig:topo_by_A_uniform_col}}
    \hfill
    \subfloat[Metropolis rule]{\includegraphics[width=0.45\linewidth]{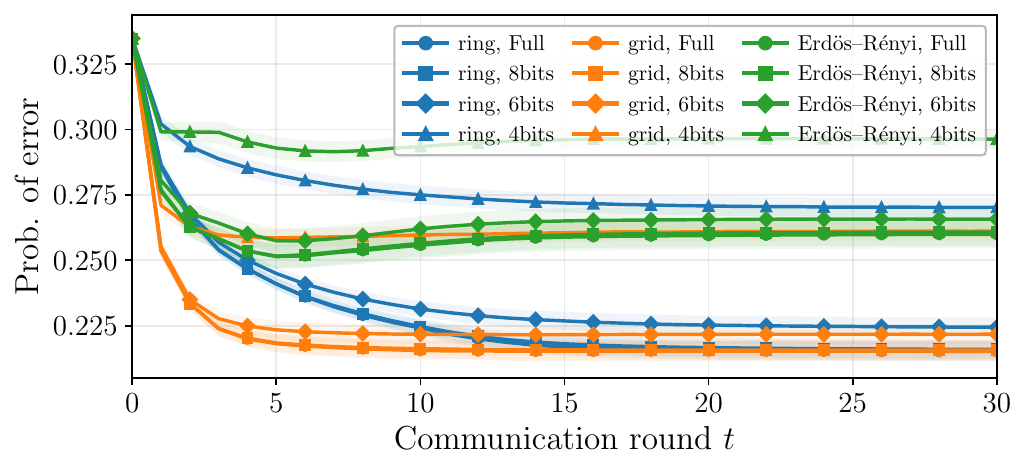}\label{fig:topo_by_A_metropolis_col}}
    \caption{Effect of graph topology and quantization on classification performance under different construction rules of $A$.}
    \label{fig:topo_by_A}
\end{figure*}

\begin{figure*}[t]
    \centering
    \subfloat[Full-precision communication]{\includegraphics[width=0.45\linewidth]{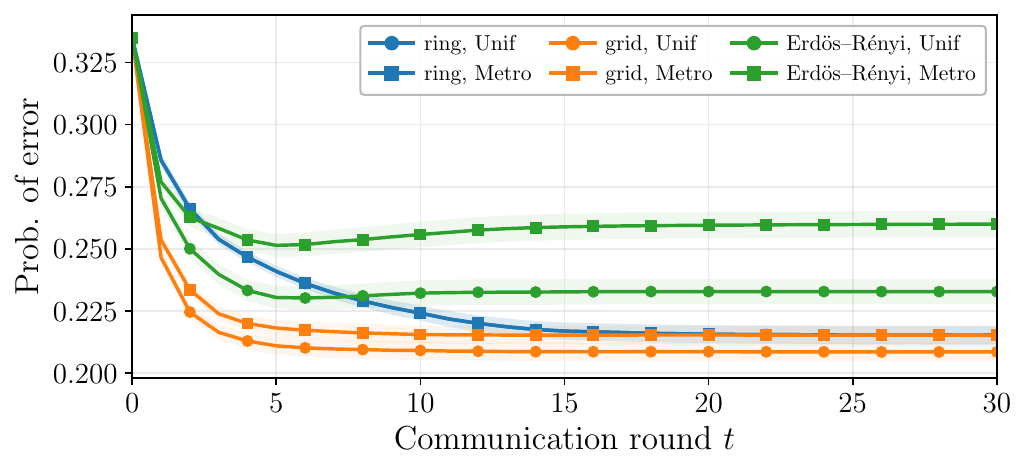}\label{fig:topo_by_bits_32bit}}
    \hfill
    \subfloat[8-bit communication]{\includegraphics[width=0.45\linewidth]{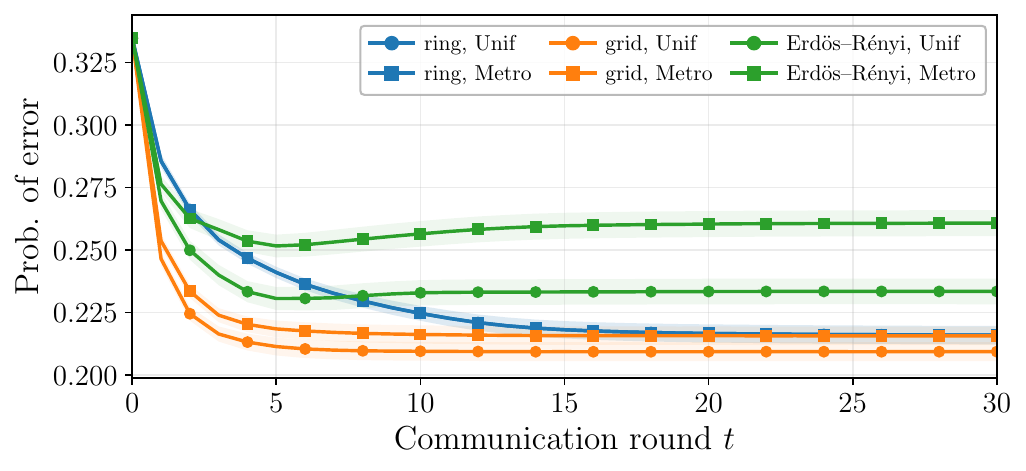}\label{fig:topo_by_bits_8bit}}\\
    \subfloat[6-bit communication]{\includegraphics[width=0.45\linewidth]{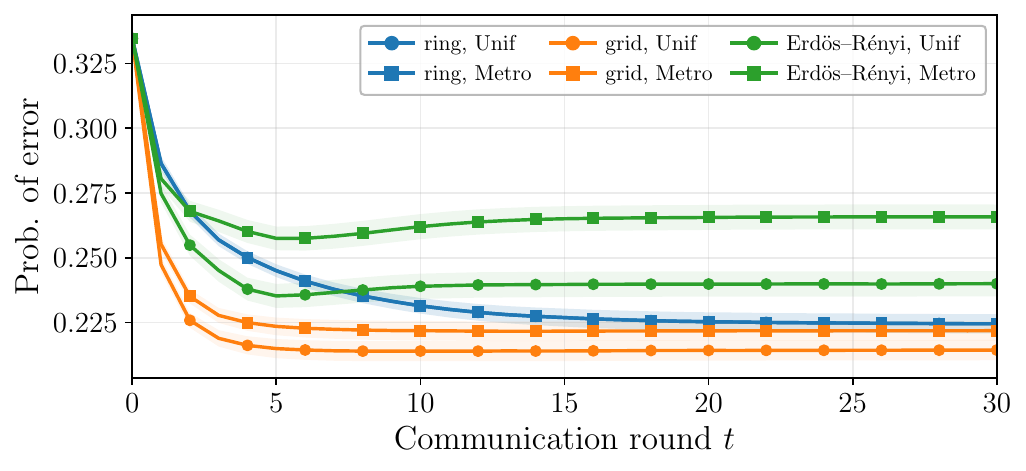}\label{fig:topo_by_bits_6bit}}
    \hfill
    \subfloat[4-bit communication]{\includegraphics[width=0.45\linewidth]{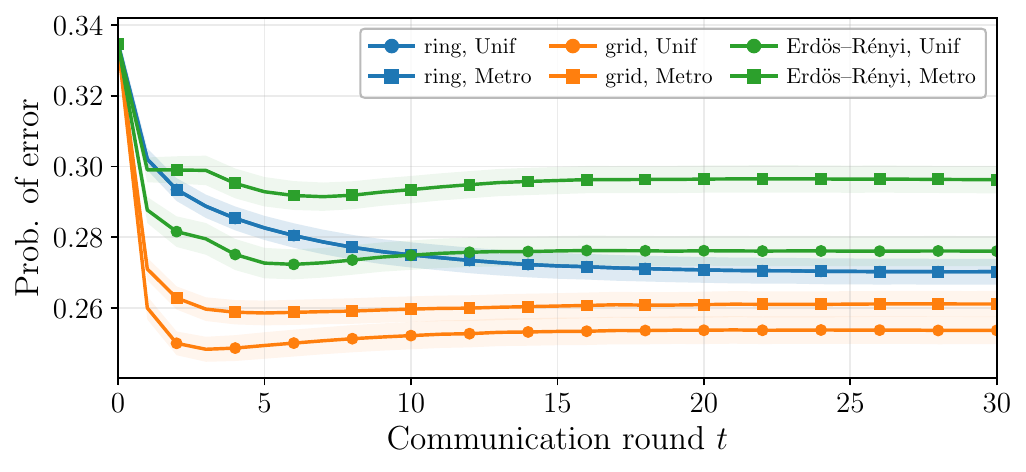}\label{fig:topo_by_bits_4bit}}
    \caption{Effect of graph topology and the construction rule of $A$ under different communication precisions.}
    \label{fig:topo_by_bits}
\end{figure*}

\subsubsection{Adaptive Stopping}
\label{app:stopping_details}

This subsection provides additional details for the adaptive stopping study in Section~\ref{sec:stopping_rule}. Our implementation follows a local event-triggered communication protocol. For the $i$-th testing sample and each agent $k$, let $x_{i,k}^{\mathrm{sent}}(t)$ denote the last value that agent $k$ has transmitted up to round $t$. At the next round, agent $k$ computes a candidate statistic given by
\begin{equation}
x_{i,k}^{\mathrm{new}}(t)
=
\sum_{\ell\in\mathcal{N}_k}
a_{\ell k}\,
x_{i,\ell}^{\mathrm{sent}}(t-1).
\end{equation}
At initialization, we set $x_{i,k}^{\mathrm{sent}}(0)=x_{i,k}(0)$, where $x_{i,k}(0)$ is the local decision statistic before communication. The stopping decision is then made from $x_{i,k}^{\mathrm{new}}(t)$ together with the agent's local memory. If the trigger fires, the stored communicated value is updated; otherwise the previous transmitted value is kept. Therefore, an agent that remains silent at one round still computes a fresh candidate statistic and may become active again later if the trigger condition is violated.

We investigate two local stopping-rule families. The first is a \emph{$\Delta$-trigger} rule, which is based directly on the change in the communicated statistic. Under this rule, agent $k$ transmits the candidate statistic $x_{i,k}^{\mathrm{new}}(t)$ to its neighbors at round $t$ only if
\begin{equation}
\left|
x_{i,k}^{\mathrm{new}}(t)-x_{i,k}^{\mathrm{sent}}(t-1)
\right|>\varepsilon_{\mathrm{send}}
\label{eq:app_delta_stop_rule}
\end{equation}
where $\varepsilon_{\mathrm{send}}>0$ denotes the threshold parameter. The second is a \emph{label-stability with confidence} rule. Let
\begin{equation}
y_{i,k}^{\mathrm{new}}(t)
=
\mathrm{sign}\bigl(x_{i,k}^{\mathrm{new}}(t)\bigr)
\end{equation}
be the label induced by the candidate statistic, and let $s_{i,k}(t)$ denote the number of consecutive rounds up to round $t$ over which this induced label has remained unchanged. In the implementation, the counter is updated by comparing the induced labels at two consecutive rounds. Thus, $L=1$ corresponds to one unchanged transition, $L=2$ to two consecutive unchanged transitions, and so on. Agent $k$ does not transmit sample $i$ at round $t$ when
\begin{equation}
s_{i,k}(t)\geq L,
\quad
|x_{i,k}^{\mathrm{new}}(t)|\geq \tau_{\mathrm{conf}}
\label{eq:app_label_stop_rule}
\end{equation}
where $L$ is a patience parameter and $\tau_{\mathrm{conf}}$ is a confidence threshold.

To evaluate the performance of the two stopping rules, we sweep $\varepsilon_{\mathrm{send}}$ over a dense grid for the $\Delta$-trigger rule, and sweep $(L,\tau_{\mathrm{conf}})$ over dense grids for the label-stability rule. Since transmissions may stop at different times for different agents and different samples, the communication cost for each rule is measured by the total number of transmissions per testing sample, obtained by counting the realized transmissions over directed non-self links across the communication rounds. The Pareto-efficient operating points shown in Fig.~\ref{fig:stopping_frontier} are then extracted in the plane of probability of error versus the total number of transmissions.

\subsubsection{Robustness}
\label{app:robust_details}

This subsection provides additional details for the robustness study in Section~\ref{sec:robustness}. For the $i$-th testing sample, let $x_{i,k}(0)$ denote the local decision statistic of agent $k$ before communication. The collaborative recursion then evolves according to
\begin{equation}
x_{i,k}(t)
=
\sum_{\ell\in\mathcal{N}_k}
a_{\ell k}x_{i,\ell}(t-1).
\label{eq:app_degroot_robust}
\end{equation}
For each perturbation scenario, we replace the initial statistics $\{x_{i,k}(0)\}$ by corrupted values $\{x_{i,k}^{\rm corr}(0)\}$ and compare the resulting performance with the corresponding clean trajectory.

In each repetition, $m\in\{1,\ldots,K\}$ agents are corrupted, chosen uniformly without replacement, and kept fixed throughout that repetition. We consider three perturbation families.

\paragraph{Benign noisy agents}
For each corrupted agent $k$, independent Gaussian perturbations are injected into the initial statistics:
\begin{equation}
x_{i,k}^{\rm corr}(0)
=
x_{i,k}(0)+\alpha_{\rm noise}\zeta_{i,k},
\quad
\zeta_{i,k}\sim\mathcal{N}(0,1).
\end{equation}
To normalize the perturbation level $\alpha_{\rm noise}$, we first compute the sample standard deviation of $\{x_{i,k}(0)\}_{i=1}^{N_{\rm test}}$ for each agent $k$, and then take the median of these values across agents. Denoting the resulting scale estimate by $\widehat{s}$, we set
\begin{equation}
\alpha_{\rm noise}
=
\eta_{\rm noise}\widehat{s},
\end{equation}
with a dimensionless multiplier $\eta_{\rm noise}$.

\paragraph{Faulty agents}
We consider both stuck-at-zero failures and constant-bias failures. In the stuck-at-zero model,
\begin{equation}
x_{i,k}^{\rm corr}(0)=0,
\end{equation}
whereas in the constant-bias model,
\begin{equation}
x_{i,k}^{\rm corr}(0)
=
x_{i,k}(0)+b_{\rm fault},
\quad
b_{\rm fault}
=
\eta_{\rm bias}\widehat{s}.
\end{equation}
Here, $\eta_{\rm bias}$ is the dimensionless bias multiplier varied in the experiments.

\paragraph{Adversarial reports}
For each corrupted agent $k$, we apply a scaled sign flip,
\begin{equation}
x_{i,k}^{\rm corr}(0)
=
-\eta_{\rm adv}x_{i,k}(0),
\quad
\eta_{\rm adv}\geq0.
\end{equation}
This model includes several special cases: $\eta_{\rm adv}=0$ corresponds to silencing, $\eta_{\rm adv}=1$ to a pure sign flip, and $\eta_{\rm adv}>1$ to an amplified wrong-sign report.

\bibliographystyle{IEEEtran}
\bibliography{refs}